\documentclass{CSML}

\def\dOi{12(3:8)2016}
\lmcsheading%
{\dOi}
{1--59}
{}
{}
{Oct.~30, 2015}
{Sep.~\phantom06, 2016}
{}

\ACMCCS{[{\bf Theory of computation}]: Computational complexity and
  cryptography---Complexity theory and logic; Logic}

\newcommand{\amp}{_{\it amplitude}}
\newcommand{\spa}{_{\it space}}
\newcommand{\tim}{_{\it time}}

\newcommand{\br}{\mbox{\it Br}}
\newcommand{\mcarry}{\mbox{\it Carry}}
\newcommand{\bitsum}{\mbox{\it Bitsum}}
\newcommand{\adcarry}{\mbox{\it Carry1}}
\newcommand{\subcarry}{\mbox{\it Borrow1}}
\newcommand{\modulo}{\mbox{\hspace {3pt}mod\hspace{3pt}}}

\newcommand{\cltw}{\mbox{\bf CL12}}

\newcommand{\areleven}{\mbox{\bf CLA11}} 
\newcommand{\pa}{\mbox{\bf PA}}

\newcommand{\thr}{\areleven_{\mathcal A}^{\mathcal R}}
\newcommand{\rep}{\mbox{\raisebox{2pt}{\scriptsize $|\hspace{-4pt}\sim$}}}

 \newcommand{\blank}{\mbox{{\scriptsize {\sc Blank}}}}

 \newcommand{\bit}{\mbox{\it Bit}}

\newcommand{\intimpl}{\mbox{\hspace{2pt}$\circ$\hspace{-0.14cm} \raisebox{-0.043cm}{\Large --}\hspace{2pt}}}

\newcommand{\successor}{\mbox{\hspace{1pt}\boldmath $'$}}

\newcommand{\elz}[1]{\mbox{$\parallel\hspace{-3pt} #1 \hspace{-3pt}\parallel$}}

\newcommand{\oo}{\bot}            
\newcommand{\pp}{\top}            
\newcommand{\xx}{\wp}

\newcommand{\seq}[1]{\langle #1 \rangle}

\newcommand{\gneg}{\mbox{\small $\neg$}}                  
\newcommand{\mli}{\hspace{2pt}\mbox{\small $\rightarrow$}\hspace{2pt}}                      
\newcommand{\cla}{\mbox{$\forall$}}      
\newcommand{\cle}{\mbox{$\exists$}}        
\newcommand{\mld}{\hspace{2pt}\mbox{\small $\vee$}\hspace{2pt}}     
\newcommand{\mlc}{\hspace{2pt}\mbox{\small $\wedge$}\hspace{2pt}}   
\newcommand{\ade}{\mbox{\large $\sqcup$}}      
\newcommand{\ada}{\mbox{\large $\sqcap$}}      
\newcommand{\add}{\hspace{2pt}\mbox{\small $\sqcup$}\hspace{2pt}}                     
\newcommand{\adc}{\hspace{2pt}\mbox{\small $\sqcap$}\hspace{2pt}} 
\newcommand{\tlg}{\bot}               
\newcommand{\twg}{\top}               

\theoremstyle{italic}
\theoremstyle{italic}\newtheorem{lemma}[thm]{Lemma}
\theoremstyle{plain}
\theoremstyle{plain}\newtheorem{remark}[thm]{Remark}
\theoremstyle{plain}

\usepackage{amsfonts} 
\usepackage{hyperref}
\theoremstyle{plain} 

\begin{document}

\title[Build your own clarithmetic I]{Build your own clarithmetic I: Setup and completeness}

\author[G.~Japaridze]{Giorgi Japaridze}
\address{Department of Computing Sciences, Villanova University, 800 Lancaster Avenue, Villanova, PA 19085, USA}	
\urladdr{http://www.csc.villanova.edu/$\sim$japaridz/}  
\email{giorgi.japaridze@villanova.edu}  


\begin{abstract}
  \noindent {\em Clarithmetics} are number theories based on {\em computability logic}. Formulas of these 
theories represent interactive computational problems, and their ``truth'' is understood as existence of an algorithmic
solution. Various complexity constraints on such solutions induce various versions of clarithmetic. The
present paper introduces a parameterized/schematic version $\mbox{\bf CLA11}_{P_4}^{P_1,P_2,P_3}$. By  tuning the three parameters $P_1,P_2,P_3$ in an essentially mechanical manner, one automatically obtains sound and complete theories with respect to a wide range of target {\em tricomplexity}
classes, i.e., combinations of time (set by $P_3$), space (set by $P_2$) and so called amplitude (set by $P_1$) complexities. Sound in the sense that
every theorem $T$ of the system represents an interactive number-theoretic computational problem with
a solution from the given tricomplexity class and, furthermore, such a solution can be automatically extracted from a
proof of $T$. And complete in the sense that every interactive number-theoretic problem with a solution
from the given tricomplexity class is represented by some theorem of the system. Furthermore, through tuning the 4th parameter $P_4$, at the cost of sacrificing recursive axiomatizability but not simplicity or elegance, 
the above {\em extensional
completeness} can be strengthened to {\em intensional completeness}, according to which every formula 
representing a problem with a solution from the given tricomplexity class is a theorem of the system. This article
is published in two parts. The present Part I introduces the system and proves its completeness, while the
forthcoming Part II is devoted to proving soundness.
\end{abstract}

\maketitle

\tableofcontents

\newpage
\section{Introduction}\label{intr}

\subsection{Computability logic}\label{scmptlg}
{\em Computability logic} (CoL for short),\label{x1} together with its accompanying proof theory termed {\em cirquent calculus},\label{0crq} has evolved in recent years in a long series of publications 
\cite{bauerTOCL}-\cite{bauerLMCS}, \cite{Jap03}-\cite{cla5}, \cite{Ver}, \cite{Qu}, \cite{xuIGPL}-\cite{xu2016}.
 It is a mathematical platform and long-term program for rebuilding logic as
a formal theory of computability, as opposed to the more traditional role of logic as a formal theory of
truth. Under CoL's approach, logical operators stand for operations on computational problems, formulas
represent such problems, and their ``truth'' is seen as algorithmic solvability. In turn, computational problems
--- understood in their most general, {\em interactive} sense --- are defined as games played by a machine against
its environment, with ``algorithmic solvability'' meaning existence of a machine that wins the game against
any possible behavior of the environment. With this semantics, CoL provides a systematic answer to the
question ``{\em what can be computed?}'', just like classical logic is a systematic tool for telling what
is true. Furthermore, as it happens, in positive cases ``{\em what} can be computed'' always allows itself to be
replaced by ``{\em how} can be computed'', which makes CoL of potential interest in not only theoretical computer
science, but many applied areas as well, including interactive knowledge base systems, resource oriented
systems for planning and action, or declarative programming languages. 

Both syntactically and semantically, CoL is a conservative extension of classical first order logic. Classical sentences and predicates are seen in it as special, simplest cases of computational problems --- specifically, as games  with no moves, automatically won by the machine when true and lost when false. Such games are termed {\bf elementary}.\label{x4} All operators of classical logic remain present in the language of CoL, with their semantics  generalized from elementary games to all games. Namely: $\neg A$ is $A$ with the roles of the two players interchanged. 
 $A\mlc B$ is a game where both $A$ and $B$ are played in parallel, and where the machine wins if it wins in both 
components. $A\mld B$ is similar, with the difference that here winning in just one component is sufficient.  $A\mli B$ is understood as $\gneg A\mld B$, playing which, intuitively, means reducing $B$ to $A$. $\cla x A(x)$ is a game winning which means playing $A(x)$ in a uniform, $x$-independent way  so that a win for all possible values of $x$ is guaranteed. $\cle xA(x)$ is similar, only here existence of just one lucky value is sufficient. These operators are conservative generalizations of their classical counterparts in the sense that  the meanings of the former happen to coincide with the meanings of the latter when the operators are restricted to elementary games only.    

In addition to $\neg,\mlc,\mld,\mli,\cla,\cle$, there is a host of  ``non-classical'' connectives and quantifiers. Out of those, the present paper only deals with the so called {\em choice}\label{x5}  group of operators: $\adc,\add,\ada,\ade$,\label{x6a} referred to as choice (``constructive'') conjunction, disjunction, universal quantifier and existential quantifier, respectively. 
$A\adc B$ is a game where the environment chooses between $A$ and $B$, after which the play continues according to the rules of the chosen component. $A\add B$ is similar, only here 
the choice is made by the machine.    
In $\ada x A(x)$, the environment chooses a value $n$ for $x$, and the play continues as $A(n)$. In the dual $\ade x A(x)$, such a choice is made by the machine. 

The language of CoL allows us to specify an infinite variety of meaningful computational problems and relations between them in a systematic way.  
Here  are some examples, where $f$ is a unary function, $p,q$ are unary predicates, and $A\leftrightarrow B$ abbreviates $(A\mli B)\mlc(B\mli A)$.  
$\ada x\bigl(p(x)\add \neg p(x)\bigr)$ expresses the problem of deciding $p$. Indeed, this is a game where, at the beginning, the environment selects a value $n$ for $x$. In traditional terms, this event can be viewed as providing $n$ as an ``input''. The game then continues as $p(n)\add \neg p(n)$ and, in order to win, the machine has to choose the true $\add$-disjunct. So,  $p$ is decidable if and only if the machine has an algorithmic winning strategy in  $\ada x\bigl(p(x)\add \neg p(x)\bigr)$. Quite similarly, $\ada x\ade y\bigl(y=f(x)\bigr)$ can be seen to be the problem of computing $f$. Next, 
$\ada x\ade y\bigl(p(x)\leftrightarrow q(y)\bigr)$ is the problem of many-one reducing $p$ to $q$. If we want to specifically say that $f$ is a such a  reduction, then $\ada x\ade y\bigl(y=f(x)\mlc (p(x)\leftrightarrow q(y))\bigr)$ can be written. If we additionally want to indicate that here $f$ is in fact one-one reduction, we can write $\ada x\ade y\bigl(y=f(x)\mlc (p(x)\leftrightarrow q(y))\mlc\cla z(y=f(z)\mli z=x)\bigr)$.
Bounded Turing reduction of $p$ to $q$ takes the form 
\[\ada y_1\bigl(q(y_1)\add \neg q(y_1)\bigr)\mlc\ldots\mlc \ada y_n\bigl(q(y_n)\add \neg q(y_n)\bigr)\mli \ada x\bigl(p(x)\add \neg p(x)\bigr).\]
If, instead, we write
\[\ada y_1\ldots\ada y_n\Bigl(\bigl(q(y_1)\add \neg q(y_1)\bigr)\mlc\ldots\mlc \bigl(q(y_n)\add \neg q(y_n)\bigr)\Bigr)\mli \ada x\bigl(p(x)\add \neg p(x)\bigr),\]
then bounded weak truth-table reduction is generated. And so on. In all such cases, imposing various complexity constraints on the allowable computations, as will be done in the present paper,  yields the corresponding complexity-theoretic counterpart of the concept. For instance, if  computations are required to run in polynomial time, then $\ada x\ade y\bigl(p(x)\leftrightarrow q(y)\bigr)$ becomes polynomial time many-one reduction, more commonly referred to as simply ``polynomial time reduction''.

Lorenzen's \cite{Lor59}, Hintikka's \cite{Hintikka73} and Blass's \cite{Bla72,Bla92} dialogue/game semantics should be named as the most direct precursors  
of CoL. The presence of  close connections with intuitionistic logic \cite{Propint} and Girard's \cite{Gir87} linear logic at the level of syntax and overall philosophy is also a fact.  
A rather comprehensive and readable, tutorial-style introduction to CoL can be found in the first 10 sections of \cite{Japfin}, which
is the most recommended reading for a first acquaintance with the subject.
 A more compact yet self-contained introduction to the fragment of CoL relevant to the
present paper is given in \cite{cl12}. 

\subsection{Clarithmetic}
Steps towards claiming specific application areas for CoL have already been
made in the direction of basing applied theories --- namely, {\em Peano arithmetic} $\pa$\label{x2} --- on CoL instead of
the traditional, well established and little challenged  alternatives
such as classical or intuitionistic logics. Formal arithmetical systems based on CoL have been baptized
in \cite{cla4} as {\bf clarithmetics}.\label{x3} By now ten clarithmetical theories, named {\bf CLA1} through {\bf CLA10}, have been
introduced and studied \cite{Japtowards,cla4,cla8,cla5}. These theories are notably simple: most
of them happen to be conservative extensions of $\pa$ whose only non-classical  axiom is the sentence $\ada x\ade y(y = x\successor)$   
asserting computability of the successor function $\successor$,\label{x32}  
and whose only non-logical 
rule of inference is ``constructive induction'', the particular form of which varies from system to  system. The diversity of such theories is typically related to different complexity conditions  
imposed on the underlying concept of interactive computability. For instance,
{\bf CLA4} soundly and completely captures  the set of  polynomial
time solvable interactive number-theoretic problems, {\bf CLA5} does the same for polynomial space, {\bf CLA6} for elementary
recursive time (=space), {\bf CLA7} for primitive recursive time (=space), and {\bf CLA8} for $\pa$-provably recursive time (=space).

\subsection{The present system}
The present paper introduces a new system of clarithmetic, named {\bf CLA11}.\label{x6} Unlike its predecessors, 
this one is a scheme of clarithmetical theories rather than a particular theory. As such, it can be written
as $\mbox{\bf CLA11}^{P_1,P_2,P_3}_{P_4}$
where $P_1$, $P_2$, $P_3$, $P_4$ are ``tunable'' parameters, with different specific settings
of those parameters defining different particular theories of clarithmetic --- different 
{\em instances} of {\bf CLA11},\label{xinst} as we shall refer to them.
Technically, $P_1$, $P_2$, $P_3$ are sets of terms or pseudoterms used as bounds for certain quantifiers in certain postulates, and $P_4$ is a  set of formulas that act as supplementary axioms. The latter is typically empty yet ``expandable''. Intuitively, the value of $P_1$ determines the so called amplitude complexity  of the class
of problems captured by the theory, i.e., the complexity measure
concerned with the sizes of the machine's moves relative to the sizes of the environment's moves. $P_2$ determines the space complexity of that class. $P_3$
determines the time complexity of that class. And $P_4$ governs the intensional strength of the theory. Here
{\bf intensional}\label{x8} strength is about what {\em formulas} are {\em provable} in the theory. This is as opposed to 
 {\bf extensional}\label{x9} strength, which  is about
what number-theoretic {\em problems} are {\em representable}\label{x10} in the theory, where a problem $A$ is said to be representable
iff there is a provable formula $F$ that expresses $A$ under the standard interpretation (model) of arithmetic.

Where $P_1, P_2, P_3$ are sets of terms or pseudoterms identified with the functions that they represent in the
standard model of arithmetic, we say that a computational problem has a $(P_1, P_2, P_3)$ {\bf tricomplexity}\label{x7} solution
if it has a solution (machine's algorithmic winning strategy) that runs in $p_1$ amplitude, $p_2$ space and $p_3$ time
for some triple $(p_1, p_2, p_3) \in P_1 \times P_2 \times P_3$. The main result of this paper is that, as long as
the parameters of $\mbox{\bf CLA11}^{P_1,P_2,P_3}_{P_4}$
satisfy certain natural ``regularity'' conditions, the theory
is sound and complete with respect to the set of problems that have $(P_1, P_2, P_3)$ tricomplexity solutions.
Sound in the sense that every theorem $T$ of $\mbox{\bf CLA11}^{P_1,P_2,P_3}_{P_4}$
represents a number-theoretic computational
problem with a $(P_1, P_2, P_3)$ tricomplexity solution and, furthermore, such a solution can be mechanically extracted from a
proof of $T$. And complete in the sense that every number-theoretic problem with a $(P_1, P_2, P_3)$ tricomplexity
solution is represented by some theorem of $\mbox{\bf CLA11}^{P_1,P_2,P_3}_{P_4}$. Furthermore, as long as $P_4$ contains or entails
all true  sentences of $\pa$, the above {\em extensional completeness}\label{x11} automatically
strengthens to {\em intensional completeness},\label{x12} according to which every formula expressing a problem with a
$(P_1, P_2, P_3)$ tricomplexity solution is a theorem of the theory. Note that intensional completeness
implies extensional completeness but not vice versa, because the same problem may be expressed by many
different formulas, some of which may be provable and some not. G\"{o}del's celebrated theorem is about intensional
rather than extensional incompleteness. In fact, extensional completeness is not at all interesting in the context of classical-logic-based theories such as $\pa$: in such theories, unlike CoL-based theories, it is trivially achieved, because the provable formula $\twg$ represents every true sentence.
G\"{o}del's incompleteness theorem retains its validity for clarithmetical theories, meaning that
 intensional completeness of such theories can only be achieved at the expense of sacrificing recursive
axiomatizability.

The above-mentioned ``regularity'' conditions on the parameters of {\bf CLA11} are rather simple and easy-to-satisfy. As a result, by just ``mechanically'' varying those parameters, we can generate  
a great variety of theories 
for one or another tricomplexity class, the main constraint being that the space-complexity component of the triple should be at
least logarithmic, the amplitude-complexity component  at least linear, and the time-complexity component 
at least polynomial. Some natural examples of such tricomplexities are:

\begin{quote} \  \  

Polynomial amplitude + logarithmic space + polynomial time

Linear amplitude + $O(\log^i)$ space (for any particular $i\in\{1,2,3,\ldots\}$) + polynomial time

Linear amplitude + polylogarithmic space + polynomial time

Linear amplitude + linear space + polynomial time

Polynomial amplitude + polynomial space + polynomial time

Polynomial amplitude + polynomial space + quasipolynomial time

Polynomial amplitude + polynomial space + exponential time

Quasilinear amplitude + quasilinear space + polynomial time

Elementary amplitude + elementary space + elementary time

Primitive recursive amplitude + primitive recursive space + primitive recursive time

You name it\ldots 
 
\end{quote}

\subsection{Related work}
It has been long noticed that many complexity classes can be characterized by certain versions of arithmetic.
Of those, systems of bounded arithmetic\label{x13} should be named as the closest predecessors of
our systems of clarithmetic. In fact, most clarithmetical systems, including {\bf CLA11}, can be classified as bounded
arithmetics because, as with the latter, they control computational complexity by explicit resource bounds attached
to quantifiers, usually in induction or similar postulates.\footnote{Only the quantifiers $\adc$ and $\add$, not $\forall$ or $\exists$. It should be noted that the earlier ``intrinsic theories'' of Leivant
\cite{bbb6} also follow the tradition of quantifier restriction in induction.}   The best known alternative line of research \cite{bel92,bel00,bbb2,hof00,leiv94,sim88},
primarily developed by recursion theorists, controls computational complexity via type information instead. On the logical side, one should also mention ``bounded linear logic'' \cite{bbb5} and
``light linear logic''  \cite{bbb4} of Girard et al. Here we will not attempt any comparison with these alternative
approaches because of  big differences in the defining frameworks.

The story of bounded arithmetic starts with Parikh's 1971 work \cite{parikh}, where the first system  $I\Delta_0$ of
bounded arithmetic was introduced. Paris and Wilkie, in \cite{paris1} and a series of other papers, advanced the study of $I\Delta_0$ and of how it relates to complexity theory. Interest towards the area dramatically intensified after the
appearance of Buss' 1986 influential work \cite{Buss}, where systems of bounded arithmetic for polynomial hierarchy,
polynomial space and exponential time were introduced. Clote and Takeuti \cite{bbb3}, Cook and Nguyen \cite{cook} and
others introduced a host of theories related to other complexity classes. See \cite{BussChapter,cook,Hajek,bbb7}  for comprehensive
surveys and discussions of this line of research. The treatment of bounded arithmetic found in \cite{cook}, which
uses the two-sorted vocabulary of Zambella \cite{zambella}, is among the newest.
Just like the present paper, it offers a method for designing one's own system of bounded arithmetic for
a spectrum of complexity classes within P. Namely, one only needs to add a single axiom to the  base theory $V^0$, where the axiom states the existence of a solution to a complete problem of the
complexity class.

All of the above theories of bounded arithmetic are weak subtheories of $\pa$, typically obtained by
imposing certain syntactic restrictions on the induction axiom or its equivalent, and then adding some old theorems
of $\pa$ as new axioms to partially bring back the baby thrown out with the bath water. Since the weakening of the
deductive strength of $\pa$ makes certain important functions or predicates no longer  
definable, the non-logical vocabularies of these theories typically have to go beyond the original vocabulary $\{0,\successor,+,\times\}$
of $\pa$. These theories achieve soundness and extensional completeness with respect to the corresponding
complexity classes in the sense that a function $f(\vec{x})$ belongs to the target class if and only if it is provably
total in the system --- that is, if there is a $\Sigma_1$-formula $F(\vec{x},y)$ that represents  
the graph of $f(\vec{x})$, such that the system proves $\cla \vec{x}\cle !yF(\vec{x},y)$.

\subsection{Differences with bounded arithmetic}
Here we want to point out several differences between the above systems of bounded
arithmetic and our clarithmetical theories, including (the instances of) {\bf CLA11}. 

\subsubsection{Generality} While the other approaches are about functions, clarithmetics are about interactive problems, with functions being nothing but special cases of the latter. This way, clarithmetics allow us to systematically study not only computability in its narrow sense, but also many other meaningful  properties and relations, such as, for instance,  various sorts of reducibilities (cf. Section \ref{scmptlg}). Just like function effectiveness, such relations happen to be special cases of our broad concept of computability. Namely, a relation holds if and only if the corresponding interactive problem has a solution. Having said that, the differences discussed in the subsequent paragraphs of this subsection hold regardless of whether one keeps in mind the full generality of clarithmetics or restricts attention 
back to functions  only, the ``common denominators'' of the two approaches.

\subsubsection{Intensional strength} Our systems extend rather than restrict $\pa$. Furthermore, instead of $\pa$, as a classical basis one can take
anything from a very wide range of sound theories, beginning from certain weak fragments of $\pa$ and ending
with the absolute-strength theory $Th(N)$\label{x103} of the standard model $N$ of arithmetic (the ``truth arithmetic''). It is exactly due to this
flexibility that we can achieve not only extensional but also intensional completeness --- something inherently
unachievable within the traditional framework of bounded arithmetic, where computational soundness 
by its very definition entails deductive weakness.

\subsubsection{Language} Due to the fact that our theories are no longer weak, there is no need to have
 any new {\em non-logical} primitives
in the language and the associated new axioms in the theory: all recursive or 
arithmetical relations and functions can be expressed through $0,\successor,+,\times$  in the standard way. Instead, as mentioned earlier, the language of our
theories of clarithmetic only has two additional {\em logical} connectives $\adc,\add$   and two additional quantifiers $\ada,\ade$. It is CoL's 
constructive semantics for these operators that allows us to express nontrivial computational
problems. Otherwise, formulas not containing these operators --- formulas of the pure/traditional language of
$\pa$, that is --- only express elementary problems (i.e., moveless games --- see page \pageref{x4}). This explains how our approach makes it possible to reconcile unlimited deductive strength 
with computational soundness. For instance, the formula $\cla x\cle yF(x,y)$ may be provable even if $F(x,y)$ is  
the graph of a function which is ``too hard'' to compute.  This does not have any
relevance to the complexity class characterized by the theory because the formula $\cla x\cle yF(x,y)$, unlike its
``constructive counterpart'' $\ada x\ade yF(x,y)$, carries no nontrivial computational meaning.\footnote{It should be noted that the idea of differentiating between operators (usually only quantifiers) with and without
computational connotation has been surfacing now and then in the literature on complexity-bound arithmetics. For instance, the language of a 
system constructed in \cite{Sch} for polynomial time, along with ``ordinary'' quantifiers used in similar treatments, contains the ``computationally irrelevant'' quantifier $\cla ^{nc}$.}

\subsubsection{Quantifier alternation} Our approach admits arbitrarily many alternations of boun\-ded quantifiers in induction or whatever
similar postulates, whereas the traditional bounded arithmetics are typically very sensitive in this respect,
with different quantifier complexities yielding different computational complexity classes.\footnote{Insensitivity with respect to quantifier alternations is not really without precedents in the literature. See, for instance, \cite{bel98}.
The system introduced there, however, in its creator's own words from  \cite{bbb2}, is ``inadequate as a working logic, e.g., awkwardly
defined and not closed under modus ponens''.}

\subsubsection{Uniformity} As noted, both our approach and that of \cite{cook} offer uniform treatments of otherwise disparate  systems for various complexity classes. The spectrums of complexity classes for which  the two approaches allow one to uniformly
construct adequate systems are, however, different. Unlike the present work, \cite{cook} does not reach  
beyond polynomial hierarchy, thus missing, for instance, linear
space, polynomial space, quasipolynomial time or space, exponential time, etc. On the other hand,
unlike \cite{cook}, our uniform treatment is only about sequential and deterministic computation, thus missing classes such as $AC^0$, $NC^1$, $NL$ or $NC$. A more notable difference between the two approaches, however, is related to how uniformity is achieved. In the case of \cite{cook}, as already mentioned,
the way to ``build your own system'' is to add, to the base theory, an axiom expressing a
complete problem of the target complexity class. Doing so thus requires quite some nontrivial complexity-theoretic knowledge. In our case, on the other hand, adequacy is achieved by
straightforward, brute force tuning of the corresponding parameter of 
$\mbox{\bf CLA11}^{P_1,P_2,P_3}_{P_4}$. E.g., for 
 linear space, we simply need to take the $P_2$ parameter to be the set of
$(0,\successor,+)$-combinations of variables, i.e., the set of terms that ``canonically'' express the linear functions.
If we (simultaneously) want to achieve adequacy with respect to polynomial time, we shall (simultaneously)
take the $P_3$ parameter to be the set of $(0,\successor,+,\times)$-combinations of variables, i.e., the set of terms that express
the polynomial functions. And so on.

\subsection{Motivations}
Subjectively, the primary motivating factor for the author when writing this paper was that it further illustrates the scalability and appeal of CoL, his brainchild. On the objective side, the main motivations are as follows, somewhat arbitrarily divided 
 into the categories ``general'', ``theoretical'' and ``practical''.

\subsubsection{General}

Increasingly loud voices are being heard \cite{dina} that, since the real computers are interactive, it might be time in theoretical computer science to seriously consider switching from Church's narrow understanding of computational problems as functions to more general, interactive understandings. The present paper and clarithmetics in general serve the worthy job of lifting ``efficient arithmetics'' to the interactive level. Of course, these are only CoL's first modest steps in this direction, and there is still a long way to go. In any case, our generalization from functions to interaction appears to be beneficial even if, eventually, one is only interested in functions, because it allows a smoother treatment and makes our systems easy-to-understand in their own rights. Imagine how awkward it would be if one had tried to restrict the language of classical logic only to formulas with at most one 
alternation of quantifiers because more complex formulas seldom express things that we comprehend or care about, and, besides,  things can always be Skolemized anyway. 
Or, if mankind had let the Roman-European tradition prevail in its reluctance to go beyond positive integers and accept $0$ as a legitimate quantity,  to say nothing about the negative, fractional, or irrational numbers. 

The ``smoothness'' of our approach is related to the fact that, in it, all formulas --- rather than only those of the form $\cla x\cle !yF(x,y)$ with $F\in\Sigma_1$ --- have clearly defined meanings as 
computational problems. This allows us to apply certain systematic and scalable methods of analysis that otherwise would be inadequate. For instance, the soundness proofs for various clarithmetical theories go semantically by induction on the lengths of proofs,  by showing that all axioms  
have given (tri)complexity solutions, and that all rules of inference
preserve the property of having such solutions. Doing the same is impossible in the traditional approaches to bounded arithmetic (at least those based on classical logic), because not all intermediate steps in proofs will have the form  $\cla x\cle !yF(x,y)$ with $F\in\Sigma_1$. It is no accident that, to prove computational soundness,  such approaches usually have to appeal to  syntactic arguments that are around ``by good luck'', such as cut elimination. 

As mentioned, our approach 
extends rather than restricts $\pa$. This allows us to safely continue relying on  our standard arithmetical intuitions when reasoning within clarithmetic, without our hands 
being tied by various constraints, without the caution necessary when reasoning within weak theories. Generally, a feel for a formal theory and a ``sixth sense'' that it takes for someone to comfortably reason within the theory require time and efforts to develop. 
Many of us have such a ``sixth sense'' for $\pa$ but not  so many have it for weaker theories. This is so because weak theories, being artificially restricted and thus forcing us to pretend that we do not know certain things
that we actually do know, are farther from a mathematician's normal intuitions than $\pa$ is. Even if this was not the case,  mastering the one and universal theory $\pa$ is still easier and promises a greater payoff than trying to master tens of disparate yet equally important weak theories that are out there.

\subsubsection{Theoretical}
Among the main motivations for studying bounded arithmetics has been a hope that they can take us closer to solving some of the great open problems in complexity theory, for ``it ought to be easier to separate the theories corresponding to the complexity classes than to separate the classes themselves'' (\cite{cook}). The same applies to our systems of clarithmetic and {\bf CLA11} in particular that allows us to capture, in a uniform way, a very wide and diverse range of complexity classes.   

While the bounded arithmetic approach has been around and extensively studied since long ago, the progress towards realizing the above hope has been very slow.  This fact alone justifies all reasonable attempts to try something substantially new and so far not well explored. The clarithmetics line of research qualifies as such. Specifically, studying ``nonstandard models'' of clarithmetics, whatever they may mean, could be worth the effort. 

Among the factors which might be making {\bf CLA11} more promising in this respect than its traditional alternatives is that the former achieves intensional completeness while the latter inherently have to settle for merely extensional completeness. Separating theories intensionally is generally easier 
than separating them extensionally, yet intensional completeness implies that the two sorts of separation mean the same.  

Another factor relates to the ways in which theories are  axiomatized in uniform treatments, namely, the approach of {\bf CLA11} versus that of \cite{cook}. As noted earlier, the uniform method of \cite{cook} achieves extensional completeness with respect to a given complexity class by adding to the theory an axiom expressing a complete problem of that class. The same applies to the method used in \cite{bbb3}. Such axioms are typically long formulas as they carry nontrivial complexity-theoretic information.  They talk --- through encoding and arithmetization --- about graphs, computations, etc. rather
than about numbers. This makes such axioms hard to comprehend directly as number-theoretic statements, and makes the corresponding theories hard to analyze. This approach essentially means translating our complexity-theoretic knowledge into arithmetic. For this reason, it is likely to encounter the same kinds of challenges  as the ordinary, informal theory of computation does when it comes to separating complexity classes.  Also, oftentimes we may simply fail to know a complete problem of a given, not very well studied, complexity class. 

The uniform way in which {\bf CLA11} axiomatizes its instances, as explained earlier, is very different from the above. Here all axioms and rules are ``purely arithmetical'', carrying no direct complexity-theoretic information. This means that the number-theoretic contents of such theories are easy to comprehend, which, in turn, carries a promise that their model theories might be easier to successfully study, develop and use in proving independence/separation results.

\subsubsection{Practical}\label{sprc} 
More often than not, the developers of complexity-bound arithmetics have also been motivated by the potential of practical applications in computer science. Here we quote Schwichtenberg's \cite{Sch} words:  
\begin{quote}
``It is well known that it is undecidable in general whether a given program meets its specification. In contrast,
it can be checked easily by a machine whether a formal proof is correct, and from a constructive proof one
can automatically extract a corresponding program, which by its very construction is correct as well. This
at least in principle opens a way to produce correct software, e.g. for safety-critical applications.
Moreover, programs obtained from proofs are ``commented'' in a rather extreme sense. Therefore it is easy
to apply and maintain them, and also to adapt them to particular situations.''
\end{quote}
Applying the same line of thought to clarithmetics, where, by the way, all proofs qualify as ``constructive'' for the above purposes, the introductory section of \cite{cla4} further adds:
\begin{quote}
``In a more
ambitious and, at this point, somewhat fantastic perspective, after developing reasonable theorem-provers,
CoL-based efficiency-oriented systems can be seen as declarative programming languages in an extreme sense,
where human ``programming'' just means writing a formula expressing the problem whose efficient solution
is sought for systematic usage in the future. That is, a program simply coincides with its specification. The
compiler's job would be finding a proof (the hard part) and translating it into a machine-language code (the
easy part). The process of compiling could thus take long but, once compiled, the program would run fast
ever after.''
\end{quote}

\noindent What matters for applications like the above, of course, is the intensional rather than extensional strength of a theory. The greater that strength, the better the chances that a proof/program will be found for a declarative,  ad hoc specification of the goal. Attempts to put an intensionally weak theory (regardless of its extensional strength) to practical use would usually necessitate some pre-processing of the goal, such as expressing it through a certain standard-form $\Sigma_1$-formula. But this sort of pre-processing often essentially means already finding --- outside the formal system --- a solution of the target problem or, at least, already finding certain insights into such a solution.

In this respect, {\bf CLA11} fits the bill. Firstly, because it is easily, ``mechanically'' adjustable to a potentially infinite variety of target complexities that one may come across in real life. It allows us to adequately capture a complexity class from that variety without any preliminary complexity-theoretic knowledge about the class, such as knowledge of some complete problem of the class (yet another sort of ``pre-processing'') as required by the approaches in the style of  \cite{bbb3} or \cite{cook}. All relevant knowledge about the class is automatically extracted by the system from the definition (ad hoc description) of the class, without any need to look for help outside the formal theory itself. Secondly, and more importantly, {\bf CLA11} fits the bill  because of its intensional strength, which  
includes the full deductive power of $\pa$ and which is only limited by the G\"{o}del incompleteness phenomenon.   
Even when the $P_4$ parameter of a theory $\mbox{\bf CLA11}^{P_1,P_2,P_3}_{P_4}$ is empty (meaning that the theory does not possess any arithmetical knowledge that goes beyond $\pa$), the theory provides ``practically full'' information about $(P_1,P_2,P_3)$ tricomplexity computability. This is in the same sense as $\pa$, despite G\"{o}del's incompleteness, provides 
``practically full'' information about arithmetical truth. Namely, if a formula $F$ is not
provable in $\mbox{\bf CLA11}^{P_1,P_2,P_3}_{P_4}$, it is unlikely that anyone would find a $(P_1,P_2,P_3)$ tricomplexity algorithm solving the problem
expressed by $F$: either such an algorithm does not exist, or showing its correctness requires going beyond
ordinary combinatorial reasoning formalizable in $\pa$.

\subsection{How to read this paper}\label{sstn}
This paper is being published in two parts. The present Part I introduces {\bf CLA11} (Section \ref{ss11}), ``bootstraps'' it (Section \ref{sboot}), looks at certain particular instances of it (Section \ref{sharv}),  and proves its completeness (Sections \ref{s19} and \ref{s19e}). The
forthcoming \cite{cla11b} Part II  is devoted to proving the soundness of the system. 
  Even though the paper is long, a reader inclined to skip the proofs of  its main results can safely drop everything beyond Section \ref{sharv} of the present part, including the entire Part II. Dropping all proofs in the remaining sections of Part I will further reduce the amount of material to be read.

The only external source on which this paper relies is \cite{cl12}, and familiarity with the latter  is a necessary condition for reading this paper. Again, all proofs found in \cite{cl12} can be safely omitted, which should significantly reduce the size of that otherwise fairly long article. Familiarity with \cite{cl12} is also a sufficient condition, because \cite{cl12} presents a self-contained,  tutorial-style introduction to the relevant fragment of CoL. It would be accurate to say that \cite{cl12} is, in fact, ``Part 0'' of the present series of articles.  Having \cite{cl12} at hand for occasional references is necessary even for those who are well familiar with CoL but from some other sources. It contains an index, which can and should be looked up every time one encounters an unfamiliar term or notation. All definitions and conventions of \cite{cl12} are adopted in the present series without revisions. 

\section{The system \texorpdfstring{$\areleven$}{CLA11}}\label{ss11}
 $\areleven$ is a scheme of applied theories based on the system $\cltw$\label{x17} of computability logic, in the same sense as the well known (cf. \cite{Boolos,cook,Hajek,bbb7}) {\bf Peano Arithmetic}\label{x26} $\pa$ is an applied theory based on classical logic. We do not reintroduce logic $\cltw$ here, assuming that the reader is familiar with it  from \cite{cl12}. As noted just a while ago, the same holds for all other concepts used but not explained in this  article.  

\subsection{Language}\label{subsectlang}
The theories that we deal with in this paper have the same language,  obtained from the language of $\cltw$  by removing all nonlogical predicate letters,
removing all constants but $0$, and  removing all but the following three function letters:
 
\begin{itemize}
\item $successor$,\label{x33} unary. We   write $\tau\successor $ for $successor(\tau)$.
\item $sum$, binary. We   write $\tau_1+ \tau_2$ for $sum(\tau_1,\tau_2)$.
\item $product$, binary. We   write $\tau_1\times  \tau_2$ for $product(\tau_1,\tau_2)$.
\end{itemize}

Let us call this language $\mathbb{L}$.\label{x18}  Unless otherwise specified or implied by the context, when we say ``{\bf formula}'',\label{x19} it is to be understood as formula of $\mathbb{L} $. 
As always, {\bf sentences}\label{x20} are formulas with no free occurrences of variables. An {\bf $\mathbb{L}$-sequent}\label{xlsq} is a sequent all of whose formulas are sentences of $\mathbb{L}$. A {\bf paraformula}\label{21}  is defined as the result of replacing, in some formula, some free occurrences of variables by  constants. And a {\bf parasentence}\label{22} is a paraformula with no free occurrences of variables. Every formula is a paraformula but not vice versa, because a paraformula may contain constants other than $0$, which are not allowed in formulas. Yet, oftentimes we may forget about the distinction between formulas and paraformulas, and carelessly say ``formula'' where, strictly speaking, ``paraformula'' should have been said.  In any case, we implicitly let all definitions related to formulas  automatically extend to paraformulas whenever appropriate/possible.

For a formula $F$, \  
$\cla  F$\label{x25} means the $\cla$-closure of $F$, i.e., $\cla x_1\ldots\cla x_n F$, where $x_1,\ldots,x_n$ are the free variables of $F$ listed in their lexicographic order.  
Similarly for $\cle F$, $\ada F$, $\ade F$.

A  formula  is said to be {\bf elementary}\label{x27} iff it is $\adc,\add,\ada,\ade$-free.  We will be using the lowercase $p$, $q$, \ldots\ as metavariables for elementary formulas. This is as opposed to the uppercase letters $E,F,G,\ldots$, which will be used as metavariables for any (elementary or nonelementary) formulas.
 
\subsection{Peano arithmetic}\label{spari}
As one can see, $\mathbb{L} $ is an extension of the language of $\pa$  --- namely, the extension obtained by adding the choice operators $\adc,\add,\ada,\ade$. The language of $\pa$  is the {\em elementary fragment} of $\mathbb{L}$, in the sense that formulas of the former are nothing but elementary formulas of the latter. We remind the reader that, deductively, $\pa$ is the theory based on classical first-order  logic  
with the following nonlogical axioms, that we shall refer to as the {\bf Peano axioms}:\label{x28}
\begin{eqnarray*}
&  & 1. \ \cla x(0\not= x\successor ); \\
 & & 2.\ \cla x\cla y(x\successor = y\successor \mli x= y); \\
&  & 3. \ \cla x(x+ 0= x);\\
 & & 4.\ \cla x\cla y\bigl( x+ y\successor = (x+ y)\successor \bigr);\\
&  & 5. \ \cla x(x\times 0= 0);\\
 & & 6.\ \cla x\cla y\bigl(x\times  y\successor = (x\times  y)+ x\bigr);\\
&  & 7. \ \cla \Bigl(  p(0)\mlc \cla x\bigl(p(x)\mli p(x\successor )\bigr)\mli \cla x\hspace{2pt}p(x)\Bigr) \mbox{  for each elementary formula $p(x)$.}
\end{eqnarray*}

The concept of an interpretation explained in \cite{cl12}  can now be restricted to interpretations that   are only defined 
(other than the word ``{\em Universe}'') on $\successor $, $+ $ and  $\times $, as the present language $\mathbb{L}$ has no 
other nonlogical function or predicate letters. Of such interpretations, the {\bf standard interpretation}\label{x29} 
$^\dagger$ is the one whose universe $\mbox{\em Universe}^\dagger$ is the ideal universe (meaning that  $\mbox{\em Domain}^\dagger$ is $\{0,1,10,11,100,\ldots\}$ and $\mbox{\em Denotation}^\dagger$ is the identity function on $\mbox{\em Domain}^\dagger$), and that interprets  the letter $\successor $ 
as the standard successor function\label{x33s} $\mbox{\em var}_1+1$, interprets $+$ as the sum function 
$\mbox{\em var}_1+ \mbox{\em var}_2$,  and interprets $\times $ as the product function 
$\mbox{\em var}_1\times \mbox{\em var}_2$. We often terminologically identify a (para)formula $F$ with the game $F^\dagger$, 
and typically write $F$ instead of $F^\dagger$ unless doing so may cause ambiguity. Correspondingly, whenever we say that 
an elementary (para)sentence is {\bf true},\label{x30} it is to be understood as that the (para)sentence is true under  
the standard interpretation, i.e., is true in what is more commonly called the {\em standard model of arithmetic}.\label{xsma} 

Terminologically we  will further identify natural numbers with the corresponding binary numerals (constants). Usually it will be clear from the context whether we are talking about a number or a binary numeral. For instance, if we say that $x$ is greater than $y$, then we obviously mean $x$ and $y$ as numbers; on the other hand, if we say that $x$ is longer than $y$, then $x$ and $y$ are seen as numerals. Thus, $111$ (seven) is greater but not longer than $100$ (four). 

If we write \[\hat{0},\hat{1},\hat{2},\ldots\label{x31}\] within formal expressions, they are to be understood as the  terms $0,0\successor,0\successor\successor,\ldots$, respectively. Such terms will be referred to as the 
{\bf unary numerals}.\label{x34} Occasionally, we may carelessly omit $\hat{\ }$ and simply write $0,1,2,\ldots$. 

An  $n$-ary ($n\geq 0$)  
{\bf pterm}\footnote{The word ``pterm'', where ``p'' stands for ``pseudo'', is borrowed from \cite{Boolos}.}\label{x35} is an elementary  formula $p(y,x_1,\ldots,x_n)$ with all free variables as shown and one of such variables --- $y$ in the present case --- designated as what we call the {\bf value variable}\label{xvv} of the pterm,   such that $\pa$ proves $\cla x_1\ldots\cla x_n\cle ! y \hspace{1pt} \tau(y,x_1,\ldots,x_n)$. Here, as always,   $\cle ! y$ means ``there is a unique $y$ such that''.  We call $x_1,\ldots,x_n$ the {\bf argument variables}\label{xav} of the pterm.  If $p(y,\vec{x})$ is a pterm, we shall usually refer to it as $\mathfrak{p}(\vec{x})$ (or just $\mathfrak{p}$), changing  Latin to Gothic and dropping the value variable $y$ (or dropping all variables). Correspondingly, where $F(y)$ is a formula, we write $F\bigl(\mathfrak{p}(\vec{x})\bigr)$ to denote the formula $\cle y \bigl(p(y,\vec{x})\mlc F(y)\bigr)$, which, in turn, is equivalent to $\cla y \bigl(p(y,\vec{x})\mli F(y)\bigr)$. These sort of expressions, allowing us to syntactically treat pretms as if they were genuine terms of the language, are unambiguous in that all ``disabbreviations'' of them are provably equivalent in the system. 
 Terminologically, genuine terms of $\mathbb{L}$, such as $(x_1+x_2)\times x_1$, will  also count as pterms. Every $n$-ary pterm $\mathfrak{p}(x_1,\ldots,x_n)$ represents --- in the obvious sense --- some $\pa$-provably total $n$-ary function $f(x_1,\ldots,x_n)$. For further notational and terminological convenience, in many contexts we shall identify  
pterms with the  functions that they represent.

It is our convention that, unless otherwise specified, if we write a pterm as $\mathfrak{p}(x_1,\ldots,x_n)$ or $\mathfrak{p}(\vec{x})$ (as opposed to just $\mathfrak{p}$) when first mentioning it,   we always imply that the displayed variables are pairwise distinct,  and that they are exactly (all and only) the argument variables of the pterm. 
Similarly, if we write a function as $f(x_1,\ldots,x_n)$ or $f(\vec{x})$ when first mentioning it, we imply that the displayed variables are pairwise distinct,  and that $f$ is an $n$-ary function that does not depend on any variables other than the displayed ones.  
A convention in this style does not apply to formulas though: when writing a formula as $F(\vec{x})$, we do not necessarily imply that all variables of $\vec{x}$ have free occurrences in the formula, or that 
 all free variables of the formula are among $\vec{x}$ (but we still do imply that the displayed variables are distinct).

The language of $\pa$ is known to be  very expressive, despite its nonlogical vocabulary's officially being limited to only $0,\successor ,+ ,\times $. Specifically, it allows us to express, in a certain  standard way, all  recursive functions and relations, and beyond. Relying on the common knowledge of the power of the language of \pa,  we will be using standard expressions such as $x\leq y$, $y> x$, etc. in formulas as abbreviations of the corresponding proper expressions of the language. Similarly for pterms. So, for instance,  if we write ``$x< 2^y$'', it is officially to be understood as an abbreviation of a standard formula of {\bf PA} saying that $x$ is smaller than the $y$th power of $2$.

In our metalanguage, 
\[|x|\label{x40}\]  will refer to the length of (the binary numeral for) $x$. In other words, $|x|=\lceil \log_2(x+ 1)\rceil$, where,  as always, $\lceil z\rceil$ means the smallest integer $t$ with $z\leq t$. As in the case of other standard functions, the expression $|x|$ will be simultaneously  understood as a pterm naturally representing the function $|x|$. The delimiters ``$|\ldots|$'' will 
automatically also be treated as parentheses, so, for instance, when $f$ is a unary function or pterm, we will usually write ``$f|x|$'' to mean the same as the more awkward expression ``$f(|x|)$'' would normally mean.  Further generalizing this notational convention, if $\vec{x}$ stands for an $n$-tuple $(x_1,\ldots,x_n)$ ($n\geq 0$) and   we write $\tau|\vec{x}|$,\label{xtx} it is to be understood as $\tau(|x_1|,\ldots,|x_n|)$.

Among the other pterms/functions that we shall frequently use is 
\[(x)_y,\label{x41a}\] 
standing for $\lfloor x/2^y\rfloor\modulo 2$, where, as always, $\lfloor z\rfloor$ denotes the greatest integer $t$ with $z\geq t$. In other words, $(x)_y$ is the $y$th {\bf least significant bit}\label{x46a} of $x$. 
 Here, as usual,  the bit count starts from $0$ rather than $1$, and goes from right to left, i.e., from the least significant bit to the most significant bit; when $y\geq |x|$,  ``the $y$th least significant bit of  $x$'', by convention, is $0$. Sometimes we will talk about the $y$th {\bf most significant bit}\label{x47a} of  $x$, where $1\leq y\leq |x|$. In this case we count bits from left to right, and the bit count starts from $1$ rather than $0$. So, for instance, $0$ is the $4$th least significant bit and, simultaneously, the $5$th most significant bit, of   $111101111$. This number has a $99$th least significant bit (which is $0$), but it does not have a $99$th most significant bit.  

One more abbreviation that we shall frequently use is $\bit$, defined by 
\[\bit(y,x)\label{x50a}\ =_{def}\ (x)_y=1.\]
 
\subsection{Bounds}\label{sbou}
We say that a pterm $\mathfrak{p}_2$ is a {\bf syntactic variation}\label{x36} of a pterm $\mathfrak{p}_1$ iff there is a function $f$ from the set of (free and bound) variables of $\mathfrak{p}_1$ onto the set of (free and bound) variables of $\mathfrak{p}_2$ such that the following conditions are satisfied: 
\begin{enumerate}
  \item If $x,y$ are two distinct variables of $\mathfrak{p}_1$ where at least one of them is bound, then $f(x)\not=f(y)$.
  \item The two pterms only differ from each other in that, wherever $\mathfrak{p}_1$ has a (free or bound) variable $x$, $\mathfrak{p}_2$ has the variable $f(x)$ instead.
\end{enumerate}   
Example: $y+z$ is a syntactic variation of $x+y$, and so is $z+z$.     

By a   {\bf bound}\label{x60} we shall mean  a pterm $\mathfrak{p}(x_1,\ldots,x_n)$  ---  
 which may as well be written simply as  $\mathfrak{p}(\vec{x})$ or $\mathfrak{p}$ --- satisfying (making true) the following  
{\bf monotonicity}\label{x61} condition:
\[\cla x_1\ldots\cla x_n\cla y_1\ldots\cla y_n\bigl(x_1\leq y_1\mlc\ldots\mlc x_n\leq y_n \mli \mathfrak{p}(x_1,\ldots,x_n)\leq \mathfrak{p}(y_1,\ldots,y_n)\bigr).\] 
A {\bf boundclass}\label{x62} means a  set $\mathcal B$  of bounds closed under syntactic variation, in the sense that,  if a given bound is in $\mathcal B$, then so are all of its  syntactic variations.

Where $\mathfrak{p}$ is a pterm and $F$ is a formula, we  use the abbreviation $\ada x\leq \mathfrak{p} F$ for $\ada x(x\leq \mathfrak{p}\mli F)$, $\ade x \leq \mathfrak{p} F$ for $\ade x(x\leq \mathfrak{p}\mlc F)$, $\ada |x|\leq \mathfrak{p} F$ for $\ada x(|x|\leq \mathfrak{p}\mli F)$, and $\ade |x| \leq \mathfrak{p} F$ for $\ade x(|x|\leq \mathfrak{p}\mlc F)$.\label{xabq} Similarly for the blind quantifiers $\cla $ and $\cle$. And similarly for $<$ instead of $\leq$. 
 
Let $F$ be a  formula and $\mathcal B$ a boundclass. We say that $F $ is {\bf $\mathcal B$-bounded}\label{x63} iff   every $\ada$-subformula (resp. $\ade $-subformula) of $F$  has the form $\ada |z|\leq \mathfrak{b}|\vec{s}|  H$ (resp.   $\ade |z|\leq \mathfrak{b}|\vec{s}|  H$), where $z,\vec{s}$ are pairwise distinct variables  not bound by $\cla$ or $\cle $ in $F$, and $\mathfrak{b}(\vec{s})$ is a bound from ${\mathcal B}$. By simply saying ``{\bf bounded}'' we shall mean ``${\mathcal B}$-bounded for some boundclass $\mathcal B$''.

A {\bf boundclass triple}\label{x64} is a triple ${\mathcal R}=({\mathcal R}\amp,{\mathcal R}\spa,{\mathcal R}\tim)$\label{xtroica} of boundclasses.

\subsection{Axioms and rules}\label{saxru}
Every boundclass triple $\mathcal R$  and set $\mathcal A$ of sentences induces the  {\bf theory}  
$\thr$\label{x65} that we deductively define as follows.\vspace{4pt} 
  
The  {\bf axioms}  of $\thr$, with $x$ and $y$ below being arbitrary two distinct variables, are:
 \begin{eqnarray}
 &   &   \mbox{All {\bf Peano axioms}}; \label{A1}  \\
 &  & \ada x\ade y(y=x\successor), \mbox{ which we call the {\bf Successor axiom}};\label{A15}\\ 
 &  & \ada x\ade y (y=|x|), \mbox{ which we call the {\bf Log axiom}};    \label{A2} \\ 
 & & \ada x\ada y\bigl(\bit(y,x)\add \gneg\bit(y,x)\bigr),  \mbox{ which we call the {\bf Bit axiom}};    \label{A3} \\ 
 &  &  \mbox{All  sentences of $\mathcal A$, which we call  {\bf supplementary axioms}}.   \label{A4} 
\end{eqnarray}\medskip

The {\bf rules of inference} of $\thr$ are  {\em Logical Consequence}, \ {\em ${\mathcal R}$-Induction}, and 
{\em $\mathcal R$-Comprehension}. These rules are meant to deal exclusively with sentences, and correspondingly,   in our schematic representations   (\ref{r2}) and (\ref{r3}) of ${\mathcal R}$-Induction and ${\mathcal R}$-Comprehension below,     each premise or conclusion $H$ should be understood as its $\ada$-closure $\ada H$, with the prefix $\ada$    dropped merely for readability.\vspace{4pt}

The rule of {\bf Logical Consequence}\label{x66} (every application/instance of this rule, to be more precise),  abbreviated as {\bf LC},\label{x67} as already known from \cite{cl12}, is 
\begin{equation}\label{r1}
\frac{E_1\hspace{20pt}\ldots\hspace{20pt}E_n}{F},
\end{equation}
where $E_1,\ldots,E_n$ ($n\geq 0$) as well as $F$ are sentences  such that   $\cltw$ proves the sequent $E_1,\ldots,E_n\intimpl F$. 
More generally,  
we say that a parasentence $F$ is a {\bf logical consequence}\label{23} of parasentences $E_1,\ldots,E_n$ iff {\bf CL12} proves  $E_1,\ldots,E_n\intimpl F$. If here $n=0$, we can simply say that $F$ is {\bf logically valid}.\label{x24}\vspace{4pt} 

The rule of {\bf ${\mathcal R}$-Induction}\label{x68}  is 
\begin{equation}\label{r2}
 \frac{ F(0) \hspace{30pt} F(x)\mli F(x\successor) }{ x\leq \mathfrak{b}|\vec{s}|\mli F(x) },\end{equation}
where $x$ and $\vec{s}$ are pairwise distinct  variables, $F(x)$ is an  ${\mathcal R}\spa$-bounded formula, and 
$\mathfrak{b}(\vec{s})$ is a bound from ${\mathcal R}\tim$. 
We shall say that $ F(0) $ is  the {\bf basis} of induction,\label{x70} and  $ F(x)\mli F(x\successor) $ is the {\bf  inductive step}.\label{x71} Alternatively, we may refer to the two premises as the {\bf left premise}\label{x72} and the {\bf right premise},\label{x73} respectively. 
The variable $x$ has a special status here, and we say that the conclusion follows from the premises by {\bf ${\mathcal R}$-Induction on} $x$. We shall refer to the formula-variable pair $F(x)$ as the {\bf induction formula},\label{x74}
 and refer to the bound $\mathfrak{b}(\vec{s})$ as the {\bf induction bound}.\label{x75}\vspace{4pt}

 The rule of {\bf $\mathcal R$-Comprehension}\label{x69}   is 
\begin{equation}\label{r3}
\frac{p(y)\add\gneg p(y)}{\ade |x| \leq \mathfrak{b}|\vec{s}|\cla y<\mathfrak{b}|\vec{s}| \bigl(\bit(y,x)\leftrightarrow p(y)\bigr)}\end{equation}
($q_1\leftrightarrow q_2$\label{xleftrih} abbreviates $(q_1\mli q_2)\mlc (q_2\mli q_1)$), where $x,y$ and $\vec{s}$ are pairwise distinct variables,  $p(y)$ is an elementary formula not containing $x$,   and 
$\mathfrak{b}(\vec{s})$ is a  bound from ${\mathcal R}\amp$. 
We shall refer to the formula-variable pair $p(y)$ as the {\bf comprehension formula},\label{xcf} and refer to $\mathfrak{b}(\vec{s})$ as the {\bf comprehension bound}.\label{xcbn}\vspace{4pt}
  
When $\mathcal R$ is fixed in a context, we may simply say ``Induction'' and ``Comprehension'' instead of ``$\mathcal R$-Induction'' and ``$\mathcal R$-Comprehension''. Note that, of the three components of $\mathcal R$, the rule of $\mathcal R$-Induction only depends on ${\mathcal R}\spa$ and ${\mathcal R}\tim$, while ${\mathcal R}$-Comprehension only depends on  
${\mathcal R}\amp$. 

\subsection{Provability}\label{spr}
A sentence $F$ is considered to be {\bf provable} in $\thr$, written as
$\thr\vdash F,\label{x98}$
 iff there is a sequence of sentences, called a 
{\bf $\thr$-proof} of $F$, where each sentence is either an  axiom, or  follows from some previous sentences by one of the three  rules of  $\thr$, and where the last sentence is $F$. An {\bf extended $\thr$-proof}\label{x80}
 is defined in the same way, only, with the additional requirement that each application of LC should come together with an attached $\cltw$-proof of the corresponding sequent. 
  
Generally, in the context of $\thr$,  as in the above definition of provability and proofs,  we will only be interested in proving {\em sentences}.
In the premises and conclusions of (\ref{r2}) and (\ref{r3}), however,  we wrote not-necessarily-closed formulas  and pointed out that they were to be understood as their $\ada$-closures. For technical convenience, we continue this practice and agree that,    whenever we write 
$\thr\vdash F$ or say ``$F$ is provable'' for a non-sentence $F$, it  simply means that $\thr\vdash \ada F$. Similarly, 
when we say that $F$ is a logical consequence of $E_1,\ldots,E_n$, what we shall mean is that $\ada F$ is a logical 
consequence of $\ada E_1,\ldots,\ada E_n$. Similarly, when we say that a given strategy solves a given paraformula $F$, 
it is to be understood as that the strategy solves $\ada F$ ($\ada F^\dagger$, that is).  To summarize, when dealing with 
$\thr$ or reasoning within this system, any formula or paraformula with free variables should be understood as its 
$\ada$-closure, unless otherwise specified or implied by the context. An exception is when 
$F$ is an elementary paraformula and we say that $F$ is {\bf true}. This is to be understood as that the $\cla$-closure 
$\cla F$ of $F$ 
is true (in the standard model), for ``truth'' is only meaningful for elementary parasentences (which $\ada F$ generally would not be). An important fact on which we will often rely yet only implicitly so, is that the parasentence $\cla F\mli \ada F$ or the closed sequent $\cla F\intimpl \ada F$ is (always) $\cltw$-provable. In view of the soundness of $\cltw$ (Theorem 8.2 of \cite{cl12}), this means that whenever $F$ an elementary paraformula and $\cla F$ is true, $\ada F$ is automatically won by a strategy that does nothing.

\begin{remark}\label{j3}
{\em Our choice of $\pa$ as the ``{\bf elementary basis}''\label{x81} of $\thr$ --- that is, as the classical theory whose axioms constitute the axiom group (\ref{A1}) of $\thr$ ---   is   rather arbitrary, and its only explanation is that $\pa$ is the best known and easiest-to-deal-with recursively enumerable theory. Otherwise, for the purposes of this paper, a much weaker elementary basis   would suffice. It is interesting to understand exactly what weak subtheories of $\pa$ are sufficient as elementary bases of $\thr$,  but we postpone to the future any attempts  to  answer this question. 
Our choice of the language $\mathbb{L}$ is also arbitrary, and the results of this paper, as typically happens in similar cases,  generalize to a wide range of ``sufficiently expressive'' languages.}  
\end{remark}

As $\pa$ is well known and well studied, we safely assume that the reader has a good feel for what it can prove, so we do not usually further justify $\pa$-provability  claims that we make.  A reader less familiar with $\pa$ can take it as a rule of thumb that, despite G\"{o}del's incompleteness theorems, 
$\pa$ proves every true number-theoretic fact that a contemporary high school  student can establish, or that mankind was  aware of before 1931.
One fact worth noting at this point is  that, due to the presence of the axiom group (\ref{A1}) and the rule of LC, 
\begin{equation}\label{suppa}
\mbox{\em $\thr$ proves every sentence provable in $\pa$.}
\end{equation}  

\subsection{Regularity}\label{sreg}
Let $\mathcal B$ be a set of bounds. We define the {\bf linear closure}\label{x82} of ${\mathcal B}$ as the smallest boundclass $\mathcal C$ 
such that the following  conditions are satisfied:
\begin{itemize}
  \item ${\mathcal B}\subseteq{\mathcal C}$;
  \item $0\in {\mathcal C}$;
  \item whenever a bound $\mathfrak{b}$ is in ${\mathcal C}$, so is the bound $\mathfrak{b}\successor$;\footnote{We assume the presence of some fixed, natural way which, given any pterms $\mathfrak{b},\mathfrak{c}$, generates the pterms (whose meanings are) $\mathfrak{b}\successor$, $\mathfrak{b}+\mathfrak{c}$, $\mathfrak{b}\times\mathfrak{c}$. Similarly for any other standard combinations of pterms/functions, such as, for instance, composition $\mathfrak{b}(\mathfrak{c})$.} 
  \item whenever two bounds $\mathfrak{b}$ and $\mathfrak{c}$ are in ${\mathcal C}$, so is the bound  $\mathfrak{b}+\mathfrak{c}$.
\end{itemize}
The {\bf polynomial closure}\label{x83} of $\mathcal B$ is defined as the smallest boundclass $\mathcal C$ that satisfies the above four  conditions and, in addition, also satisfies the following condition:
\begin{itemize} 
  \item whenever two bounds $\mathfrak{b}$ and $\mathfrak{c}$ are in ${\mathcal C}$, so is the bound  $\mathfrak{b}\times\mathfrak{c}$.
\end{itemize}
Correspondingly, we say that $\mathcal B$ is {\bf linearly closed}\label{x84} 
(resp. {\bf polynomially closed}) 
iff $\mathcal B$ is the same as its linear 
(resp. polynomial) 
closure. 

Let $\mathfrak{b}=\mathfrak{b}(\vec{x})=\mathfrak{b}(x_1,\ldots,x_m)$ and $\mathfrak{c}=\mathfrak{c}(\vec{y})=\mathfrak{c}(y_1,\ldots,y_n)$ be functions  or pterms understood as functions.
We write \[\mathfrak{b} \preceq \mathfrak{c}\label{x85} \]  iff $m=n$ and $\mathfrak{b}(\vec{a})\leq \mathfrak{c}(\vec{a})$ is true for all constants $\vec{a}$.
 Next, where $\mathcal B$ and $\mathcal C$ are boundclasses, we write $\mathfrak{b}\preceq {\mathcal C}$   to mean that $\mathfrak{b}\preceq \mathfrak{c}$ for some $\mathfrak{c}\in{\mathcal C}$, and  write 
${\mathcal B}\preceq {\mathcal C}$ to mean that $\mathfrak{b}\preceq {\mathcal C}$ for all $\mathfrak{b}\in{\mathcal B}$. Finally, where 
$\mathfrak{a}_1,\mathfrak{s}_1,\mathfrak{t}_1,\mathfrak{a}_2,\mathfrak{s}_2,\mathfrak{t}_2$ are bounds, we write 
$(\mathfrak{a}_1,\mathfrak{s}_1,\mathfrak{t}_1)\preceq (\mathfrak{a}_2,\mathfrak{s}_2,\mathfrak{t}_2)$\label{xapr19} to mean that 
$\mathfrak{a}_1\preceq \mathfrak{a}_2$, $\mathfrak{s}_1\preceq \mathfrak{s}_2$ and $\mathfrak{t}_1\preceq \mathfrak{t}_2$.
   
\begin{defi}\label{dt}
We say that a boundclass triple ${\mathcal R}$  is {\bf regular}\label{x86} iff the following conditions are satisfied:\footnote{Not all of these conditions are  independent from each other.} 
\begin{enumerate}  
  \item For every bound $\mathfrak{b}(\vec{s})\in {\mathcal R}\amp\cup {\mathcal R}\spa\cup {\mathcal R}\tim$  and any (=some) variable $z$ not occurring in $\mathfrak{b}(\vec{s})$, the game $\ada \ade z(z=\mathfrak{b}|\vec{s}|)$  has an ${\mathcal R}$ tricomplexity\label{xnw} solution (in the sense of Convention 12.4 of \cite{cl12}), and such a solution can be effectively constructed from $\mathfrak{b}(\vec{s})$. 
 \item  ${\mathcal R}\amp$ is {\bf at least linear},\label{x89} ${\mathcal R}\spa$ is {\bf at least logarithmic}, and ${\mathcal R}\tim$ is {\bf at least polynomial}.\label{x90} This is in the sense that, for any variable $x$, we have $x\preceq {\mathcal R}\amp$, $|x|\preceq {\mathcal R}\spa$ and $x,x^2,x^3,\ldots\preceq {\mathcal R}\tim$.
\item All three components of ${\mathcal R}$ are linearly closed and, in addition, ${\mathcal R}\tim$ is also polynomially closed.
\item For each  component ${\mathcal B}\in\{{\mathcal R}\amp,{\mathcal R}\spa,{\mathcal R}\tim\}$ of $\mathcal R$, whenever $\mathfrak{b}(x_1,\ldots,x_n)$ is a bound in $\mathcal B$  and 
$\mathfrak{c}_1,\ldots,\mathfrak{c}_n\in{\mathcal R}\amp\cup{\mathcal R}\spa $, we have  $\mathfrak{b}(\mathfrak{c}_1,\ldots,\mathfrak{c}_n)\preceq {\mathcal B}$. 
    \item For every triple 
 $\bigl(\mathfrak{a}(\vec{x}),\mathfrak{s}(\vec{x}),\mathfrak{t}(\vec{x})\bigr)$ of bounds in ${\mathcal R}\amp\times {\mathcal R}\spa\times{\mathcal R}\tim$  there is a triple  
 $\bigl(\mathfrak{a}'(\vec{x}),$ $\mathfrak{s}'(\vec{x}),\mathfrak{t}'(\vec{x})\bigr)$ in ${\mathcal R}\amp\times {\mathcal R}\spa\times{\mathcal R}\tim$
such that
$\bigl(\mathfrak{a}(\vec{x}),\mathfrak{s}(\vec{x}),\mathfrak{t}(\vec{x})\bigr)\preceq \bigl(\mathfrak{a}'(\vec{x}),\mathfrak{s}'(\vec{x}),\mathfrak{t}'(\vec{x})\bigr)$ 
and  $|\mathfrak{t}'(\vec{x})|\preceq \mathfrak{s}'(\vec{x})\preceq 
\mathfrak{a}'(\vec{x})\preceq \mathfrak{t}'(\vec{x})$.
\end{enumerate}  
\end{defi}\medskip
   
\noindent Our use of the ``Big-O'' notation\label{xbon} below and elsewhere is standard. One of several equivalent ways to define it is to say that, given any two $n$-ary functions --- or pterms seen as functions --- $f(\vec{x})$ and $g(\vec{y})$,  $f(\vec{x})= O(g(\vec{y}))$ (or simply $f=O(g)$) means that there is a natural number $k$ such that $f(\vec{a})\leq kg(\vec{a})+k$ for all $n$-tuples $\vec{a}$ of natural numbers. If we say ``$O(g)$ amplitude'', it is to be understood as ``$f$  amplitude for some $f$ with $f=O(g)$''. Similarly for space and time. 

\begin{lemma}\label{lcol}
Assume $\mathcal R$ is a regular boundclass triple, ${\mathcal B}\in\{{\mathcal R}\amp,{\mathcal R}\spa,{\mathcal R}\tim\}$, 
 $\mathfrak{f}=\mathfrak{f}(x_1,\ldots,x_n)$ $(n\geq 0$) is a function,    
$\mathfrak{b}$ is an $n$-ary  bound from  ${\mathcal B}$, and $\mathfrak{f}=O(\mathfrak{b})$.  Then $\mathfrak{f}\preceq {\mathcal B}$.  
\end{lemma}

\begin{proof} Assume the conditions of the lemma. 
The condition $\mathfrak{f}=O(\mathfrak{b})$ means that, for some number $k$, \ $\mathfrak{f}(\vec{z})\preceq \hat{k}\times\mathfrak{b}(\vec{z})+\hat{k}$.  But, by condition 2 of Definition \ref{dt}, ${\mathcal B}$ is linearly closed.  Hence $\hat{k}\times \mathfrak{b}(\vec{z})+\hat{k}$ is in $\mathcal B$. Thus, $\mathfrak{f}\preceq {\mathcal B}$. 
\end{proof} 

\begin{remark}\label{pizza}
{\em When $\mathcal R$ is a regular boundclass triple, the above lemma allows us to safely rely on asymptotic (``Big-O'') terms and asymptotic analysis  when trying to show that a given machine $\mathcal M$ runs in time ${\mathcal R}\tim$. Namely, it is sufficient to show that $\mathcal M$ runs in time $O(\mathfrak{b})$ for some $\mathfrak{b}\in {\mathcal R}\tim$ or even just $\mathfrak{b}\preceq {\mathcal R}\tim$. Similarly for space and amplitude. }
\end{remark}

\begin{defi}\label{dadm}
We say that a theory $\thr$   is {\bf regular}\label{x91} iff the boundclass triple ${\mathcal R}$  is regular and, in addition, the following conditions are satisfied: 
\begin{enumerate}  
  \item Every sentence  of $\mathcal A$ has an ${\mathcal R}$ tricomplexity solution. Here, if $\mathcal A$ is infinite, we additionally require that there is   an effective procedure that returns an ${\mathcal R}$ tricomplexity   solution for each sentence of $\mathcal A$.
  \item For every bound $\mathfrak{b}(\vec{x})$ from 
${\mathcal R}\amp\cup{\mathcal R}\spa\cup{\mathcal R}\tim$  and every (=some) variable $z$ not occurring in 
$\mathfrak{b}(\vec{x})$, $\thr$ proves $ \ade z(z=\mathfrak{b}|\vec{x}|)$.
\end{enumerate}
\end{defi}

\subsection{Main result}\label{smr}
By an ({\bf arithmetical}) {\bf problem}\label{x93} in this paper we mean a game $G$ such that, for some sentence $X$, \ $G= X^\dagger$ (remember that $^\dagger$ is the standard interpretation). Such a sentence $X$ is said to be a {\bf representation}\label{x94} of $G$.
We say that a  problem $G$ is {\bf representable}\label{x95} in $\thr$  and write \[\thr \hspace{2pt}\rep\hspace{2pt} G\label{x99}\] iff $G$   has a $\thr$-provable representation.  

The {\bf truth arithmetic}, denoted\label{x101} $Th(N)$, is the set of all true elementary sentences.   
We agree that, whenever $\mathcal A$ is a set of (not necessarily elementary) sentences, ${\mathcal A}!$ is an abbreviation defined by 
\[{\mathcal A}!\label{x100}\ =\ {\mathcal A}\cup Th(N).\]

In these terms,  the central theorem of the present paper reads  as follows:

\begin{thm}\label{tt1}
Assume  a theory $\thr$ is  regular. Then the following conditions are satisfied:
\begin{enumerate}  
\item  {\bf  Extensional  adequacy:} A  problem $G$   has an ${\mathcal R}$ tricomplexity solution iff  $\thr\hspace{2pt}\rep \hspace{2pt}G$.\vspace{2pt} 

\item   {\bf  Intensional adequacy:}   A sentence $X$   has an ${\mathcal R}$ tricomplexity solution iff $\areleven^{\mathcal R}_{{\mathcal A}!}\vdash X$.\vspace{2pt} 

\item {\bf  Constructive soundness:} There is an effective procedure that takes an arbitrary extended $\areleven^{\mathcal R}_{{\mathcal A}!}$-proof of an arbitrary sentence $X$ and constructs an    
${\mathcal R}$ tricomplexity  solution for $X$. 
\end{enumerate}
\end{thm}

\begin{proof} The completeness (``only if'') parts of clauses 1 and 2  will be proven in Sections \ref{s19} and \ref{s19e}, respectively,
  and the soundness (``if'') part of either clause is immediately implied by   clause 3. The latter will be verified in \cite{cla11b}.
\end{proof}

\section{Bootstrapping \texorpdfstring{$\thr$}{CLA11AR}}\label{sboot}
Throughout this section, we assume that $\thr$ is a regular theory. Unless otherwise specified, ``provable'' means ``provable in $\thr$''. ``Induction'' and ``Comprehension'' mean ``$\mathcal R$-Induction'' and ``$\mathcal R$-Comprehension'', respectively. We continue to use our old convention according to which, context permitting,  $F$ can be written instead of $\ada F$. 

In order to prove the completeness of $\thr$, some work on establishing the provability of certain basic theorems in the system has to be done. This is also a good opportunity for the reader to gain intuitions about our system. This sort of often boring but necessary work  is called {\em bootstrapping}, named after the expression ``to lift oneself by one's bootstraps'' (cf. \cite{BussChapter}).  

\subsection{How we reason in clarithmetic}
Trying to generate full formal proofs in $\thr$, just like doing so in $\pa$, would be far from reasonable in a paper meant to be read by humans. This task is comparable with showing the existence of a Turing machine for one or another function. 
Constructing Turing machines if full detail is seldom feasible, and one usually resorts to some sort of lazy/informal constructions, such 
as constructions that rely on the Church-Turing thesis. Thesis 9.2 of \cite{cl12} will  implicitly  act in the role of ``our Church-Turing 
thesis'' when dealing with $\thr$-provability, allowing us to replace formal proofs with informal/intuitive descriptions of interpretation-independent winning strategies --- according to the thesis, once such a strategy exists for a given formula, we can be sure that the formula is provable.
In addition, we will be heavily 
relying  on our observation (\ref{suppa}) that   $\thr$ 
proves everything provable in $\pa$.   As noted earlier,  since $\pa$ is well known and since it proves ``essentially all'' true arithmetical facts,  we will hardly ever  
try to justify the $\pa$-provability claims that we  make, often only implicitly. Furthermore,   in relatively simple cases, we usually will not try to justify our 
 $\cltw$-provability claims of the sort $\cltw\vdash E_1,\ldots,E_n\intimpl F$ either and, instead, simply say that $F$ 
follows from $E_1,\ldots,E_n$ by LC (Logical Consequence), or that $F$ is a {\em logical consequence} of $E_1,\ldots,E_n$, or that $E_1,\ldots,E_n$ {\em logically imply}\label{xlcc} $F$. What allows us to take this kind of liberty is 
that $\cltw$ is an analytic 
system, and verifying provability in it is a mechanical 
job that a distrustful reader can do on his or her own; alternatively, our non-justified $\cltw$-provability 
claims can always be verified intuitively/informally based on Thesis 9.2 of \cite{cl12}.\footnote{Of course, when dealing with formula schemes (e.g., as in Fact \ref{rr26}) rather than particular formulas, the analyticity of $\cltw$ may not always be directly usable. However, in such cases, Thesis 9.2 of \cite{cl12} still remains at our full disposal.}

The following fact is the simplest of those established in this section, so let us look at its proof as a warm-up exercise.
Remember from Section  \ref{spari}  that $\hat{0}=0$, \ $\hat{1}=0\successor$, \ $\hat{2}=0\successor\successor$, \ $\hat{3}=0\successor\successor\successor$, \ \ldots

\begin{fact}\label{numerals}
For any natural number $n$, $\thr\vdash \ade z(z=\hat{n})$.
\end{fact}

\begin{proof} Fix an $n$ and argue in $\thr$. 
Using  $0$ and the Successor axiom, we find the value $y_1$ of $0\successor$. Then, using $y_1$ and the Successor axiom again, we find the value $y_2$ of $0\successor\successor$. And so on, $n$ times. This way, we find the value $y_n$ of 
$\hat{n}$. We now choose $y_n$ for $z$ in   $\ade z(z=\hat{n})$ and win this game. 
\end{proof}

What is the precise meaning of the second sentence of  the above proof? The Successor axiom $\ada x\ade y (y=x\successor)$ is a resource that we  can use any number of times. As such, it is a game played and always won by its {\em provider}\label{xprovider} (=our environment) in the role of $\pp$ against us, with us acting in the role of $\oo$. So, a value for $x$ in this game should be picked by us. We choose $0$, bringing the game down to $\ade y (y=0\successor)$. The resource provider will   have to respond with a choice of a value (constant) $y_1$ for $y$, further bringing the game down to $y_1=0\successor$. This elementary game is  true (otherwise the provider would have lost), meaning that   $y_1$ is the value --- which we have just {\em found} --- of $0\successor$. 

The rest of the proof of Fact \ref{numerals} should be understood as that we play  $\ada x\ade y (y=x\successor)$ against its provider once again, but this time we specify the value of $x$ as $y_1$, bringing the game down to $\ade y (y=y_1\successor)$. In response, the provider will have to further bring the game down to $y_2=y_1\successor$ for some constant $y_2$. This means that now we know the value $y_2$ of $0\successor\successor$. And so on. Continuing this way, eventually we  come to know the value $y_n$ of $\hat{n}$. 
Now we can and do win the  
target game $\ade z(z=\hat{n})$ by choosing $y_n$ for $z$ in it, thus bringing it down to the true $y_n=\hat{n}$.  

Out of curiosity, let us also take a look at a formal counterpart of our informal proof of $\ade z(z=\hat{n})$. Specificity, consider the case of $n=2$.  A non-extended $\thr$-proof of $\ade z(z=\hat{2})$
consists of just the following two lines:

\begin{quote} \ \vspace{-8pt}

I. \ $\ada x\ade y(y=x\successor)$  \ \ \ Successor axiom

II. $\ade z(z=0\successor\successor)$ \ \ \ LC: I

\ \vspace{-8pt}

\end{quote}

Step II above is justified by LC which, in an extended proof, needs to be supplemented with a $\cltw$-proof of the sequent $\ada x\ade y(y=x\successor)\intimpl \ade z(z=0\successor\successor)$. Below is such a proof:

\begin{quote} \ \vspace{-8pt}

1. $y_1=0\successor,\ y_2=y_1\successor\ \intimpl\ y_2=0\successor\successor$ \ \ \ Wait: (no premises)

2. $y_1=0\successor,\ y_2=y_1\successor\ \intimpl\ \ade z(z=0\successor\successor)$ \ \ \ $\ade$-Choose: 1

3. $y_1=0\successor,\ \ade y(y=y_1\successor)\ \intimpl\ \ade z(z=0\successor\successor)$ \ \ \ Wait: 2

4. $y_1=0\successor,\ \ada x\ade y(y=x\successor)\ \intimpl\ \ade z(z=0\successor\successor)$ \ \ \ $\ada$-Choose: 3

5. $\ade y(y=0\successor),\ \ada x\ade y(y=x\successor)\ \intimpl\ \ade z(z=0\successor\successor)$ \ \ \ Wait: 4

6. $\ada x\ade y(y=x\successor),\ \ada x\ade y(y=x\successor)\ \intimpl\  \ade z(z=0\successor\successor)$ \ \ \  $\ada$-Choose: 5

7. $\ada x\ade y(y=x\successor) \ \intimpl\  \ade z(z=0\successor\successor)$ \ \ Replicate: 6

\ \vspace{-8pt}

\end{quote}

Unlike the above case, most formulas shown to be $\thr$-provable in this section will have free occurrences of variables. As a very simple example, consider $\ade y(y=x)$.  Remembering that it is just a lazy way to write $\ada x\ade y(y=x)$, our informal justification/strategy (translatable into a formal $\thr$-proof) for this formula would go like this:
\begin{quote} {\em Wait till Environment chooses a constant $c$ for $x$, thus bringing the game down to $\ade y(y=c)$. Then choose the same $c$ for $y$. We win because the resulting elementary game $c=c$ is true.}
\end{quote}
However, more often than not, in cases like this we will omit the routine phrase ``wait till Environment chooses constants for all free variables of the formula'', and correspondingly treat the free variables of the formula as standing for the constants already chosen by Environment for them. So, a shorter justification for the above $\ade y(y=x)$ would be:
\begin{quote}{\em Choose (the value of) $x$ for $y$. We win because the resulting elementary game $x=x$ is true.}
\end{quote}
Of course, an even more laconic justification  would be just the phrase ``{\em Choose $x$ for $y$.}'', quite sufficient and thus acceptable  due to the simplicity of the case. Alternatively, we can  simply say that the formula $\ade y(y=x)$ is logically valid (follows by LC from no premises).  

A reader who would like to see some 
additional illustrations and explanations, can browse Sections 11 and 12 of \cite{cla4}. In any case,
the informal methods of reasoning induced by computability logic and clarithmetic in particular
cannot be concisely or fully {\em explained}, but rather they should be {\em learned} through experience and practicing, not unlike the way one learns a foreign language. A reader who initially does not find some of our informal $\thr$-arguments very clear, should not feel disappointed. Greater fluency and better understanding will come gradually and inevitably. Counting on that, as we advance in this paper, 
 the degree of  ``laziness'' of our informal reasoning within $\thr$ will gradually increase, more and more often omitting explicit references to CL,  $\pa$, axioms or certain already established and frequently used facts when justifying certain relatively simple steps. 

\subsection{Reasonable Induction}\label{ssrea}
\begin{fact}\label{rr26}
The set of theorems of $\thr$ will remain the same if, instead of the ordinary $\mathcal R$-Induction rule 
(\ref{r2}), 
 one takes the following rule, which we call {\bf Reasonable $\mathcal R$-Induction}:\label{xrin}
\begin{equation}\label{r23}
 \frac{ F(0) \hspace{30pt} x< \mathfrak{b}|\vec{s}|\mlc F(x)\mli F(x\successor) }{ x\leq \mathfrak{b}|\vec{s}|\mli F(x) },\end{equation}
where $x$, $\vec{s}$,  $F(x)$,  $\mathfrak{b}$ are as in 
(\ref{r2}). 
\end{fact}

\begin{proof} 
To see that the two rules are equivalent, observe that, while having identical left premises and identical conclusions,  the right premise of (\ref{r23}) is weaker than that of (\ref{r2})  --- the latter immediately implies the former by LC. This means that whenever  old induction 
 is applied, its conclusion can just as well be obtained through first weakening the premise 
$F(x)\mli F(x\successor)$ to $x< \mathfrak{b}|\vec{s}|\mlc F(x)\mli F(x\successor)$ using LC, and then applying (\ref{r23}). 

For the opposite direction, consider an application of (\ref{r23}). Weakening (by LC) its left premise $F(0)$, we find the following formula provable:
\begin{equation}\label{227a}
0\leq\mathfrak{b}|\vec{s}|\mli F(0).
\end{equation} 
Next, the right premise  $x< \mathfrak{b}|\vec{s}|\mlc F(x)\mli F(x\successor)$ of (\ref{r23}), together with the $\pa$-provable $\cla (x\successor\leq \mathfrak{b}|\vec{s}|\mli x< \mathfrak{b}|\vec{s}|)$ and $\cla (x\successor \leq\mathfrak{b}|\vec{s}|\mli x \leq\mathfrak{b}|\vec{s}|)$,  can be seen to logically imply  
\begin{equation}\label{227b}
\bigl(x\leq\mathfrak{b}|\vec{s}|\mli F(x)\bigr)\mli \bigl(x\successor \leq\mathfrak{b}|\vec{s}|\mli F(x\successor)\bigr).
\end{equation}
Applying rule (\ref{r2}) to (\ref{227a}) and (\ref{227b}), we get $x\leq\mathfrak{b}|\vec{s}|\mli\bigl(x\leq\mathfrak{b}|\vec{s}|\mli F(x)\bigr)$. The latter, by LC, immediately yields the target $x\leq\mathfrak{b}|\vec{s}|\mli F(x)$. 
\end{proof}

\subsection{Reasonable Comprehension}
\begin{fact}\label{rr266}
The set of theorems of $\thr$ will remain the same if, instead of the ordinary $\mathcal R$-Comprehension rule 
(\ref{r3}), one  takes the following rule, which we call {\bf Reasonable $\mathcal R$-Comprehension}:\label{xrrco}
\begin{equation}\label{r333}
\frac{y<\mathfrak{b}|\vec{s}|\mli p(y)\add\gneg p(y)}{\ade |x| \leq \mathfrak{b}|\vec{s}|\cla y<\mathfrak{b}|\vec{s}| \bigl(\bit(y,x)\leftrightarrow p(y)\bigr)},\end{equation}
where $x$, $y$, $\vec{s}$,  $p(y)$,  $\mathfrak{b}$ are as in (\ref{r3}).
\end{fact}

\begin{proof} 
The two rules have identical conclusions, and the premise of (\ref{r333}) is a logical consequence of the premise of 
(\ref{r3}). So, whatever can be proven using 
(\ref{r3}), can just as well be proven using (\ref{r333}). 

For the opposite direction, consider an application of (\ref{r333}).   
Of course, $\thr$ proves the logically valid $y=y\add\gneg y=y$ without using either version of comprehension.
From here, by (\ref{r3}),  we obtain 
\begin{equation}\label{ii1}
\ade |x| \leq \mathfrak{b}|\vec{s}|\cla y<\mathfrak{b}|\vec{s}| \bigl(\bit(y,x)\leftrightarrow y=y\bigr),
\end{equation}
 which essentially means that the system proves the existence of a number $x_0$ whose binary representation consists of $\mathfrak{b}|\vec{s}|$ ``$1$''s.  Argue in $\thr$. Using (\ref{ii1}), we find the above number $x_0$. 
From $\pa$, we can see that  $|x_0|= \mathfrak{b}|\vec{s}|$. Now, we can win the game 
\begin{equation}\label{ii2}
y<\mathfrak{b}|\vec{s}| \add \gneg y< \mathfrak{b}|\vec{s}|.
\end{equation}
Namely, our strategy for (\ref{ii2}) is to find whether $\bit(y,x_0)$ is true or not using the Bit axiom; then, if true, we --- based on $\pa$ --- conclude that $y<|x_0|$, i.e., that $y<\mathfrak{b}|\vec{s}|$, and choose the left $\add$-disjunct in (\ref{ii2}); otherwise we conclude that $\gneg y<|x_0|$, i.e., $\gneg y<\mathfrak{b}|\vec{s}|$, and choose the right $\add$-disjunct in (\ref{ii2}). 

The following is a logical consequence of  (\ref{ii2}) and of the premise of (\ref{r333}):
\begin{equation}\label{tunnel}
\bigl(y<\mathfrak{b}|\vec{s}|\mlc p(y)\bigr)\add \gneg \bigl(y<\mathfrak{b}|\vec{s}|\mlc p(y)\bigr).
\end{equation}
Indeed, here is  a strategy for (\ref{tunnel}). Using (\ref{ii2}), figure out whether  $y<\mathfrak{b}|\vec{s}| $ is true or false. If false, choose the right $\add$-disjunct in (\ref{tunnel}) and rest your case. Suppose now   $y<\mathfrak{b}|\vec{s}| $ is true. Then, using the premise of (\ref{r333}), figure out whether $p(y)$ is true or false. If true (resp. false), choose the left (resp. right) $\add$-disjunct in (\ref{tunnel}). 
 
Applying rule (\ref{r3})  to (\ref{tunnel}) yields 
\begin{equation}\label{ii3}
\ade |x|\leq \mathfrak{b}|\vec{s}|\cla y<\mathfrak{b}|\vec{s}|\Bigl(\bit(y,x)\leftrightarrow\bigl(y<\mathfrak{b}|\vec{s}|\mlc p(y)\bigr)\Bigr).
\end{equation}
Now, the conclusion of (\ref{r333}), obtaining which was our goal, can easily be seen to be a logical consequence of (\ref{ii3}). 
\end{proof}

\subsection{Addition}
Throughout this and the subsequent subsections  we assume that the variables involved in a formula whose provability is claimed are pairwise distinct.  

\begin{fact}\label{plus}
$\thr\vdash  \ade z(z=u+v)$.
\end{fact}

\begin{proof} We shall rely on the pencil-and-paper algorithm for adding two numbers with ``carrying''  
which everyone is familiar with, as the algorithm is taught at the elementary school level (albeit for decimal rather than binary numerals). Here is an example to refresh our memory. Suppose we are adding the two binary numbers  $u=10101$ and  $v=1101$. They, 
together with the resulting number $z=100010$, should be written as rows in a right-aligned table 
as shown below:
\[\begin{array}{r}
10101\\
\mbox{\raisebox{7pt}{$+$}}\hspace{10pt} 1101\\
----\\
100010
\end{array}\]
The algorithm constructs the sum $z$  bit by bit, in the right-to-left order, i.e., starting from the least significant bit $(z)_0$. At any step $y>0$ we  have a ``carry'' $c_{y-1}\in\{0,1\}$ from the preceding step $y-1$. For uniformity, at step $0$, i.e., when computing $(z)_0$, the ``carry'' $c_{-1}$ from the non-existing ``preceding step'' $\#-1$ is stipulated to be $0$. Anyway, at each step $y=0,1,2,\ldots$, we first find the sum $t_y=(u)_y+(v)_y+c_{y-1}$. Then we declare $(z)_y$ to be $0$ (resp. $1$) if $t_y$ is even (resp. odd); and we declare $c_y$ --- the carry from the present step $y$ that should be ``carried over'' to the next step $y+1$ --- to be $0$ (resp. $1$) if $t_y\leq 1$ (resp. $t_y>1$).

Let 
$\adcarry(y,u,v)$\label{xcar1}  be a natural arithmetization of the predicate ``When calculating the $y$th least significant bit 
of $u+v$ using the above pencil-and-paper  algorithm, the carry $c_y$ generated by the corresponding ($y$th) step   is $1$.'' 

Argue in $\thr$. Our main claim is
\begin{equation}\label{trt3}
y\leq |u|+|v|\mli \bigl(\adcarry(y,u,v)\add\gneg\adcarry(y,u,v)\bigr)\mlc \bigl(\bit(y,u+v)\add \gneg \bit(y,u+v)\bigr),
\end{equation}
which we justify by Induction on $y$. Note that the  conditions of ${\mathcal R}$-Induction are indeed satisfied here: in view of the relevant clauses  of Definition \ref{dt}, 
the linear bound $u+v$ used in the antecedent of (\ref{trt3}) is in ${\mathcal R}\tim$ as it should. To solve the basis 
\begin{equation}\label{99}
\bigl(\adcarry(0,u,v)\add\gneg\adcarry(0,u,v)\bigr)\mlc \bigl(\bit(0,u+v)\add \gneg \bit(0,u+v)\bigr),
\end{equation}
 we use  the Bit axiom  and figure out whether $\bit(0,u)$ and $\bit(0,v)$ are true. If both are true, we choose 
$\adcarry(0,u,v)$ and $\gneg \bit(0,u+v)$ in the corresponding two conjuncts of (\ref{99}). If both are false, we choose 
$\gneg \adcarry(0,u,v)$ and $\gneg \bit(0,u+v)$. Finally, if exactly one of the two is true, we choose  $\gneg \adcarry(0,u,v)$ and $ \bit(0,u+v)$.

 The inductive step is 
\begin{equation}\label{trt}
\begin{array}{c}
\bigl(\adcarry(y,u,v)\add\gneg\adcarry(y,u,v)\bigr)\mlc \bigl(\bit(y,u+v)\add \gneg \bit(y,u+v)\bigr)\mli\\
 \bigl(\adcarry(y\successor,u,v)\add\gneg\adcarry(y\successor,u,v)\bigr)\mlc \bigl(\bit(y\successor,u+v)\add \gneg \bit(y\successor,u+v)\bigr) .
\end{array}
\end{equation}
The above is obviously solved by the following strategy.  We wait till the adversary tells us, in the antecedent, whether $\adcarry(y,u,v)$ is true.
After that, using the Successor axiom, we  compute the value of $y\successor$  and then, using the Bit axiom,  figure out  whether $\bit(y\successor,u)$ and  $\bit(y\successor,v)$ are true. If at least two of these three statements are true, we choose 
 $\adcarry(y\successor,u,v)$ in the left conjunct of the consequent of (\ref{trt}), otherwise choose   $\gneg\adcarry(y\successor,u,v)$. Also, 
if either one or all three statements are true, we additionally choose $\bit(y\successor,u+v)$ in the right conjunct of the consequent of (\ref{trt}), otherwise   choose  $\gneg\bit(y\successor,u+v)$.
(\ref{trt3}) is thus proven.

Of course (\ref{trt3}) logically implies $y< |u|+|v| \mli \bit(y,u+v)\add \gneg \bit(y,u+v)$, from which, by Reasonable 
Comprehension (where the comprehension bound $u+v$ is linear and hence, by Definition \ref{dt}, is guaranteed to be in ${\mathcal R}\amp$ as it should), we get 
\begin{equation}\label{io2}
\ade |z|\leq |u|+|v|\cla y< |u|+|v|\bigl(\bit(y,z)\leftrightarrow  \bit(y,u+v)\bigr).
\end{equation}

The following is a true (by $\pa$) sentence:
\begin{equation}\label{io3}
\cla u\cla v\cla |z|\leq |u|+|v|\Bigl(\cla y< |u|+|v|\bigl(\bit(y,z)\leftrightarrow \bit(y,u+v)\bigr)\mli z=u+v\Bigr).
\end{equation}
Now, the target $\ade z(z=u+v)$ is a logical consequence of (\ref{io2}) and (\ref{io3}). 
\end{proof}

\subsection{Trichotomy}
 
\begin{fact}\label{tri}
$\thr\vdash (u<v)\add (u=v)\add (u>v)$.
\end{fact}

\begin{proof} Argue in $\thr$. Let $u<^xv$ be an abbreviation of $ (u\modulo 2^{x}) <(v\modulo 2^{x})$, $u=^xv$ an abbreviation of 
$(u \modulo 2^x)=(v\modulo 2^x)$, and $u>^xv$ an abbreviation of $(u\modulo 2^x) >(v\modulo 2^x)$. 

By Induction on $x$, we first want to prove 
\begin{equation}\label{24a}
x\leq |u|+|v|\mli  (u <^x v)\add (u =^x v) \add (u>^x v).
\end{equation}
The basis  $(u <^0 v)\add (u =^0 v) \add (u>^0 v)$ of induction is won by choosing the obviously true $u =^0 v$ component.  The inductive step is 
\begin{equation}\label{24b}
(u <^x v)\add (u =^x v) \add (u>^x v)\mli (u <^{x'} v)\add (u =^{x'} v) \add (u>^{x'} v).
\end{equation}
To solve (\ref{24b}),  using the Bit axiom, we figure out the truth status of $\bit(x,u)$ and 
$\bit(x,v)$. If $\bit(x,u)$ is false while $\bit(x,v)$ is true, 
 we choose $u <^{x'} v $ in the consequent of  (\ref{24b}). If vice versa, we choose 
 $u >^{x'} v $. Finally, if both $\bit(x,u)$ and 
$\bit(x,v)$ are true or both are false,  we wait till Environment resolves the antecedent of (\ref{24b}). If it chooses $u <^x v$ (resp. $u =^x v$, resp. $u >^x v$) there, we choose $u <^{x'} v$ (resp. $u =^{x'} v$, resp. $u>^{x'} v$) in the consequent. With some basic knowledge from $\pa$, our strategy cab be seen to be successful. 

Having established (\ref{24a}), this is how we solve  $(u<v)\add (u=v)\add (u>v)$.
Using the Log axiom and Fact \ref{plus}, we find the value   $d$ with  $d=|u|+|v|$. Next, we plug  $d$ for  $x$ (i.e., specify $x$ as $d$) in   (\ref{24a}), resulting in 
\begin{equation}\label{24c}
d\leq |u|+|v|\mli  (u <^d v)\add (u =^d v) \add (u>^d v). 
\end{equation} 
The antecedent of (\ref{24c}) is true, so (\ref{24c})'s provider will have to resolve the consequent. If the first (resp. second, resp. third) $\add$-disjunct is chosen there, we choose the first (resp. second, resp. third) $\add$-disjunct in the target $ (u<v)\add (u=v)\add (u>v)$ and rest our case. 
By $\pa$, we know that $(u <^d v)\mli (u <  v)$, $(u =^d v)\mli (u =  v)$ and $(u>^d v)\mli (u>  v)$ are true. It is therefore obvious that our strategy succeeds. 
 \end{proof}

\subsection{Subtraction}
In what follows, we use $\ominus$ for a natural pterm for {\bf limited subtraction}, defined by $u\ominus v=\max(0,u-v)$. 

\begin{fact}\label{minus}
$\thr\vdash \ade z(z=u\ominus v)$.
\end{fact}

\begin{proof} The present proof is  rather similar to our earlier proof of Fact \ref{plus}. It relies on the elementary school pencil-and-paper algorithm for computing $u-v$ (when $u\geq v$). 
This algorithm, just like the algorithm for $u+v$, constructs the value $z$ of $u-v$  digit by digit, in the right-to-left order.  At any step $y>0$, we  have a ``borrow'' (which is essentially nothing but a ``negative carry'') $b_{i-1}\in\{0,1\}$ 
 from the preceding step $y-1$. For  step $0$, the ``borrow'' $b_{-1}$ from the non-existing ``preceding step'' $\#-1$ is stipulated to be $0$. At each step $y=0,1,2,\ldots$, we first find the value $t_y=(u)_y-(v)_y-c_{y-1}$. Then we declare $(z)_y$ to be $0$ (resp. $1$) if $t_y$ is even (resp. odd); and we declare $b_y$ --- the value ``borrowed'' by the present step $y$ from the next step $y+1$ --- to be $0$ (resp. $1$) if $t_y>-1$ (resp. $t_y\leq -1$). 

Let $\subcarry(y,u,v)$\label{xborr}  be a natural arithmetization of the predicate ``$u\geq v$ and, when calculating the $y$th least significant bit of $u-v$ using the above pencil-and-paper  algorithm, the value $b_y$  borrowed from the $(y+1)$th step is $1$. 
 For instance,  $\subcarry(0,110,101)$ is true, $\subcarry(1,110,101)$ is false and $\subcarry(2,110,101)$ is also false. 

 Argue in $\thr$. Our main claim is
\begin{equation}\label{trt3s}
y\leq |u|\mli \bigl(\subcarry(y,u,v)\add\gneg\subcarry(y,u,v)\bigr)\mlc\bigl(\bit(y,u\ominus v)\add \gneg \bit(y,u\ominus v)\bigr),
\end{equation}
which we justify by Induction on $y$. For the basis 
\[\bigl(\subcarry(0,u,v)\add\gneg\subcarry(0,u,v)\bigr)\mlc\bigl(\bit(0,u\ominus v)\add \gneg \bit(0,u\ominus v)\bigr),\] using Fact \ref{tri}, we figure out whether $u\geq v$ of not. If not, we choose $\gneg\subcarry(0,u,v)$ and $\gneg \bit(0,u\ominus v)$. Now assume $u\geq v$. 
Using the Bit axiom, we determine the truth status of $\bit(0,u)$ and $\bit(0,v)$. If $\bit(0,u)\leftrightarrow \bit(0,v)$, we  choose
 $\gneg \subcarry(0,u,v)$ and $\gneg \bit(0,u\ominus v)$; if $\bit(0,u)\mlc \gneg\bit(0,v)$, we  choose
 $\gneg \subcarry(0,u,v)$ and $\bit(0,u\ominus v)$; and if $\gneg\bit(0,u)\mlc \bit(0,v)$, we  choose
 $\subcarry(0,u,v)$ and $\bit(0,u\ominus v)$.

The inductive step is 
\begin{equation}\label{trts}
\begin{array}{c}
\bigl(\subcarry(y,u,v)\add\gneg\subcarry(y,u,v)\bigr)\mlc \bigl(\bit(y,u\ominus v)\add \gneg \bit(y,u\ominus v)\bigr)\mli \\
\bigl( \subcarry(y\successor,u,v)\add\gneg\subcarry(y\successor,u,v)\bigr)\mlc \bigl(\bit(y\successor,u\ominus v)\add \gneg \bit(y\successor,u\ominus v)\bigr) .
\end{array}
\end{equation}
The above is obviously solved by the following strategy.  Using Fact \ref{tri}, we figure out whether $u\geq v$ or not. If not, we choose $\gneg\subcarry(y\successor,u,v)$ and 
$\gneg \bit(y\successor,u\ominus v)$ in the consequent of (\ref{trts}). Now assume   $u\geq v$.  We wait till the adversary tells us, in the antecedent, whether $\subcarry(y,u,v)$ is true.
Using the Bit axiom in combination with the Successor axiom, we also figure out whether $\bit(y\successor,u)$ and  $\bit(y\successor,v)$ are true. If we have $\subcarry(y,u,v)\mlc \bit(y\successor,v) $ or 
$\gneg\bit(y\successor,u)\mlc \bigl(\subcarry(y,u,v)\mld \bit(y\successor,v)\bigr)$, then we choose 
 $\subcarry(y\successor,u,v)$ in the consequent of (\ref{trts}), otherwise we choose   $\gneg\subcarry(y\successor,u,v)$. Also, 
if $\bit(y\successor,u)\leftrightarrow \bigl(\subcarry(y,u,v)\leftrightarrow \bit(y\successor,v)\bigr)$,
 we choose $\bit(y\successor,u\ominus v)$ in the consequent of (\ref{trts}), otherwise we choose $\gneg \bit(y\successor,u\ominus v)$. 

(\ref{trt3s}) is proven. It obviously implies $
y< |u| \mli \bit(y,u\ominus v)\add \gneg \bit(y,u\ominus v)$, from which, 
 by Reasonable Comprehension, we get 
\begin{equation}\label{io2s}
\ade |z|\leq |u|\cla y< |u|\bigl(\bit(y,z)\leftrightarrow  \bit(y,u\ominus v)\bigr).
\end{equation}

The following is a true (by $\pa$) sentence:
\begin{equation}\label{io3s}
\cla u\cla v\cla |z|\leq |u|\Bigl(\cla y< |u|\bigl(\bit(y,z)\leftrightarrow \bit(y,u\ominus v)\bigr)\mli z=u\ominus v\Bigr).
\end{equation}
Now, the target $\ade z(z=u\ominus v)$ is a logical consequence of (\ref{io2s}) and (\ref{io3s}). 
\end{proof}

\subsection{Bit replacement}
Let $\br_0(x,s)$ (resp. $\br_1(x,s)$)\label{xbrhg} be a natural pterm for the function that, on arguments $x$ and $s$, returns the number whose binary representation is obtained from that of $s$ by replacing the $x$th least significant bit $(s)_x$ by $0$ (resp. by $1$).    

\begin{fact}\label{br}
For either  $i\in\{0,1\}$, $\thr\vdash x< |s|\mli \ade z\bigl(z=\br_i(x,s)\bigr)$.  
\end{fact}

\begin{proof} Consider either $i\in\{0,1\}$.  Arguing in $\thr$, we claim that
\begin{equation}\label{uuu0}
\bit\bigl(y,\br_i(x,s)\bigr)\add\gneg \bit\bigl(y,\br_i(x,s)\bigr).
\end{equation}
This is our strategy for  (\ref{uuu0}). Using Fact \ref{tri}, we figure out whether $y=x$ or not. If $y=x$, we choose the left  $\add$-disjunct of   (\ref{uuu0}) if $i$ is $1$, and choose the right  $\add$-disjunct if $i$ is $0$. Now suppose $y\not=x$. In this case, using the Bit axiom,  we figure out whether $\bit(y,s)$ is true or not. If it is true, we choose the left $\add$-disjunct in (\ref{uuu0}), otherwise we choose the right $\add$-disjunct.  It is not hard to see that, this way, we win.

From (\ref{uuu0}), by Comprehension, we get
\begin{equation}\label{uuu1}
\ade |z|\leq |s|\cla y< |s|\Bigl(\bit(y,z)\leftrightarrow  \bit\bigl(y,\br_i(x,s)\bigr)\Bigr).
\end{equation}
From $\pa$, it can also be seen that the following sentence is true: 
\begin{equation}\label{uuu2}
\cla s\cla x<|s| \cla |z| \leq |s|\Bigl[\cla y< 
|s|\Bigl(\bit(y,z)\leftrightarrow  \bit\bigl(y,\br_i(x,s)\bigr)\Bigr)\mli z=\br_i(x,s)\Bigr].
\end{equation}
Now, the target $x< |s|\mli \ade z\bigl(z=\br_i(x,s)\bigr)$ is a logical consequence of (\ref{uuu1}) and (\ref{uuu2}).
\end{proof}

\subsection{Multiplication}
In what follows, $\lfloor u/2\rfloor$\label{xslja} is a pterm for the function that, for a given number $u$, returns the number whose binary representation is obtained from that of $u$ by deleting the least significant bit if 
such a bit exists (i.e., if $u\not=0$), and returns $0$ otherwise. 

\begin{lemma}\label{div2}
$\thr\vdash  \ade z(z=\lfloor u/2\rfloor)$.
\end{lemma}

\begin{proof} Argue in $\thr$. We first claim that 
\begin{equation}\label{sep1}
\bit(y,\lfloor u/2\rfloor)\add\gneg\bit(y,\lfloor u/2\rfloor).
\end{equation}
To win (\ref{sep1}), we compute the value $a$ of $y\successor$ using the Successor axiom.  Next, using the Bit axiom, we figure out whether the $a$th least significant bit of $u$ is $1$ or $0$. If it is $1$, we choose the left $\add$-disjunct of (\ref{sep1}), otherwise choose the right $\add$-disjunct.

From (\ref{sep1}), by Comprehension, we get
\begin{equation}\label{sep1a}
\ade |z|\leq |u|\cla y<|u|\bigl(\bit(y,z)\leftrightarrow \bit(y,\lfloor u/2\rfloor)\bigr).
\end{equation}   
From $\pa$, we also know that 
\begin{equation}\label{odsh}
\cla u\cla |z|\leq |u|\Bigl(\cla y<|u|\bigl(\bit(y,z)\leftrightarrow \bit(y,\lfloor u/2\rfloor)\bigr)\mli z=\lfloor u/2\rfloor\Bigr).
\end{equation}
Now, the target $\ade z(z=\lfloor u/2\rfloor)$ is a logical consequence of (\ref{sep1a}) and (\ref{odsh}).
\end{proof}

In what follows, $\bitsum(x,y,u,v)$\label{xbitsum} is (a pterm for) the function  
\[
(u)_0\times (v)_{y\ominus 0}+(u)_1\times (v)_{y\ominus 1}+(u)_2\times (v)_{y\ominus 2}+\ldots+(u)_{\min(x,y)}\times 
(v)_{y\ominus x}
\] 
(here, of course, $\min(x,y)$\label{xmin}  means the smaller of $y,x$).

Take a note of the following obvious facts:
\begin{eqnarray}
&  \pa\vdash \cla\bigl(\bitsum(x,y,u,v)\leq |u|\bigr). & \label{jhal}\\
& \pa\vdash \cla\bigl(x\geq y\mli \bitsum(x\successor,y,u,v)=\bitsum (x,y,u,v)\bigr). & \label{str0}\\
& \pa\vdash \cla\bigl(x>|u|\mli \bitsum(x,y,u,v)=\bitsum (|u|,y,u,v)\bigr). & \label{str1}
\end{eqnarray}

\begin{lemma}\label{bitsum}
$\thr\vdash  \ade z\bigl(z=\bitsum(x,y,u,v)\bigr)$.
\end{lemma}

\begin{proof} 
Argue in $\thr$. By Induction on $x$, we want to show that 
\begin{equation}\label{sep1b}
x\leq |u|\mli \ade |z|\leq ||u||\bigl(z=\bitsum(x,y,u,v)\bigr).
\end{equation}
Here and later in similar cases, as expected, ``$||u||$'' is not any sort of new notation, it simply stands for ``$|(|u|)|$''. Note that the consequent of the above formula is logarithmically  bounded (namely, the bound for $\ade$ is $|u|$, unlike the linear bound $u$ used in the antecedent) and hence, in view of clause 2 of Definition \ref{dt}, is guaranteed to be ${\mathcal R}\spa$-bounded as  required by the  conditions of $\mathcal R$-Induction. 

The basis $\ade |z|\leq ||u||\bigl(z=\bitsum(0,y,u,v)\bigr)$ is solved by choosing, for $z$, the constant 
 $b$ with  $b=(u)_0\times (v)_y$. Here our writing ``$\times$'' should not suggest that we are  relying on the system's (not yet proven) knowledge of how to compute multiplication. Rather, $(u)_0\times (v)_y$ has a simple propositional-combinatorial   meaning: it means $1$ if both $\bit(0,u)$ and $\bit(y,v)$ are true, and means $0$ otherwise. So, $b$ can be computed by just using the Bit axiom twice and then, if $b$ is $1$, further using Fact \ref{numerals}. 

The inductive step is 
\begin{equation}\label{noww}
\ade |z|\leq ||u||\bigl(z=\bitsum(x,y,u,v)\bigr)\mli \ade |z|\leq ||u||\bigl(z=\bitsum(x\successor,y,u,v)\bigr). 
\end{equation}
To solve the above, we wait till Environment chooses a constant $a$ for $z$ in the antecedent. 
After that, using Fact \ref{tri}, we figure out whether $x< y$.  If not, we choose $a$ for $z$ in the consequent and, in view of (\ref{str0}), win.   Now  suppose $x< y$.
With the help of  the Successor axiom, Bit axiom, Fact \ref{minus} and perhaps also Fact \ref{numerals}, we find the constant $b$ with $b=(u)_{x'}\times (v)_{y\ominus x'}$.  
 Then, using Fact \ref{plus}, we find the constant $c$ with $c=a+b$, and  specify $z$ as $c$ in the consequent. With some basic knowledge from $\pa$ including (\ref{jhal}), our strategy can be seen to win (\ref{noww}).

Now, to solve the target  $\ade z\bigl(z=\bitsum(x,y,u,v)\bigr)$, we do the following. We first wait till Environment specifies values $x_0,y_0,u_0,v_0$ for the (implicitly $\ada$-bound) variables $x,y,u,v$, thus bringing the game down to $\ade z\bigl(z=\bitsum(x_0,y_0,u_0,v_0)\bigr)$. (Ordinarily, such a step would be omitted in an informal argument and we would simply use 
$x,y,u,v$ to denote the constants chosen by Environment for these variables; but we are being more cautious in the present case.)  
Now, using the Log axiom, we find the value $c_0$ of $|u_0|$ and then, using Fact \ref{tri}, we figure out the truth status of 
$x_0\leq c_0$. If it is true, then, choosing $x_0,y_0,u_0,v_0$ for the free variables $x,y,u,v$ of (\ref{sep1b}), we force the provider of 
(\ref{sep1b}) to choose a constant $d$ for $z$ such that $d=\bitsum(x_0,y_0,u_0,v_0)$ is true. We select that very constant $d$ for 
$z$ in $\ade z\bigl(z=\bitsum(x_0,y_0,u_0,v_0)\bigr)$, and celebrate victory. Now suppose $x_0\leq c_0$ is false. We do exactly the same as in the preceding case, with the only difference that we choose $c_0,y_0,u_0,v_0$ (rather than $x_0,y_0,u_0,v_0)$ for the free variables $x,y,u,v$ of (\ref{sep1b}). In view of (\ref{str1}), we win.  
\end{proof}

\begin{fact}\label{mult}
$\thr\vdash  \ade z(z=u\times v)$.
\end{fact}

\begin{proof} The pencil-and-paper algorithm for multiplying binary numbers, which creates a picture like the following one, is also well known: 

\[\begin{array}{r}
11011\\
\mbox{\raisebox{7pt}{$\times$}}\hspace{15pt} 101\\
----\\
11011\\
+\hspace{10pt}000000\\
1101100\\
-----\\
10000111
\end{array}\] 
One way to describe it is as follows. The algorithm constructs the value $z$ of the product $u\times v$ bit by bit, in 
the right-to-left order. At any step $y>0$ we  have a carry $c_{i-1}$  from the preceding step $y-1$ (unlike the carries 
that emerge in the addition algorithm, here the carry can be greater than $1$). For  step $0$,  the ``carry'' $c_{-1}$ 
from the non-existing ``preceding step'' $\#-1$ is stipulated to be $0$. At each step $y=0,1,2,\ldots$, we first find the 
sum $t_y=\bitsum(y,y,u,v)+c_{y-1}$. Then we declare $(z)_y$ to be $0$ (resp. $1$) if $t_y$ is even (resp. odd); and we declare $c_y$ to be  $\lfloor t_y/2\rfloor$.

Let $\mcarry(y,u,v)$\label{xsmacr} be a natural pterm for ``the carry $c_y$ that we get at step $y\geq 0$ when computing $u\times v$''. Take a note of the following $\pa$-provable fact: 
\begin{equation}\label{stasta}
\cla\bigl(\mcarry(y,u,v)\leq |u|\bigr).
\end{equation} 

Arguing in $\thr$, we  claim that 
\begin{equation}\label{mn}
 y \leq |u|+|v|\mli \ade |w|\leq ||u|| \bigl(\mcarry(y,u,v)=w\bigr)\mlc \bigl(\bit(y,u\times v)\add\gneg\bit(y,u\times v)\bigr).
\end{equation}

This claim can be proven by Induction on $y$. The basis is 
\begin{equation}\label{tho}
 \ade |w|\leq ||u|| \bigl(\mcarry(0,u,v)=w\bigr)\mlc \bigl(\bit(0,u\times v)\add\gneg\bit(0,u\times v)\bigr). \end{equation}
Our strategy for (\ref{tho}) is as follows. Using Lemma \ref{bitsum}, we  compute the value $a$ of $\bitsum(0,0,u,v)$.
 Then, using Lemma \ref{div2}, we  compute the  value $b$ of $\lfloor a/2\rfloor$. After that, we choose $b$ for $w$ in the left conjunct  of (\ref{tho}). Also, using the Bit axiom, we figure out whether $\bit(0,a)$ is true.  If yes, if  we choose $\bit(0,u\times v)$ in the right conjunct of (\ref{tho}), otherwise we choose $\gneg\bit(0,u\times v)$. With some basic knowledge from $\pa$ including (\ref{stasta}), we can see that victory is guaranteed.

The inductive step is 
\begin{equation}\label{mna}
\begin{array}{c}
 \ade |w|\leq ||u||  \bigl(\mcarry(y,u,v)=w\bigr)\mlc \bigl(\bit(y,u\times v)\add\gneg\bit(y,u\times v)\bigr)\mli \\
\ade |w|\leq ||u||  \bigl(\mcarry(y\successor,u,v)=w\bigr)\mlc \bigl(\bit(y\successor,u\times v)
\add\gneg\bit(y\successor,u\times v)\bigr).
\end{array}
\end{equation}
Here is our strategy for (\ref{mna}). We wait till, in the antecedent, the adversary tells us the carry $a=\mcarry(y,u,v)$ from the $y$th  step. Using the Successor axiom, we also find the value $b$ of $y\successor$. Then, using Lemma  \ref{bitsum}, we  compute the value $c$ of $\bitsum(b,b,u,v)$. Then, using Fact \ref{plus}, we  compute the value $d$ of $a+c$. Then, using Lemma  \ref{div2}, we compute  
the value $e$ of $\lfloor d/2\rfloor$. Now, we choose $e$ for $w$ in the consequent of (\ref{mna}).  Also, using the Bit axiom, we figure out whether $\bit(0,d)$ is true. If true, we choose $\bit(y\successor,u\times v)$ in the consequent of (\ref{mna}), otherwise we choose $\gneg\bit(y\successor,u\times v)$.  Again, with some basic knowledge from $\pa$ including (\ref{stasta}), we can see that victory is guaranteed.

The following formula is a logical consequence of   (\ref{mn}) and the $\pa$-provable fact 
$\cla(y < |u|+ |v|+\hat{1}\mli y \leq |u|+ |v| )$:
\begin{equation}\label{mnb}
 y < |u|+ |v|+\hat{1}\mli \bit(y,u\times v) \add\gneg \bit(y,u\times v) .
\end{equation}
From (\ref{mnb}), by Reasonable Comprehension, we get 
\begin{equation}\label{mnbj}
 \ade |z|\leq |u|+ |v|+\hat{1}\cla y < |u|+ |v|+\hat{1}\bigl(\bit(y,z)\leftrightarrow \bit(y,u\times v)\bigr) .
\end{equation}
By $\pa$, we also have
\begin{equation}\label{anany}
 \cla \Bigl( |z|\leq |u|+ |v|+\hat{1}\mlc\cla  y < |u|+ |v|+\hat{1} \bigl(\bit(y,z)\leftrightarrow \bit(y,u\times v)\bigr)\mli z=u\times v\Bigr) .
\end{equation}

Now, the target  $\ade z(z=u\times v)$ is a logical consequence of   (\ref{mnbj}) and (\ref{anany}). 
\end{proof}

\section{Some instances of \texorpdfstring{{\bf CLA11}}{CLA11}}\label{sharv}

In this section we are going to see an infinite yet incomplete  series of natural theories that are regular and thus adequate 
(sound and complete) 
in the sense of Theorem \ref{tt1}. 
 All these theories look like $\areleven_{\emptyset}^{\mathcal R}$, with the subscript $\emptyset$ indicating that there are no supplementary axioms. 

Given a set $S$ of bounds, by $S^\heartsuit$ (resp. $S^\spadesuit$)\label{xmasti} we shall denote the linear (resp. polynomial) closure of $S$. 

\begin{lemma}\label{uy}
Consider any regular boundclass triple $\mathcal R$, and any set $S$ of bounds. Assume that, for every pterm $\mathfrak{p}(\vec{x})\in S$, 
we have  $\areleven_{\emptyset}^{\mathcal R}\vdash  \ade z(z=\mathfrak{p}|\vec{x}|)$ 
for some (=any) variable $z$ not occurring in $\mathfrak{p}$. Then the same holds for $ S^\spadesuit$ --- and hence also $S^\heartsuit$  --- instead of $S$. 
\end{lemma}

\begin{proof} Straightforward (meta)induction on the complexity of pterms, relying on the Successor axiom, Fact  \ref{plus} and Fact  \ref{mult}.
\end{proof}

\begin{lemma}\label{ed}
Consider any regular boundclass triple $\mathcal R$, any pterms $\mathfrak{p}(\vec{x})$ and $\mathfrak{a}(\vec{x})$, and any variable $z$ not occurring in these pterms. Assume $\mathfrak{a}(\vec{x})$ is in 
${\mathcal R}\amp$, and $\areleven_{\emptyset}^{\mathcal R}$ proves the following two sentences:
\begin{eqnarray}\label{yy}
& \ada \ade z\bigl(z=\mathfrak{p}(\vec{x})\bigr); & \label{yy1}\\
& \cla\bigl(\mathfrak{p}(\vec{x})\leq\mathfrak{a}(\vec{x})\bigr) .& \label{yy2}
\end{eqnarray} 
 Then $\areleven_{\emptyset}^{\mathcal R}$ also proves 
$\ada \ade z(z=2^{\mathfrak{p}|\vec{x}|})$.
\end{lemma}

\begin{proof} Assume the conditions of the lemma, and  argue in $\areleven_{\emptyset}^{\mathcal R}$.  We claim that 
\begin{equation}\label{poi}
\bit(y,2^{\mathfrak{p}|\vec{x}|})\add\gneg\bit(y,2^{\mathfrak{p}|\vec{x}|}).
\end{equation}
Our strategy for (\ref{poi}) is as follows. Using the Log axiom, we compute the values $\vec{c}$ of $|\vec{x}|$. Then, relying on (\ref{yy1}),  we  find the value $a$ of $\mathfrak{p}(\vec{c})$.  From $\pa$, we know that the $a$th least significant bit of $2^a$ --- and only that bit --- is a $1$. So, 
 using Fact \ref{tri}, we compare $a$ with $y$. If $a=y$, we choose $\bit(y,2^{\mathfrak{p}|\vec{x}|})$ in (\ref{poi}), otherwise choose 
$\gneg\bit(y,2^{\mathfrak{p}|\vec{x}|})$. 

From (\ref{poi}), by Comprehension, we get 
\[\ade |z|\leq (\mathfrak{a}|\vec{x}|)\successor\cla y<(\mathfrak{a}|\vec{x}|)\successor\bigl(\bit(y,z)\leftrightarrow \bit(y,2^{\mathfrak{p}|\vec{x}|})\bigr). \]
The above, in view of the $\pa$-provable fact $|2^{\mathfrak{a}|\vec{x}|}|= (\mathfrak{a}|\vec{x}|)\successor$,  implies 
\begin{equation}\label{pod}
\ade |z|\leq |2^{\mathfrak{a}|\vec{x}|}|\cla y<|2^{\mathfrak{a}|\vec{x}|}|\bigl(\bit(y,z)\leftrightarrow \bit(y,2^{\mathfrak{p}|\vec{x}|})\bigr).  
\end{equation}
Obviously, from $\pa$ and (\ref{yy2}), we  also have 
\begin{equation}\label{ioo}
\cla \Bigl(|z|\leq |2^{\mathfrak{a}|\vec{x}|}|\mlc \cla y<|2^{\mathfrak{a}|\vec{x}|}|\bigl(\bit(y,z)\leftrightarrow \bit(y,2^{\mathfrak{p}|\vec{x}|})\bigr)\mli z=2^{\mathfrak{p}|\vec{x}|} \Bigr).
\end{equation}
Now, the target  $\ade z(z=2^{\mathfrak{p}|\vec{x}|})$ is a logical consequence of (\ref{pod}) and (\ref{ioo}). 
\end{proof}

Here we define the following series ${\mathcal B}_{1}^{1},{\mathcal B}_{1}^{2},{\mathcal B}_{1}^{3},\ldots,{\mathcal B}_{2},{\mathcal B}_{3},{\mathcal B}_{4},{\mathcal B}_{5},{\mathcal B}_{6},{\mathcal B}_{7},{\mathcal B}_{8}$ of sets of terms:
\begin{enumerate}
  \item \begin{enumerate}  \item ${\mathcal B}_{1}^{1}\ = \ \{|x|\}^\heartsuit$ ({\bf logarithmic} boundclass); \item ${\mathcal B}_{1}^{2}\ =\ \{|x|^2\}^\heartsuit$; \item ${\mathcal B}_{1}^{3}\ = \ \{|x|^3\}^\heartsuit$; \item \ldots;\end{enumerate}
  \item ${\mathcal B}_{2} \ =\ \{|x|\}^\spadesuit$ ({\bf polylogarithmic} boundclass);
  \item ${\mathcal B}_{3}\ =\ \{x\}^\heartsuit$ ({\bf linear} boundclass);
  \item ${\mathcal B}_{4}\ =\ \{x\times |x|, \ x\times |x|^2, \ x\times |x|^3,\ \ldots\}^\heartsuit$ ({\bf quasilinear} boundclass);
  \item ${\mathcal B}_{5}\ =\ \{x\}^\spadesuit$ ({\bf polynomial} boundclass);
\item ${\mathcal B}_{6}\ =\ \{2^{|x|},\ 2^{|x|^2},\ 2^{|x|^3},\ \ldots\}^\spadesuit$ ({\bf quasipolynomial} boundclass);
\item ${\mathcal B}_{7}\ =\ \{2^x\}^\spadesuit$ ({\bf exponential-with-linear-exponent} boundclass);
\item ${\mathcal B}_{8}\ =\ \{2^{x},\ 2^{x^2},\ 2^{x^3},\ \ldots\}^\spadesuit$ ({\bf exponential-with-polynomial-exponent} boundclass).
\end{enumerate}

Note that all elements of any of the above sets are bounds, i.e., monotone pterms. Further, since all sets have the form $S^\heartsuit$ or  $S^\spadesuit$, they are (indeed) boundclasses, i.e., are closed under syntactic variation. 

\begin{fact}\label{dds}
For any boundclass triple $\mathcal R$ listed below, the theory $\areleven_{\emptyset}^{\mathcal R}$ is regular:

\noindent $({\mathcal B}_3,{\mathcal B}_{1}^{1},{\mathcal B}_5)$;  $({\mathcal B}_3,{\mathcal B}_{1}^{2},{\mathcal B}_5)$; $({\mathcal B}_3,{\mathcal B}_{1}^{3},{\mathcal B}_5)$; \ldots; $({\mathcal B}_3,{\mathcal B}_{2},{\mathcal B}_5)$;  $({\mathcal B}_3,{\mathcal B}_{2},{\mathcal B}_6)$; $({\mathcal B}_3,{\mathcal B}_{2},{\mathcal B}_7)$;
$({\mathcal B}_3,{\mathcal B}_{3},{\mathcal B}_5)$; $({\mathcal B}_3,{\mathcal B}_{3},{\mathcal B}_6)$; $({\mathcal B}_3,{\mathcal B}_{3},{\mathcal B}_7)$;
$({\mathcal B}_4,{\mathcal B}_{1}^{1},{\mathcal B}_5)$; $({\mathcal B}_4,{\mathcal B}_{1}^{2},{\mathcal B}_5)$; $({\mathcal B}_4,{\mathcal B}_{1}^{3},{\mathcal B}_5)$; \ldots; $({\mathcal B}_4,{\mathcal B}_{2},{\mathcal B}_5)$; $({\mathcal B}_4,{\mathcal B}_{2},{\mathcal B}_6)$; $({\mathcal B}_4,{\mathcal B}_{4},{\mathcal B}_5)$; $({\mathcal B}_4,{\mathcal B}_{4},{\mathcal B}_6)$; $({\mathcal B}_4,{\mathcal B}_{4},{\mathcal B}_7)$; 
$({\mathcal B}_5,{\mathcal B}_{1}^{1},{\mathcal B}_5)$; $({\mathcal B}_5,{\mathcal B}_{1}^{2},{\mathcal B}_5)$; $({\mathcal B}_5,{\mathcal B}_{1}^{3},{\mathcal B}_5)$; \ldots; $({\mathcal B}_5,{\mathcal B}_{2},{\mathcal B}_5)$; $({\mathcal B}_5,{\mathcal B}_{2},{\mathcal B}_6)$;
$({\mathcal B}_5,{\mathcal B}_{5},{\mathcal B}_5)$; $({\mathcal B}_5,{\mathcal B}_{5},{\mathcal B}_6)$; $({\mathcal B}_5,{\mathcal B}_{5},{\mathcal B}_7)$; $({\mathcal B}_5,{\mathcal B}_{5},{\mathcal B}_8)$.
\end{fact}

\begin{proof} Let $\mathcal R$ be any one of the above-listed triples. By definition, a theory $\areleven_{\emptyset}^{\mathcal R}$ is regular iff the triple $\mathcal R$ is regular and,  in addition, $\areleven_{\emptyset}^{\mathcal R}$ satisfies the two conditions of Definition
 \ref{dadm}.

To verify that $\mathcal R$ is regular, one has to make sure that all five conditions of Definition \ref{dt} are satisfied by any value of $\mathcal R$ from the list. This is a  rather easy  job. For instance, the satisfaction of condition 3 of Definition \ref{dt}  is automatically 
guaranteed in view of the fact that all of the boundclasses ${\mathcal B}_{1}^{1},\ldots,{\mathcal B}_8$ have the form $S^\heartsuit$ or  $S^\spadesuit$, and the  
${\mathcal R}\tim$  component of each of the listed triples has the form $S^\spadesuit$. We leave a verification of the satisfaction of the other conditions of Definition  \ref{dt} to the reader. 

As for Definition \ref{dadm}, 
 condition 1 of it is trivially satisfied because the set of the supplementary axioms of 
each theory $\areleven_{\emptyset}^{\mathcal R}$  under question is empty. So, it remains to only verify the satisfaction of condition 2. Namely, we shall show that, for every bound $\mathfrak{b}(\vec{x})$ from ${\mathcal R}\amp$, ${\mathcal R}\spa$ or ${\mathcal R}\tim$,  $\areleven_{\emptyset}^{\mathcal R}$ proves $\ade z(z=\mathfrak{b}|\vec{x}|)$. Let us start with ${\mathcal R}\spa$.

Assume ${\mathcal R}\spa={\mathcal B}_{1}^{1}=\{|x|\}^\heartsuit$. In view of Lemma \ref{uy}, in order to show (here and below in similar situations) that $\areleven_{\emptyset}^{\mathcal R}\vdash\ade z(z=\mathfrak{b}|\vec{x}|)$  for every bound $\mathfrak{b}(\vec{x})$ from this boundclass, it is sufficient for us to just show that $\areleven_{\emptyset}^{\mathcal R}\vdash\ade z(z=||x||)$. But this is indeed so: apply the Log axiom to $x$ twice.  

Assume ${\mathcal R}\spa={\mathcal B}_{1}^{2}=\{|x|^2\}^\heartsuit$. Again, in view of Lemma \ref{uy}, it is sufficient for us to show that $\areleven_{\emptyset}^{\mathcal R}$ proves $\ade z(z=||x||^2)$, i.e., $\ade z(z=||x||\times ||x||)$. But this is indeed so: apply the Log axiom to $x$ twice to obtain the value $a$ of $||x||$, and then apply Fact \ref{mult} to compute the value of $a\times a$. 

The cases of  ${\mathcal R}\spa$ being ${\mathcal B}_{1}^{3}$, ${\mathcal B}_{1}^{4}$, \ldots will be handled in a similar way, relying on  Fact \ref{mult} several times rather than just once. 
 
The case of ${\mathcal R}\spa={\mathcal B}_{2}=\{|x|\}^\spadesuit$ will be handled in exactly the same way as we handled 
${\mathcal R}\spa={\mathcal B}_{1}^{1}=\{|x|\}^\heartsuit$.

So will be the case of ${\mathcal R}\spa={\mathcal B}_{3}=\{x\}^\heartsuit$, with the only difference that, the Log axiom needs to be applied only once rather than twice.

Assume ${\mathcal R}\spa={\mathcal B}_{4}=\{x\times |x|,x\times |x|^2,x\times |x|^3,\ldots\}^\heartsuit$. In view of Lemma \ref{uy}, it is sufficient for us to show that, for any $i\geq 1$, $\areleven_{\emptyset}^{\mathcal R}\vdash\ade z(z=|x|\times ||x||^i)$. This provability indeed holds due to the Log axiom (applied twice) and Fact \ref{mult} (applied $i$ times). 

The case of ${\mathcal R}\spa={\mathcal B}_{5}=\{x\}^\spadesuit$ will be handled in exactly the same way as we handled 
${\mathcal R}\spa={\mathcal B}_{3}=\{x\}^\heartsuit$.

Looking back at the triples listed in the present lemma, we see that ${\mathcal R}\spa$ is always one of ${\mathcal B}_{1}^{1},{\mathcal B}_{1}^{2},\ldots$, ${\mathcal B}_{2}$, ${\mathcal B}_{3}$, ${\mathcal B}_{4}$, ${\mathcal B}_{5}$. This means we are done with ${\mathcal R}\spa$. If ${\mathcal R}\amp$ or ${\mathcal R}\tim$ is one of ${\mathcal B}_{1}^{1},{\mathcal B}_{1}^{2},\ldots$, ${\mathcal B}_{2}$, ${\mathcal B}_{3}$, ${\mathcal B}_{4}$, ${\mathcal B}_{5}$, the above argument applies without any changes. In fact, ${\mathcal R}\amp$ is always one of ${\mathcal B}_{3}$, ${\mathcal B}_{4}$, ${\mathcal B}_{5}$, meaning that we are already done with ${\mathcal R}\amp$ as well. So, it only remains to consider ${\mathcal R}\tim$ in the cases where the latter is one of ${\mathcal B}_{6}$, ${\mathcal B}_{7}$, ${\mathcal B}_{8}$.

Assume ${\mathcal R}\tim={\mathcal B}_{6}=\{2^{|x|},2^{|x|^2},2^{|x|^3},\ldots\}^\spadesuit$. In view of Lemma \ref{uy}, it is sufficient for us to show that, for any $i\geq 1$, $\areleven_{\emptyset}^{\mathcal R}\vdash\ade z(z=2^{||x||^i})$. Consider any such $i$. Relying on the Log axiom once and Fact \ref{mult} $i$ times, we find that $\areleven_{\emptyset}^{\mathcal R}\vdash\ade z(z=|x|^i)$. Also, as $\mathcal R$ is a regular boundclass triple,  ${\mathcal R}\amp$ is at least linear, implying that it contains a bound $\mathfrak{a}(x)$  with $\pa\vdash\cla x\bigl(|x|^i\leq \mathfrak{a}(x)\bigr)$. Hence, by Lemma \ref{ed}, $\areleven_{\emptyset}^{\mathcal R}\vdash \ade z(z=2^{||x||^i})$, as desired.

 Assume ${\mathcal R}\tim={\mathcal B}_{7}=\{2^x\}^\spadesuit$. It is sufficient to show that $\areleven_{\emptyset}^{\mathcal R}\vdash\ade z(z=2^{|x|})$. The sentence $\ade z(z=x)$ is logically valid and hence provable in $\areleven_{\emptyset}^{\mathcal R}$. Also, due to being at least linear, ${\mathcal R}\amp$ contains a bound $\mathfrak{a}(x)$  with $\pa\vdash\cla x\bigl(x\leq \mathfrak{a}(x)\bigr)$. Hence, by Lemma \ref{ed}, $\areleven_{\emptyset}^{\mathcal R}\vdash \ade z(z=2^{|x|})$, as desired.

Finally,  assume ${\mathcal R}\tim={\mathcal B}_{8}=\{2^x,2^{x^2},2^{x^3},\ldots\}^\spadesuit$.  It is sufficient for us to show that, for any $i\geq 1$, $\areleven_{\emptyset}^{\mathcal R}\vdash\ade z(z=2^{|x|^i})$. Consider any such $i$. Relying on Fact \ref{mult} $i$ times, we find that $\areleven_{\emptyset}^{\mathcal R}\vdash\ade z(z=x^i)$. Also, 
 ${\mathcal R}\amp$, which in our case --- as seen from the list of triples --- can only be  ${\mathcal B}_{5}=\{x\}^\spadesuit$, contains the bound $x^i$, for which we trivially have $\pa\vdash\cla x(x^i\leq x^i)$. Hence, by Lemma \ref{ed}, $\areleven_{\emptyset}^{\mathcal R}\vdash \ade z(z=2^{|x|^i})$, as desired.
\end{proof}

In view of Theorem
 \ref{tt1}, 
an immediate corollary of Fact \ref{dds} is that, where $\mathcal R$ is any one of the boundclass triples listed in Fact \ref{dds}, the theory $\areleven_{\emptyset}^{\mathcal R}$ (resp. $\areleven_{\emptyset !}^{\mathcal R}$) is extensionally (resp. intensionally) adequate with respect to computability in the corresponding tricomplexity. 
 For instance, $\areleven_{\emptyset}^{({\mathcal B}_3,{\mathcal B}_{2},{\mathcal B}_5)}$ and 
$\areleven_{\emptyset !}^{({\mathcal B}_3,{\mathcal B}_{2},{\mathcal B}_5)}$ are adequate with 
respect to (simultaneously) linear amplitude, polylogarithmic space and polynomial time computability; 
$\areleven_{\emptyset}^{({\mathcal B}_5,{\mathcal B}_{3},{\mathcal B}_8)}$ and 
$\areleven_{\emptyset !}^{({\mathcal B}_5,{\mathcal B}_{3},{\mathcal B}_8)}$ are adequate with respect to polynomial amplitude, 
linear space and exponential time computability; and so on.
 
Fact \ref{dds} was just to somewhat illustrate the scalability and import of Theorem
\ref{tt1}.  There are many meaningful and interesting boundclasses and boundclass triples yielding regular and hence adequate theories yet not mentioned in this section. 

\section{Extensional completeness}\label{s19}
We let $\thr$  be an arbitrary but fixed regular theory. Additionally, we pick and fix an arbitrary arithmetical problem $A$ with an $\mathcal R$ tricomplexity solution.  Proving the extensional completeness of $\thr$ --- i.e., the completeness part of Theorem \ref{tt1}(1)  ---  means showing the existence of a theorem of $\thr$ which, under the standard interpretation $^\dagger$, equals to (``expresses'') $A$. This is what the present section is exclusively devoted to.

\subsection{\texorpdfstring{$X$}{X}, \texorpdfstring{$\mathcal X$}{calX} and \texorpdfstring{$(\mathfrak{a},\mathfrak{s},\mathfrak{t})$}{a,s,t}} \label{xxxxx}
By definition, the above $A$ is an arithmetical problem because, for some sentence $X$, $A=X^\dagger$. For the rest of Section \ref{s19}, we fix such a sentence $X$, and fix $\mathcal X$ as an HPM\label{xHPM} (=strategy)  which solves $X^\dagger$ in $\mathcal R$ tricomplexity. 
In view of Lemma 10.1 of \cite{cl12} and Lemma \ref{lcol},  
 we may and will assume that, as a solution of $X^\dagger$, $\mathcal X$ is provident. 
We further fix three unary bounds $\mathfrak{a}(x)\in{\mathcal R}\amp$, $\mathfrak{s}(x)\in{\mathcal R}\spa$ and $\mathfrak{t}(x)\in{\mathcal R}\tim$ such that $\mathcal X$ is an $(\mathfrak{a},\mathfrak{s},\mathfrak{t})$ tricomplexity solution of $X^\dagger$. 
In view of conditions 2, 3 and 5 of  Definition \ref{dt}, we may and will assume that the following sentence is true:
\begin{equation}
 \cla x\bigl(x\leq \mathfrak{a}(x)\mlc |\mathfrak{t}(x)|\leq \mathfrak{s}(x)\leq \mathfrak{a}(x)\leq \mathfrak{t}(x)\bigr). \label{aprili20}
\end{equation}

$X$ is not necessarily provable in $\thr$, and our goal is to construct another sentence $\overline{X}$ so that  $A=\overline{X}^\dagger$ and so that $\overline{X}$ is guaranteed to be provable in $\thr$.  

Following our earlier conventions,  more often than not we will drop the superscript $^\dagger$ applied 
to (para)for\-mulas, writing $F^\dagger$ simply as $F$. 

We also agree that, throughout the present section, unless otherwise suggested by the context, different metavariables $x,y,z,s,s_1,\ldots$ stand for different variables of the language of $\thr$.

\subsection{Preliminary insights}\label{gggg}
It might be worthwhile to try to get some preliminary insights into the basic idea behind our extensional completeness proof before going into its details.
Let us consider a simple special case  where $X$ is $\ada s\ade y\hspace{1pt}p(s,y)$ for some elementary formula $p(s,y)$. 

The assertion ``$\mathcal X$ is an $(\mathfrak{a},\mathfrak{s},\mathfrak{t})$ tricomplexity solution of $X$'' 
can be formalized in the language of $\pa$ as a certain sentence $\mathbb{W}$. 
Then we let the earlier mentioned  $\overline{X}$ be the sentence 
$\ada s\ade y \bigl(\mathbb{W}\mli p(s,y)\bigr)$.
Since $\mathbb{W}$ is true,  $\mathbb{W}\mli p(s,y)$ is equivalent to $p(s,y)$. This means that $\overline{X}$ and $X$, as games, are the same --- that is,  $\overline{X}^\dagger=X^\dagger$. It now remains to understand why $\thr\vdash \overline{X}$.  Let us agree to write ``${\mathcal X}(s)$'' as an abbreviation of the phrase ``${\mathcal X}$ in the scenario where, at the very beginning of the play, $\mathcal X$'s adversary made the move $\#s$, and made no other moves afterwards''. Argue in $\thr$.

A central lemma, proven by ${\mathcal R}$-induction in turn relying on the results of Section \ref{sboot},  is one establishing that  the work of $\mathcal X$ is provably ``{\em traceable}''. A simplest version of this lemma applied to our present case would look like 
\begin{equation}\label{insig}
t\leq \mathfrak{t}|s|\mli \ade |v|\leq \mathfrak{s}|s| \mbox{\em Config}(s,t,v),
\end{equation}
where $\mbox{\em Config}(s,t,v)$ is an elementary formula asserting that $v$ is a partial description of the $t$'th configuration of ${\mathcal X}(s)$. Here $v$ is not a full description as it omits certain information.   Namely, $v$ does not include the contents of $\mathcal X$'s buffer and run tape, because this could make $|v|$ bigger than the allowed $\mathfrak{s}|s|$;  on the other hand, $v$ includes all other  information necessary for finding a similar partial description of the next configuration, such as scanning head locations or work-tape contents.  

Tracing the work of ${\mathcal X}(s)$ up to its $(\mathfrak{t}|s|)$th step in the style of (\ref{insig}),  
one of the following two eventual scenarios will be observed:
\begin{eqnarray}
&   \mbox{\em ``${\mathcal X}(s)$  does something wrong'';} &\label{ins3}\\
& \mbox{\em $\gneg(\ref{ins3})\ \mlc$ ``at some point, ${\mathcal X}(s)$ makes the move $\#c$ for some constant $c$''.} &\label{ins2}
\end{eqnarray} 
Here ``${\mathcal X}(s)$  does something wrong'' is an assertion that ${\mathcal X}(s)$ makes an illegal move, or  makes an oversized (exceeding $\mathfrak{a}|s|$) move, or consumes too much (exceeding $\mathfrak{s}|s|$) work-tape space, or  makes no moves at all, etc. --- any observable fact that contradicts $\mathbb{W}$. As an aside, why do we consider ${\mathcal X}(s)$'s not making any moves as ``wrong''? Because it means that ${\mathcal X}(s)$ either loses the game or violates the $\mathfrak{t}$ time bound by making an unseen-by-us move sometime after step $\mathfrak{t}|s|$. 

We will know precisely which of (\ref{ins3}) or (\ref{ins2}) is the case. That is, we will have the resource  
\begin{equation}\label{eqdis}
(\ref{ins3})\add (\ref{ins2}).
\end{equation}

If (\ref{ins3}) is the case, then $\mathcal X$ does not satisfy what $\mathbb{W}$ asserts about it,  so $\mathbb{W}$ is false. 
In this case, 
we can win $\ade y\bigl(\mathbb{W}\mli \hspace{1pt}p(s,y)\bigr)$ by choosing $0$ (or any other constant) for $y$, because the resulting $\mathbb{W}\mli \hspace{1pt}p(s,0)$, having a false antecedent, is true. Thus, as  we have just established, 
\begin{equation}\label{eqdis1}
(\ref{ins3})\mli  \ade y\bigl(\mathbb{W}\mli \hspace{1pt}p(s,y)\bigr).
\end{equation}
Now suppose (\ref{ins2}) is the case. This means that the play of $X$  by ${\mathcal X}(s)$ hits $p(s,c)$. If $\mathbb{W}$ is true and thus $\mathcal X$ is a winning strategy for $X$, then $p(s,c)$ has to be true, because hitting a false parasentence would make $\mathcal X$ lose. Thus, $\mathbb{W}\mli p(s,c)$ is true. If so, we can win $\ade y\hspace{1pt}\bigl(\mathbb{W}\mli p(s,y)\bigr)$ by choosing $c$ for $y$. But how can we obtain $c$? We know that $c$ is on ${\mathcal X}(s)$'s run tape at the $(\mathfrak{t}|s|)$th step. However, as mentioned, the partial description $v$ of the $(\mathfrak{t}|s|)$th configuration that we can obtain from (\ref{insig}) does not include this possibly ``oversized'' constant. It is again the traceability of the work of $\mathcal X$ --- in just a slightly different form from (\ref{insig}) --- that comes in to help. Even though we cannot keep track of the evolving (in ${\mathcal X}$'s buffer) $c$ in its entirety while tracing the work of ${\mathcal X}(s)$ in the style of (\ref{insig}), finding any given bit of $c$ is no problem. And this is sufficient, because our ability to find all particular bits of $c$, due to Comprehension, allows us to assemble the constant $c$ itself.  
In summary, 
we have 
\begin{equation}\label{eqdis2}
(\ref{ins2})\mli  \ade y\bigl(\mathbb{W}\mli \hspace{1pt}p(s,y)\bigr).
\end{equation}

Our target $\overline{X}$ is now a logical consequence of (\ref{eqdis}), (\ref{eqdis1}) and (\ref{eqdis2}).\vspace{5pt}

What we saw above was about the exceptionally simple case of $X=\ada s \ade y\hspace{1pt}p(s,y)$, and the general case is  much more complex, of course. Among other things,  showing the provability of $\overline{X}$ requires a certain metainduction on its complexity. But the idea that we have just tried to explain, with certain adjustments and refinements, still remains at the core of the proof.

\subsection{The sentence \texorpdfstring{$\mathbb{W}$}{W}}
Remember the operation of prefixation from \cite{cl12}. It takes a constant game $G$ together with a legal position $\Phi$ of 
$G$, and returns a constant game $\seq{\Phi}G$. Intuitively, $\seq{\Phi}G$ is the game to which $G$ is brought down by 
the labmoves of $\Phi$. This is an ``extensional'' operation, insensitive with respect to how games are represented/written.  
Below we define an ``intensional'' version $\seq{\cdot}!\cdot$ of prefixation, which differs from its extensional counterpart  
in that, instead of dealing with games, it deals with parasentences. Namely:

Assume $F$ is a parasentence and  $\Phi$ is a legal position of $F$. We define the parasentence $\seq{\Phi}!F$\label{xiprf}  inductively as follows:
\begin{itemize}
  \item $\seq{}!F=F$ (as always, $\seq{}$ means the empty position).
  \item For any nonempty legal position $\seq{\lambda,\Psi}$ of $F$, where $\lambda$ is a labmove and $\Psi$ is a sequence of labmoves: 
  \begin{itemize}  
  \item If $\lambda$ signifies a choice of a component $G_i$ in an occurrence of a subformula $G_0\add G_1$ or $G_0\adc G_1$ of $F$, and $F'$ is the result of replacing that occurrence by $G_i$ in $F$, then $\seq{\lambda,\Psi}! F=\seq{\Psi}F'$.
  \item If $\lambda$ signifies a choice of a constant $c$ for a variable $x$ in an occurrence of a subformula $\ade xG(x)$ or $\ada xG(x)$ of $F$, and $F'$ is the result of replacing that occurrence by $G(c)$ in $F$, then $\seq{\lambda,\Psi}! F=\seq{\Psi}F'$.  
\end{itemize}  
\end{itemize}\smallskip

\noindent For example, $\seq{\oo 1.\#101, \pp 1.0}!\Bigl(E\mlc\ada x\bigl(G(x)\add H(x)\bigr)\Bigr)\ =\ E\mlc G(101)$.

We assume that the reader is sufficiently familiar with G\"{o}del's technique of encoding and arithmetizing. Using that technique,  
we can construct an elementary sentence  $\mathbb{W}_1$\label{xwlong}  which asserts that 
\begin{equation}\label{prct}
\mbox{\em ``$\mathcal X$  is a provident $(\mathfrak{a},\mathfrak{s},\mathfrak{t})$  tricomplexity solution of $X$''.}
\end{equation}
While we are not going to actually construct $\mathbb{W}_1$ here, some clarifications could still be helpful. A brute force attempt to express (\ref{prct}) would have to include the phrase ``for all computation branches of $\mathcal X$''. Yet, there are uncountably many computation branches, and thus they cannot be encoded  through natural numbers. Luckily, this does not present a problem. Instead of considering all computation branches, for our purposes it is sufficient to only consider  $\oo$-legal branches of $\mathcal X$ with finitely many $\oo$-labeled moves.  Call such branches 
{\bf  relevant}.\label{xrelv}  Each branch is fully determined by what moves are made in it by Environment and when. Since  the number of Environment's moves in any relevant branch is finite, all such branches can be listed according to --- and, in a sense,  identified with --- the corresponding finite sequences of Environment's timestamped moves. This means that there are only countably many  relevant branches,  and  they can be encoded with natural numbers. Next, let us say that a parasentence $E$ is 
{\bf relevant}\label{xrpr} 
iff $E=\seq{\Gamma}!X$ for some legal position $\Gamma$ of $X$. In these terms, the formula $\mathbb{W}_1$ can be constructed as a natural arithmetization of the following, expanded, form of (\ref{prct}):
\begin{quote} 
{\em ``$\mathfrak{a},\mathfrak{s},\mathfrak{t}$ are bounds\hspace{1pt}\footnote{I.e., $\mathfrak{a},\mathfrak{s},\mathfrak{t}$ are monotone pterms --- see Section \ref{sbou}. This condition is implicit in (\ref{prct}).}  and, for any relevant computation branch $B$, the following conditions are satisfied:
\begin{enumerate}
\item {\bf ($\mathcal X$ plays $X$ in $(\mathfrak{a},\mathfrak{s},\mathfrak{t})$ tricomplexity):} For any step $c$ of  $B$, where $\ell$ is the background of $c$, we have:
\begin{enumerate}
  \item The spacecost of $c$ does not exceed $\mathfrak{s}(\ell)$;
  \item If $\mathcal X$ makes a move $\alpha$ at step $c$, then the magnitude of $\alpha$ does not exceed $\mathfrak{a}(\ell)$ and the timecost of $\alpha$ does not exceed $\mathfrak{t}(\ell)$.  
  \end{enumerate} 
\item {\bf ($\mathcal X$ wins $X$):} There is a  legal position $\Gamma$ of $X$ and a parasentence $H$ such that $\Gamma$ is the run spelled by $B$, 
$H=\seq{\Gamma}!X$, and the elementarization $\elz{H}$ of $H$ is true.
\item {\bf ($\mathcal X$ plays $X$ providently):} There is an integer $c$ such that, for any $d\geq c$, $\mathcal X$'s buffer at step $d$ of $B$ is empty.'' 
\end{enumerate} } 
\end{quote}
Clause 2 of the above description relies on the predicate ``true'' which, in full generality, by Tarski's theorem, is non-arithmetical.  
However, in the present case, the truth predicate  is limited to the parasentences $\elz{H}$ where $H$ is a relevant parasentence. Due to $H$'s being relevant,  all occurrences of blind quantifiers  in $\elz{H}$ are inherited from $X$. This means that, as long as $X$ is fixed (and, in our case, it is indeed fixed), the $\cla,\cle$-depth of $\elz{H}$ is bounded by a constant. It is well known (cf. \cite{BussChapter}) that limiting the $\cla,\cle$-depths of arithmetical parasentences to any particular value makes the corresponding truth predicate expressible in the language of $\pa$. So, it is clear that constructing $\mathbb{W}_1$ formally does not present a problem.

We now define the sentence $\mathbb{W}$\label{xwwwww} by 
\[\mathbb{W}\ =_{def} \ \mathbb{W}_1\mlc (\ref{aprili20}).\]

 \subsection{The overline notation}\label{sovnot}
A {\bf literal}\label{xliteral} is  $\twg$, $\tlg$, or a (nonlogical) atomic formula  with or without negation $\gneg$. By a {\bf politeral}\label{ipoliteral} of a formula we mean a positive (not in the scope of $\gneg$) occurrence of a literal in it. For instance, the occurrence of $p$, as well as of $\gneg q$ --- but not of $q$ --- is a politeral of $p\mlc\gneg q$.
While a politeral is not merely a literal but a literal $L$  {\em together} with a fixed occurrence, we shall often refer to it just by the name $L$ of the literal, assuming that it is clear from the context which (positive) occurrence of $L$ is meant.

As we remember, our goal is to construct a formula $\overline{X}$ which expresses the same problem as $X$ does and which is provable in $\thr$. Where $E$ is $X$ or any other formula, we let 
$\overline{E}\label{ipver}$
be the result of replacing in $E$ every politeral $L$ by $\mathbb{W}\mli L$.

\begin{lemma}\label{august12a}
For any  formula $E$, including $X$, we have  $E^\dagger=\overline{E}^\dagger$. 
\end{lemma}

\begin{proof} If $E$ is a literal, then, since $\mathbb{W}$ is true, $E $ is equivalent (in the standard model) to    $\mathbb{W}\mli E$, meaning that $E^\dagger=\overline{E}^\dagger$. The phenomenon $E^\dagger=\overline{E}^\dagger$ now automatically extends from literals to all formulas.   
\end{proof}

In view of the above lemma, what now remains to do for the completion of our extensional completeness proof is to show that $\thr\vdash\overline{X}$. The rest of Section \ref{s19} is entirely devoted to this task.

\begin{lemma}\label{jan4d}
For any  formula $E$, $\thr\vdash \mathbb{W} \mld \cla \overline{E}$. 
\end{lemma}

\begin{proof} Induction on the complexity of $E$. The base, which is about the cases where $E$ is a literal $L$, is straightforward, as then 
$\mathbb{W}\mld \cla \overline{E}$ is the classically valid $\mathbb{W}\mld \cla (\mathbb{W}\mli L)$.
If $E$ has the form $H_0\mlc H_1$, $H_0\mld H_1$, $H_0\adc H_1$ or $H_0\add H_1$ then, by the induction hypothesis, $\thr$ proves $\mathbb{W} \mld  \cla \overline{H_0}$ and $\mathbb{W} \mld \cla \overline{H_1}$, from which $ \mathbb{W} \mld  \cla \overline{E}$ follows by LC.   Similarly, if $E$ has the form 
$\cla xH(x)$, $\cle x H(x)$, $\ada xH(x)$ or $\ade xH(x)$, then, by the induction hypothesis,  $\thr$ proves $\mathbb{W} \mld  \cla \overline{H(x)}$, from which  $\mathbb{W} \mld  \cla \overline{E}$ follows by LC.
\end{proof}

\subsection{Configurations}\label{sds}
Let us  fix 
$\mathfrak{y}\label{xyyy}$
 as the number of work tapes of $\mathcal X$, and 
$\mathfrak{d}\label{x5dd}$
 as  the maximum possible number of labmoves in any legal run of $X$ (the {\em depth} of $X$). 

For the rest of Section \ref{s19}, by a {\bf configuration}\label{xconf} we shall mean a description of what intuitively can be thought of as the ``current'' situation at some step of $\mathcal X$. Specifically, such a description consists of the following $7$ pieces of information: 
\begin{enumerate}
  \item  The state of $\mathcal X$.
  \item  A $\mathfrak{y}$-element array of the contents of the corresponding   $\mathfrak{y}$ work tapes of $\mathcal X$.
  \item The content of $\mathcal X$'s buffer.
  \item The content of $\mathcal X$'s run tape. 
  \item A $\mathfrak{y}$-element array of the locations of the corresponding $\mathfrak{y}$ work-tape heads of $\mathcal X$.
  \item The location of the run-tape head of $\mathcal X$. 
  \item The string  that $\mathcal X$  put into its buffer on the transition to the ``current'' configuration from the predecessor configuration;  if there is no predecessor configuration, then such a string is empty.  
\end{enumerate}\smallskip 

\noindent Notice a difference between our present meaning of ``configuration'' (of $\mathcal X$) and the normal meaning of this word as given in \cite{cl12}. Namely, the piece of information from item 7 is not normally part of a configuration, as this information 
is not really necessary in order to be able to find the next configuration.  

It also is important to point out that any possible combination of any possible settings of the above $7$ parameters is considered to be a configuration, regardless of whether such settings can actually be reached in some computation branch of $\mathcal X$ or not. For this reason, we shall use the adjective {\bf reachable}\label{xreaco} to characterize those configurations that can actually be reached. 

 We fix some reasonable encoding of configurations. For technical convenience, we assume that every  configuration has a unique code, and vice versa: every natural number is the code of some unique configuration. With this one-to-one correspondence in mind, we will routinely identify configurations with their codes. Namely, for a number $c$, instead of saying ``the configuration encoded by $c$'', we may simply say ``the configuration $c$''. ``The state of $c$'', or ``$c$'s state'',   will mean the state of the machine $\mathcal X$ in configuration $c$ --- i.e., the $1$st one of the above-listed $7$ components of $c$. Similarly for the other components of a configuration, such as tape or buffer contents and scanning head locations.  

By the {\bf background}\label{xcbgsd} of a configuration $c$ we shall mean the greatest of the magnitudes of the $\oo$-labeled moves on $c$'s run tape, or $0$ if there are no such moves. 

The following definition, along with the earlier fixed constant $\mathfrak{d}$, involves the constants $\mathfrak{m}$ and $\mathfrak{p}$ introduced  later in Section \ref{smp}.

\begin{defi}\label{dcor}
We say that a configuration $c$ is {\bf uncorrupt}\label{xuncog} iff, where $\Gamma$ is the position spelled on $c$'s run tape, $\alpha$ is the string found in $c$'s buffer and $\ell$ is the background of $c$, all of the following conditions are satisfied:
\begin{enumerate}
  \item $\Gamma$ is a legal position of $X$.
  \item $\ell\leq \mathfrak{a}(\ell)\mlc |\mathfrak{t}(\ell)|\leq \mathfrak{s}(\ell)\leq \mathfrak{a}(\ell)\leq \mathfrak{t}(\ell)$. 
  \item  $|\mathfrak{m}|\leq \mathfrak{s}(\ell)$, where $\mathfrak{m}$ is as in (\ref{pp1}).
  \item  $|\mathfrak{d}(\mathfrak{a}(\ell)+\mathfrak{p}+1)+1|\leq \mathfrak{s}(\ell)$, where $\mathfrak{d}$ is as at the beginning of Section \ref{sds}  
and $\mathfrak{p}$ is as in (\ref{pp2}).
 \item The number of non-blank cells  on any one of the work tapes of $c$ does not exceed $\mathfrak{s}(\ell)$.
 \item There is no $\pp$-labeled move in $\Gamma$ whose magnitude exceeds $\mathfrak{a}(\ell)$. 
  \item If $\alpha$ is nonempty, then there is a string  $\beta$ such that $\seq{\Gamma,\pp\alpha\beta}$ is a legal position of $X $ and the magnitude of the move $\alpha\beta$ does not exceed $\mathfrak{a}(\ell)$. 
\end{enumerate} 
As expected, ``{\bf corrupt}''\label{xncog} means ``not uncorrupt''.  If $c$ merely satisfies condition 1 of Definition \ref{dcor}, then we say that $c$ is {\bf semiuncorrupt}.\label{xscrff} 
\end{defi} 

We define the {\bf yield}\label{xyield}  of a semiuncorrupt configuration   $c$ as the game $\seq{\Gamma}! X$, where $\Gamma$ is the position spelled on $c$'s run tape. 

Let $c,d$ be two configurations and $k$ a natural number. We say that $d$ is a {\bf $k$th unadulterated successor}\label{xus} of $c$ iff there is a sequence $a_0,\ldots,a_k$ ($k\geq 0$) of configurations such that $a_0=c$, $a_k=d$ and, for each $i\in\{1,\ldots,k\}$, we have: (1) $a_i$ is a legitimate successor of (possible next configuration immediately after) $a_{i-1}$, and 
(2) $a_i$'s   run tape content is the same as that of $a_{i-1}$. Note that every configuration $c$ has at most one $k$th unadulterated successor. The latter 
is the configuration to which $c$ evolves within $k$ steps/transitions in the scenario where Environment does not move, as long as $\mathcal X$ does not move in that scenario either (otherwise, if $\mathcal X$ moves, $c$ has no $k$th unadulterated successor).   Also note that every configuration $c$ has a $0$th unadulterated successor, which is $c$ itself. 

For simplicity and without loss of generality, we shall assume that the work-tape alphabet of $\mathcal X$ --- for each of its work tapes --- consists of just $0$, $1$ and \blank, and that the leftmost cells of the work tapes never contain a $0$.\footnote{If not, $\mathcal X$ can be easily modified using rather standard techniques so as to satisfy this condition without losing any of the relevant properties  of the old $\mathcal X$. The same can be said about the additional assumptions made in the following paragraph.}  Then, remembering from \cite{cl12} that an HPM never writes a $\blank$ and never moves its head past the leftmost blank cell, the content of a given work tape at any given time can be understood as the bitstring $b_{n-1},\ldots,b_{0}$, where $n$ is the number of  non-blank cells on the tape\footnote{If $n=0$, then the string $b_{n-1},\ldots,b_{0}$ is empty.} and, for each $i\in\{1,\ldots,n\}$,  $b_{n-i}$ is the bit written in the $i$th cell of the tape
(here the cell count starts from $1$, with the $1$st cell being the leftmost cell of the tape). We agree to consider the number represented by such a string --- i.e., the number $b_{n-1}\times 2^{n-1}+b_{n-2}\times 2^{n-2}+\ldots+b_{1}\times 2^{1}+b_{0}\times 2^{0}$ --- to be the code of the corresponding content of the work tape.  As with configurations, we will routinely identify work-tape contents with their codes. 

For further simplicity and again without loss of generality, we assume that, on  any transition, $\mathcal X$ puts at most one symbol into its buffer. We  shall further assume that, on a transition to a move state,  $\mathcal X$ never repositions any of its scanning heads and never modifies the content of any of its work tapes.  

\subsection{The white circle and black circle notations}
For the rest of this paper we agree that, whenever $\tau(z)$ is a unary pterm but we write $\tau(\vec{x})$ or $\tau(x_1,\ldots,x_n)$,  
it is to be understood as an abbreviation of the pterm $\tau\bigl(\max(x_1,\ldots,x_n)\bigr)$. By convention, if $n=0$, $\max(x_1,\ldots,x_n)$ is considered to be $0$. And if we write $\tau|\vec{x}|$, it is to be understood as $\tau(|x_1|,\ldots,|x_n|)$. 

Let $E(\vec{s})$ be a   formula all of whose  free variables are among  $\vec{s}$ (but not necessarily vice versa), and $z$ be a variable not among $\vec{s}$. We 
 will write   \[E^\circ(z,\vec{s})\label{xecr}\] 
to denote an elementary formula whose free variables are $z,\vec{s}$,  and which is a natural arithmetization of the predicate that,  for any constants $a,\vec{c}$ in the roles of $z,\vec{s}$, holds (that is, $E^\circ(a,\vec{c})$ is true) iff $a$ is a reachable uncorrupt configuration whose yield is $E(\vec{c})$ and whose background does not exceed $\max(\vec{c})$.
Further,  we 
 will write   \[E^{\bullet} (z,\vec{s})\label{xecrb}\] 
to denote an elementary formula whose free variables are $z,\vec{s}$, and which  is a natural arithmetization of the predicate that,  for any constants $a,\vec{c}$  in the roles of $z,\vec{s}$, holds iff $E^\circ(a,\vec{c})$ is true and $a$ has a $(\mathfrak{t}|\vec{c}|)$th unadulterated successor. 

Thus, while $E^\circ(a,\vec{c})$ simply says that the formula $E(\vec{c})$ is the yield of the (reachable, uncorrupt and $\leq \max(\vec{c})$-background) configuration $a$, the stronger $E^\bullet(a,\vec{c})$ additionally asserts that such a yield $E(\vec{c})$ is {\em persistent}, in the sense that, unless the adversary moves, $\mathcal X$ does not move --- and hence the yield of $a$ remains the same $E(\vec{c})$ --- for at least  $\mathfrak{t}|\vec{c}|$ steps beginning from $a$.

We say that a formula $E$  is {\bf critical}\label{icritical} iff one of the following conditions is satisfied:
\begin{itemize}
\item $E$ is of the form $G_0\add G_1$ or $\ade y G$;
\item $E$ is of the form $\cla y G$ or $\cle y G$, and $G$ is critical;
\item $E$ is of the form  $G_0\mld G_1$, and both $G_0$ and $G_1$ are critical;
\item $E$ is of the form  $G_0\mlc G_1$, and at least one of $G_0,G_1$ is critical.
\end{itemize}

\begin{lemma}\label{august20b}
Assume $E(\vec{s})$ is a non-critical  formula all of whose  free variables are among  $\vec{s}$. 
Then
\[\pa\vdash \cla\bigl(E^{\bullet}(z,\vec{s})  \mli \elz{\overline{E(\vec{s})}}\bigr).\] 
\end{lemma}

\begin{proof} Assume the conditions of the lemma. Argue in $\pa$. Consider arbitrary $(\cla$) values of $z$ and $\vec{s}$, which we continue writing as $z$ and $\vec{s}$. Suppose, for a contradiction, that $E^{\bullet}(z,\vec{s})$ is true but $\elz{\overline{E(\vec{s})}}$ is false.   The falsity of $\elz{\overline{E(\vec{s})}}$ implies the falsity of $\elz{E(\vec{s})}$. This is so because the only difference between the two formulas is that, wherever the latter has some politeral $L$, the former has $\mathbb{W}\mli L$.  

The truth of $E^{\bullet}(z,\vec{s})$ implies that, at some point of some actual play,  $\mathcal X$ reaches the configuration $z$,  where $z$ is uncorrupt, the yield of $z$ is $E(\vec{s})$, the background of $z$ is at most $\max(\vec{s})$ and, 
 in the scenario where Environment does not move, $\mathcal X$ does not move either for at least $\mathfrak{t}|\vec{s}|$ steps afterwards. If $\mathcal X$ does not move even after $\mathfrak{t}|\vec{s}|$ steps, then it has lost the game, because the eventual position hit by the latter is $E(\vec{s})$ and the elementarization of $E(\vec{s})$  is false (it is not hard to see that every such game is indeed lost). And if $\mathcal X$ does make a move sometime after $\mathfrak{t}|\vec{s}|$ steps, then, as long as $\mathfrak{t}$ is monotone (and if not, $\mathbb{W}$ is false), $\mathcal X$ violates the time complexity bound $\mathfrak{t}$, because the background of that move does not exceed $\max(\vec{s})$  but the timecost is greater than 
 $\mathfrak{t}|\vec{s}|$. In either case we have:
\begin{equation}\label{jan4b}
\mbox{\em $\mathbb{W}$ is false.}
\end{equation}

Consider any non-critical   formula $G$.  By induction on the complexity of $G$, we are going to show that $\elz{\overline {G}}$ is true for any ($\cla$) values of its free variables. Indeed:
\begin{itemize}
\item If $G$ is a literal, then  $\elz{\overline {G}}$ is $\mathbb{W}\mli G$ which, by (\ref{jan4b}), is true. 

\item If $G$ is $H_0\adc H_1$ or $\ada xH(x)$, then $\elz{\overline {G}}$ is $\twg$ and is thus true.

\item $G$ cannot be $H_0\add H_1$ or $\ade xH(x)$, because then it would be critical. 

\item If $G$ is $\cla yH(y)$ or $\cle y H(y)$, then $\elz{\overline {G}}$ is $\cla y\elz{\overline{H(y)}}$ or $\cle y\elz{\overline{H(y)}}$, where $H(y)$ is non-critical. In either case $\elz{\overline{G}}$ is true because, by the induction hypothesis, $\elz{\overline{H(y)}}$ is true for every value of its free variables, including variable $y$.

\item If $G$ is $H_0\mlc H_1$, then both  $H_0$ and $H_1$ are non-critical. Hence, by the induction hypothesis, both $\elz{\overline{H_0}}$ and 
$\elz{\overline{H_1}}$ are  true. Hence so is  $\elz{\overline{H_0}}\mlc\elz{\overline{H_1}}$ which, in turn, is nothing but $\elz{\overline{G}}$. 

\item Finally, if  $G$ is $H_0\mld H_1$, then one of the formulas $H_i$ is non-critical. Hence, by the induction hypothesis, $\elz{\overline{H_i}}$ is  true. Hence so is  $\elz{\overline{H_0}}\mld\elz{\overline{H_1}}$ which, in turn, is nothing but $\elz{\overline{G}}$.
\end{itemize}\smallskip

\noindent Thus, for any non-critical formula $G$, $\elz{\overline {G}}$ is true. This includes the case $G= E(\vec{s}) $  which, however,  contradicts our assumption that $\elz{\overline{E(\vec{s})}}$ is false.
\end{proof}

\begin{lemma}\label{august20a}
Assume $E(\vec{s})$ is a  critical   formula all of whose  free variables are among  $\vec{s}$. Then
\begin{equation}\label{m13a}
\thr\vdash   \cle E^{\bullet} (z,\vec{s})  \mli \cla\overline{E(\vec{s})}.
\end{equation} 
\end{lemma}

\begin{proof} Assume the conditions of the lemma. By induction on complexity, one can easily see that the $\cle$-closure of the elementarization of any critical formula is false. Thus, for whatever ($\cla$) values of $\vec{s}$, $\elz{E(\vec{s})}$ is false. Arguing further as we did in the proof of Lemma \ref{august20b} when deriving (\ref{jan4b}), we find that,  if  $E^{\bullet} (z,\vec{s})$ is true for whatever ($\cle$) values of $z$ and $\vec{s}$,  then $\mathbb{W}$ is false. And this argument can be formalized in $\pa$, so  we have 
$\pa\vdash\cle E^{\bullet} (z,\vec{s})\mli \gneg \mathbb{W}.$
This, together with Lemma \ref{jan4d}, can be easily seen to imply (\ref{m13a}) by LC.   
\end{proof}

\subsection{Titles}\label{smp}
A {\bf paralegal move}\label{xprlmv} means a string $\alpha$ such that, for some (possibly empty) string $\beta$, position $\Phi$  and player 
$\xx\in\{\pp,\oo\}$, $\seq{\Phi,\xx\alpha\beta}$ is a legal position of $X$.  
In other words, a paralegal move is  
 a  prefix of some move of some legal run of $X$. Every paralegal move $\alpha$ we divide into two parts, called 
the {\bf header}\label{xhdr} and the {\bf numer}.\label{xnmrte} Namely, if $\alpha$ does not contain the symbol $\#$, then $\alpha$ is its own header, with the numer being $0$
(i.e., the empty bit string); and if $\alpha$ is of the form $\beta\# c$, then its header is $\beta\#$ and its numer 
is $c$. When we simply say ``{\bf a header}'', it is to be understood as ``the header of some paralegal move''. Note that, unlike numers, there are only finitely many headers. For instance, if $X$ is $\ade x\hspace{1pt}p\mlc \ada y(q\add r)$ where $p,q,r$ are elementary formulas,  then the  headers are  $0.\#$, $1.\#$, $1.0$, $1.1$ and their proper prefixes --- nine strings altogether. 

Given a configuration $x$, by the {\bf title}\label{xtit} of $x$ we shall mean a partial description of $x$ consisting of the following four pieces of information, to which we shall refer as {\bf titular components}:\label{xtitc}
\begin{enumerate}
  \item  $x$'s state.
  \item The header of the move spelled in $x$'s buffer.   
  \item The string  put into the buffer on the transition to $x$ from its predecessor configuration;  if $x$ has no predecessor configuration, then such a string is  empty. 
  \item The list $\xx_1\alpha_1,\ldots,\xx_n\alpha_m$, where $m$ is the total number of labmoves on $x$'s run tape and, for each $i\in\{1,\ldots,m\}$, $\xx_i$ and $\alpha_i$ are the label ($\pp$ or $\oo$) and the header of the $i$th labmove.  
\end{enumerate}  
We say that a title is {\bf buffer-empty}\label{xbet} if its 2nd titular component is the empty string. 

Obviously there are infinitely many titles, yet only finitely many of those are titles of semiuncorrupt configurations. We fix an infinite, recursive list 
\[\mbox{\em Title}_0, \mbox{\em Title}_1, \mbox{\em Title}_2,   \ldots, \mbox{\em Title}_\mathfrak{k}, 
\mbox{\em Title}_{\mathfrak{k}+1},\mbox{\em Title}_{\mathfrak{k}+2} \ldots,\mbox{\em Title}_\mathfrak{m},\mbox{\em Title}_{\mathfrak{m}+1},\mbox{\em Title}_{\mathfrak{m}+2}\ldots\label{x5m}\]
--- together with the natural numbers  $1\leq \mathfrak{k}\leq \mathfrak{m}$ ---  
of all titles without repetitions, where $\mbox{\em Title}_0$ through  $\mbox{\em Title}_{\mathfrak{m}-1}$ 
(and only these titles) are titles of 
 semiuncorrupt configurations, with  $\mbox{\em Title}_0$ through  $\mbox{\em Title}_{\mathfrak{k}-1}$ (and only these titles) being 
buffer-empty titles of semiuncorrupt configurations.
 By the {\bf titular number}\label{xtnb} of a given configuration $c$ we shall mean the number $i$ such that $\mbox{\em Title}_i$ is $c$'s title. 

We may and will assume that, where $\mathfrak{p}$ is the size of the longest header, $\mathfrak{m}$ is as above and $\mathfrak{d}$ is as at the beginning of Section \ref{sds}, 
$\pa$ proves the following sentences:
\begin{eqnarray}
&\mathbb{W}\mli \cla x\bigl(|\hat{\mathfrak{m}}|\leq \mathfrak{s}(x)\bigr);& \label{pp1}\\
& \mathbb{W}\mli\cla x\bigl(|\hat{\mathfrak{d}}(\mathfrak{a}(x)+\hat{\mathfrak{p}}+\hat{1})+\hat{1}|\leq \mathfrak{s}(x)\bigr).& \label{pp2}
\end{eqnarray}
Indeed, if this is not the case,  we can replace $\mathfrak{s}(x)$ with $\mathfrak{s}(x)+\ldots+\mathfrak{s}(x)+\hat{k}$, $\mathfrak{a}(x)$ with $\mathfrak{a}(x)+\ldots+\mathfrak{a}(x)+\hat{k}$ and $\mathfrak{t}(x)$ with $\mathfrak{t}(x)+\ldots+\mathfrak{t}(x)+\hat{k}$, where ``$\mathfrak{s}(x)$'', ``$\mathfrak{a}(x)$'' and ``$\mathfrak{t}(x)$'' are repeated $k$ times, 
 for some sufficiently large $k$. Based on (\ref{aprili20}) and Definition \ref{dt},  one can see that, with these new values of $\mathfrak{a},\mathfrak{s},\mathfrak{t}$  and the corresponding new value of $\mathbb{W}$, (\ref{pp1}) and (\ref{pp2}) become provable while no old relevant properties of the triple are lost, such as $\mathcal X$'s being a provident $(\mathfrak{a},\mathfrak{s},\mathfrak{t})$ tricomplexity solution of $X$,  $(\mathfrak{a},\mathfrak{s},\mathfrak{t})$'s being a member of ${\mathcal R}\amp\times{\mathcal R}\spa\times{\mathcal R}\tim$, or the satisfaction of  (\ref{aprili20}).

\subsection{Further notation}\label{bignot}
Here is a list of additional notational conventions.  Everywhere below:  
$x,u,z,t$ range over natural numbers; 
  $n\in\{0,\ldots, \mathfrak{d}\}$;  
  $\vec{s}$ abbreviates an $n$-tuple $s_1,\ldots,s_n$  of variables ranging over natural numbers; 
  $\vec{v}$ abbreviates a $(2\mathfrak{y}+3)$-tuple $v_1,\ldots,v_{2\mathfrak{y}+3}$ of variables ranging over natural numbers;
   ``$|\vec{v}| \leq\mathfrak{s}|\vec{s}|$'' abbreviates $|v_1|\leq \mathfrak{s}|\vec{s}|\mlc \ldots\mlc |v_{2\mathfrak{y}+3}|\leq \mathfrak{s}|\vec{s}|$;  and 
  ``$\ade |\vec{v}| \leq\mathfrak{s}|\vec{s}| $'' abbreviates $\ade | v_{1}| \leq\mathfrak{s}|\vec{s}|\ldots \ade | v_{2\mathfrak{y}+3}| \leq\mathfrak{s}|\vec{s}|$.  Also, we identify informal statements or predicates with their natural arithmetizations.

\begin{enumerate}[widest=10]  
  \item $\mathbb{N}(x,z)$\label{xnnn} states that configuration $x$ does not have a corrupt $k$th unadulterated successor for any 
$k\leq z$. 
  \item $\mathbb{D}(x,\vec{s},\vec{v})$\label{x41} is a $\mlc$-conjunction of the following statements: 
 \begin{enumerate}  
  \item ``There are exactly $n$ (i.e., as many as the number of variables in $\vec{s}$) labmoves on configuration $x$'s run tape and, for each $i\in\{1,\ldots,n\}$, if the $i$th (lab)move is numeric, then $s_i$ is its numer''.
  \item ``$v_1$ is the location of $x$'s 1st work-tape head, \ldots, $v_{\mathfrak{y}}$ is the location of $x$'s $\mathfrak{y}$th work-tape head''.
  \item $v_{\mathfrak{y}+1}$ is the content of $x$'s  1st work tape, \ldots,  $v_{2\mathfrak{y}}$ is the content of $x$'s $\mathfrak{y}$th work tape''.
  \item ``$v_{2\mathfrak{y}+1}$ is the location of $x$'s run-tape head''.
  \item  ``$v_{2\mathfrak{y}+2}$ is the length of the numer of the move found in $x$'s buffer''.
  \item  ``$v_{2\mathfrak{y}+3}$ is $x$'s titular number, with $v_{2\mathfrak{y}+3}< \hat{\mathfrak{m}}$ (implying that $x$ is semiuncorrupt)''. 
\end{enumerate} 
\item $\mathbb{D}^\epsilon(x,\vec{s},\vec{v})$\label{x42} abbreviates $\mathbb{D}(x,\vec{s},\vec{v})\mlc v_{2\mathfrak{y}+3}<
\hat{\mathfrak{k}}$. 
  \item $\mathbb{D}_{\mbox{\tiny $\sqcup$}}(x,\vec{s})$\label{x43} and $\mathbb{D}^{\epsilon}_{\mbox{\tiny $\sqcup$}}(x,\vec{s})$ abbreviate 
$\ade |\vec{v}|\leq\mathfrak{s}|\vec{s}|\mathbb{D}(x,\vec{s},\vec{v})$ and $\ade |\vec{v}|\leq\mathfrak{s}|\vec{s}|\mathbb{D}^\epsilon(x,\vec{s},\vec{v})$, respectively.
\item $\mathbb{U}(x,t,z,u)$\label{x44} says ``Configuration $t$ is a $u$th  unadulterated successor of configuration $x$, and $u$ is the greatest number not exceeding $z$ such that $x$ has a $u$th  unadulterated successor''.
\item $\mathbb{U}_{\mbox{\tiny $\sqcup$}}^{\vec{s}}(x,t,z)$\label{x45} abbreviates  $\ade |u|\leq \mathfrak{s}|\vec{s}|\mathbb{U}(x,t,z,u)$. 
\item $\mathbb{U}_{\mbox{\tiny $\exists$}}^{\vec{s}}(x,t)$\label{x46} abbreviates  $\cle u\mathbb{U}(x,t,\mathfrak{t}|\vec{s}|,u)$. 
\item  $\mathbb{Q}(\vec{s},z)$\label{x47} abbreviates $\cla x\Bigl[\mathbb{D}_{\mbox{\tiny $\sqcup$}}(x,\vec{s})\mli  \gneg \mathbb{N}(x,z)\add\Bigl(\mathbb{N}(x,z)\mlc \cle t \bigl(\mathbb{U}_{\mbox{\tiny $\sqcup$}}^{\vec{s}}(x,t,z)\mlc  \mathbb{D}_{\mbox{\tiny $\sqcup$}}(t,\vec{s})\bigr)\Bigr)\Bigr] $.
\item $\mathbb{F}(x,y)$\label{x48} says ``$y$ is the numer of the move found in configuration $x$'s buffer''. 
\item $\tilde{E}^\circ(\vec{s})$\label{x49} abbreviates $\cle x\bigl({E}^\circ(x,\vec{s})\mlc \mathbb{D}^{\epsilon}_{\mbox{\tiny $\sqcup$}} (x,\vec{s})\bigr)$.
\item $\tilde{E}^{\bullet}(\vec{s})$\label{x50} abbreviates 
$\cle x\bigl({E}^{\bullet}(x,\vec{s})\mlc \mathbb{D}^{\epsilon}_{\mbox{\tiny $\sqcup$}} (x,\vec{s})\bigr)$.
\end{enumerate}

\subsection{Scenes}\label{sscn}

In this subsection  and later, unless otherwise suggested by the context, $n$, $\vec{s}$, $\vec{v}$ are as stipulated in Section \ref{bignot}.

Given a configuration $x$, by the {\bf scene}\label{xtitsc} of $x$ we shall mean a partial description of $x$ consisting of 
the following two pieces of information for the run tape and each of the work tapes of $x$: 
\begin{itemize}
  \item  The symbol scanned by the scanning head of the tape.  
  \item An indication  (yes/no) of whether the scanning head  is located at the  beginning of the tape.
\end{itemize}  
Take a note of the obvious fact that the number of all possible scenes is finite. We let $\mathfrak{j}$ denote that number, and 
let us correspondingly  fix a list 
\[\mbox{\it Scene}_1,\label{xscn}\ldots, \mbox{\it Scene}_{\mathfrak{j}}\]
of all scenes. Also, for each $i\in\{1,\ldots,\mathfrak{j}\}$, we let $\mbox{\it Scene}_i(x)$ be a natural formalization of the predicate ``$\mbox{\it Scene}_i$ is the scene of configuration $x$''.

According to the following lemma, information on $x$ contained in $\mathbb{D}(x,\vec{s},\vec{v})$ is sufficient to determine 
(in $\thr$) the scene of $x$. 

\begin{lemma}\label{sclemma}
$\thr$ proves 
\begin{equation}\label{lksw}
\cla x \bigl(\mathbb{D}(x,\vec{s},\vec{v})\mli \mbox{\it Scene}_1(x)\add\ldots\add \mbox{\it Scene}_{\mathfrak{j}}(x)\bigr).
\end{equation}
\end{lemma}

\begin{proof} Recall that $\vec{s}$ is the tuple $s_1,\ldots,s_n$ and $\vec{v}$ is the tuple 
$v_1,\ldots,v_{2\mathfrak{y}+3}$. Argue in $\thr$. Consider an arbitrary ($\cla$) configuration $x$,
keeping in mind  --- here and later in similar contexts --- that we do not really know the (``blind'') value of $x$. Assume  
$\mathbb{D}(x,\vec{s},\vec{v})$ is true, for otherwise (\ref{lksw}) will be won no matter how we (legally) act.  

Consider the $1$st work tape of $\mathcal X$. According to $\mathbb{D}(x,\vec{s},\vec{v})$, $v_1$ is the location of the corresponding scanning head in configuration $x$. Using Fact \ref{tri}, we figure out whether $v_1=0$. This way we come to know whether the  
scanning head of the tape is located at the  beginning 
of the tape.  Next, we know that 
$v_{\mathfrak{y}+1}$ is the content of $x$'s  1st work tape. Using the Log axiom and Fact \ref{tri}, we compare $|v_{\mathfrak{y}+1}|$ with $v_1$. If  $v_1\geq |v_{\mathfrak{y}+1}|$, we conclude that the symbol scanned by the head is $\blank$. And if 
$v_1< |v_{\mathfrak{y}+1}|$, then the symbol is either a $0$ or $1$; which of these two is the case depends on whether $\bit(v_{\mathfrak{y}+1},v_1)$ is true  or false;  we make such a determination using the Bit axiom. 

The other work tapes are handled similarly.

Finally, consider the run tape. We figure out whether $x$'s run-tape scanning head is looking at the leftmost cell of the tape by comparing $v_{2\mathfrak{y}+1}$ with $0$. The task of finding the symbol scanned by the scanning head in this case is less straightforward than in the case of the work tapes, but still doable in view of our ability to perform the basic arithmetic operations established in Section \ref{sboot}. We leave  details to the reader.

The information obtained by now fully determines which of $\mbox{\it Scene}_1,\ldots, \mbox{\it Scene}_{\mathfrak{j}}$  is the scene of $x$. We win (\ref{lksw}) by choosing the corresponding $\add$-disjunct in the consequent. 
\end{proof}

\subsection{The traceability lemma}\label{sml}

\begin{lemma}\label{l1}
$\thr\vdash z\leq \mathfrak{t}|\vec{s}| \mli  \mathbb{Q}(\vec{s},z)$. 
\end{lemma}

\begin{proof} Argue in $\thr$. 
We proceed by Reasonable $\mathcal R$-Induction on $z$. The {\em basis}  $\mathbb{Q}(\vec{s},0)$ abbreviates 
\[\cla x\Bigl(\mathbb{D}_{\mbox{\tiny $\sqcup$}}(x,\vec{s})\mli  \gneg \mathbb{N}(x,0)\add\bigl[\mathbb{N}(x,0)\mlc \cle t \bigl(\mathbb{U}_{\mbox{\tiny $\sqcup$}}^{\vec{s}}(x,t,0)\mlc  \mathbb{D}_{\mbox{\tiny $\sqcup$}}(t,\vec{s})\bigr)\bigr]\Bigr).\]
Solving it means solving the following problem for a blindly-arbitrary ($\cla$) $x$:
\[\mathbb{D}_{\mbox{\tiny $\sqcup$}}(x,\vec{s})\mli  \gneg \mathbb{N}(x,0)\add\bigl[\mathbb{N}(x,0)\mlc \cle t \bigl(\mathbb{U}_{\mbox{\tiny $\sqcup$}}^{\vec{s}}(x,t,0)\mlc  \mathbb{D}_{\mbox{\tiny $\sqcup$}}(t,\vec{s})\bigr)\bigr].\]
To solve the above, we wait till the adversary brings it down to 
\begin{equation}\label{yy11}  
|\vec{c}|\leq\mathfrak{s}|\vec{s}|\mlc \mathbb{D}(x,\vec{s},\vec{c})\mli  \gneg \mathbb{N}(x,0)\add\bigl[\mathbb{N}(x,0)\mlc \cle t \bigl(\mathbb{U}_{\mbox{\tiny $\sqcup$}}^{\vec{s}}(x,t,0)\mlc  \mathbb{D}_{\mbox{\tiny $\sqcup$}}(t,\vec{s})\bigr)\bigr]
\end{equation} 
for some $(2\mathfrak{y}+3)$-tuple $\vec{c}=c_1,\ldots,c_{2\mathfrak{y}+3}$ of constants. From now on we will assume that 
\begin{equation}\label{lkd}
|\vec{c}|\leq\mathfrak{s}|\vec{s}|\mlc \mathbb{D}(x,\vec{s},\vec{c})
\end{equation}
is true, for otherwise (\ref{yy11}) will be won no matter what. On this assumption, solving (\ref{yy11})  means solving its consequent, which disabbreviates as 
\begin{equation}\label{yy111}  
\gneg \mathbb{N}(x,0) \add \Bigl(\mathbb{N}(x,0)\mlc \cle t\bigl(\ade |r|\leq \mathfrak{s}|\vec{s}|\mathbb{U}(x,t,0,r)\mlc 
\ade|\vec{v}|\leq\mathfrak{s}|\vec{s}|\mathbb{D}(t,\vec{s},\vec{v})\bigr)\Bigr). 
\end{equation} 

In order to solve (\ref{yy111}), we first of all need to figure out whether $\mathbb{N}(x,0)$ is true. Even though we do not know the actual value of (the implicitly $\cla$-bounded) $x$, we do know that it satisfies (\ref{lkd}), and this is sufficient for our purposes. Note that $\mathbb{N}(x,0)$ is true iff $x$ is uncorrupt. So, it is sufficient to just go through the seven conditions of Definition \ref{dcor} and test their satisfaction.  From the $\mathbb{D}(x,\vec{s},\vec{c})$ conjunct of 
(\ref{lkd}), we know that $c_{2\mathfrak{y}+3}$ is $x$'s titular number. Therefore, $x$ is semiuncorrupt --- i.e., condition 1 of Definition \ref{dcor} is satisfied --- iff 
$c_{2\mathfrak{y}+3}< \hat{\mathfrak{m}}$. And whether    $c_{2\mathfrak{y}+3}<\hat{\mathfrak{m}}$ we can determine based on  Facts \ref{numerals} and \ref{tri}. Next, from the title $\mbox{\em Title}_{c_{2\mathfrak{y}+3}}$ of $x$, we can figure out which of the $n$ 
moves residing on $x$'s run tape are numeric. We look at the numers of such moves from among $s_{1},\ldots,s_{n}$ and, using Fact \ref{tri} several times, find the greatest numer $a$. After that, using the Log axiom, we find the background $\ell$ of $x$, which is nothing but $|a|$.  Knowing the value of $\ell$, we can now test the satisfaction of condition 2 of Definition \ref{dcor} based on clause 2 of Definition \ref{dadm}, 
the Log axiom and Fact \ref{tri}. Conditions 3 and 4 of Definition \ref{dcor} will be handled in a similar way. Next, from $c_{\mathfrak{y}+1},\ldots,c_{2\mathfrak{y}}$, we know the contents of the work tapes of $x$. This, in combination with the Log axiom,   allows us to determine the numbers of non-blank cells on those work tapes. Comparing 
those numbers with $\mathfrak{s}(\ell)$, 
we figure out whether condition 5 of Definition \ref{dcor} is satisfied. Checking the satisfaction of conditions 6 and 7 of Definition \ref{tri} is also a doable task, and we leave details to the reader. 
 
So, now we know whether $x$ is corrupt or not. If $x$ is corrupt,  we choose $\gneg\mathbb{N}(x,0)$ in (\ref{yy111}) and win. And if $x$ is uncorrupt, i.e., $\mathbb{N}(x,0)$ is true, then we bring (\ref{yy111}) down to
\[\mathbb{N}(x,0)\mlc \cle t\bigl(|0|\leq \mathfrak{s}|\vec{s}|\mlc\mathbb{U}(x,t,0,0)\mlc |\vec{c}|\leq\mathfrak{s}|\vec{s}|\mlc\mathbb{D}(t,\vec{s},\vec{c})\bigr) 
.\]
We win because the above is a  logical consequence of (\ref{lkd}), $\mathbb{N}(x,0)$ and the obviously true $|0|\leq \mathfrak{s}|\vec{s}|\mlc\mathbb{U}(x,x,0,0)$. 
The basis of our induction is thus proven.\vspace{3pt}

The {\em inductive step} is $z<\mathfrak{t}|\vec{s}|\mlc \mathbb{Q}(\vec{s},z)\mli \mathbb{Q}(\vec{s},z\successor)$, which partially disabbreviates as
\begin{equation}\label{y4}
\begin{array}{l}
z<\mathfrak{t}|\vec{s}|\mlc \cla x\Bigl(\mathbb{D}_{\mbox{\tiny $\sqcup$}}(x,\vec{s})\mli  \gneg \mathbb{N}(x,z)\add\bigl[\mathbb{N}(x,z)\mlc \cle t \bigl(\mathbb{U}_{\mbox{\tiny $\sqcup$}}^{\vec{s}}(x,t,z)\mlc  \mathbb{D}_{\mbox{\tiny $\sqcup$}}(t,\vec{s})\bigr)\bigr]\Bigr) \\ \mli
\cla x\Bigl(\mathbb{D}_{\mbox{\tiny $\sqcup$}}(x,\vec{s})\mli  \gneg \mathbb{N}(x,z\successor)\add\bigl[\mathbb{N}(x,z\successor)\mlc \cle t \bigl(\mathbb{U}_{\mbox{\tiny $\sqcup$}}^{\vec{s}}(x,t,z\successor)\mlc  \mathbb{D}_{\mbox{\tiny $\sqcup$}}(t,\vec{s})\bigr)\bigr]\Bigr).
\end{array}
\end{equation}
With some thought, (\ref{y4}) can be seen to be  a logical consequence of  
\[\begin{array}{l}
\cla x\cla t\Bigl[z<\mathfrak{t}|\vec{s}|\mlc \Bigl(  \gneg \mathbb{N}(x,z)\add\bigl[\mathbb{N}(x,z)\mlc  \bigl(\mathbb{U}_{\mbox{\tiny $\sqcup$}}^{\vec{s}}(x,t,z)\mlc  \mathbb{D}_{\mbox{\tiny $\sqcup$}}(t,\vec{s})\bigr)\bigr]\Bigr)  \\ \mli
 \gneg \mathbb{N}(x,z\successor)\add\bigl[\mathbb{N}(x,z\successor)\mlc \cle t \bigl(\mathbb{U}_{\mbox{\tiny $\sqcup$}}^{\vec{s}}(x,t,z\successor)\mlc  \mathbb{D}_{\mbox{\tiny $\sqcup$}}(t,\vec{s})\bigr)\bigr]\Bigr],
\end{array}\] 
so let us pick arbitrary ($\cla$) numbers $a$, $b$ in the roles of the $\cla$-bounded variables $x,t$ of the above expression and focus on  
\begin{equation}\label{yy44}
\begin{array}{l}
z<\mathfrak{t}|\vec{s}|\mlc \Bigl(  \gneg \mathbb{N}(a,z)\add\bigl[\mathbb{N}(a,z)\mlc  \bigl(\mathbb{U}_{\mbox{\tiny $\sqcup$}}^{\vec{s}}(a,b,z)\mlc  \mathbb{D}_{\mbox{\tiny $\sqcup$}}(b,\vec{s})\bigr)\bigr]\Bigr) \\ \mli
 \gneg \mathbb{N}(a,z\successor)\add\bigl[\mathbb{N}(a,z\successor)\mlc \cle t \bigl(\mathbb{U}_{\mbox{\tiny $\sqcup$}}^{\vec{s}}(a,t,z\successor)\mlc  \mathbb{D}_{\mbox{\tiny $\sqcup$}}(t,\vec{s})\bigr)\bigr].
\end{array} 
\end{equation}
To solve (\ref{yy44}), 
 we wait till the  $\add$-disjunction in its antecedent is resolved. If the adversary chooses the first $\add$-disjunct there,   we do the same  in the consequent and win, because $\gneg \mathbb{N}(a,z)$ obviously implies $\gneg \mathbb{N}(a,z\successor)$.  Now suppose the adversary chooses the second $\add$-disjunct in the antecedent. We wait further until 
 (\ref{yy44}) is brought down to 
\begin{equation}\label{446}
\begin{array}{l}
 z<\mathfrak{t}|\vec{s}|\mlc  \mathbb{N}(a,z)\mlc  |d|\leq \mathfrak{s}|\vec{s}|\mlc \mathbb{U}(a,b,z,d)\mlc |\vec{c}|\leq\mathfrak{s}|\vec{s}|\mlc \mathbb{D}(b,\vec{s},\vec{c})  \\ \mli
 \gneg \mathbb{N}(a,z\successor)\add\bigl[\mathbb{N}(a,z\successor)\mlc \cle t \bigl(\mathbb{U}_{\mbox{\tiny $\sqcup$}}^{\vec{s}}(a,t,z\successor)\mlc  \mathbb{D}_{\mbox{\tiny $\sqcup$}}(t,\vec{s})\bigr)\bigr]
\end{array}
\end{equation}
for some  constant $d$ and some $(2\mathfrak{y}+3)$-tuple $\vec{c}=c_1,\ldots,c_{2\mathfrak{y}+3}$
of constants. From now on we will assume that the antecedent 
\begin{equation}\label{gla}
 z<\mathfrak{t}|\vec{s}|\mlc  \mathbb{N}(a,z)\mlc  |d|\leq \mathfrak{s}|\vec{s}|\mlc \mathbb{U}(a,b,z,d)\mlc |\vec{c}|\leq\mathfrak{s}|\vec{s}|\mlc \mathbb{D}(b,\vec{s},\vec{c})  
\end{equation}
of (\ref{446}) is true, for otherwise we win (\ref{446}) no matter what. Our goal is to win the consequent of (\ref{446}), i.e., the game 
\begin{equation}\label{glc}
 \gneg \mathbb{N}(a,z\successor)\add\bigl[\mathbb{N}(a,z\successor)\mlc \cle t \bigl(\mathbb{U}_{\mbox{\tiny $\sqcup$}}^{\vec{s}}(a,t,z\successor)\mlc  \mathbb{D}_{\mbox{\tiny $\sqcup$}}(t,\vec{s})\bigr)\bigr].
\end{equation}

Using Fact \ref{tri}, we compare $d$ with $z$. The case $d>z$ is ruled out by our assumption (\ref{gla}), because it is inconsistent with the truth of $\mathbb{U}(a,b,z,d)$. 
If $d<z$,  we bring (\ref{glc}) down to 
\begin{equation}\label{fjs}
\mathbb{N}(a,z\successor)\mlc \cle t \bigl( |d|\leq \mathfrak{s}|\vec{s}|\mlc \mathbb{U}(a,t,z\successor,d)\mlc |\vec{c}|\leq\mathfrak{s}|\vec{s}|\mlc \mathbb{D}(t,\vec{s},\vec{c})\bigr),\end{equation}
which is a logical consequence of 
\begin{equation}\label{y44}
\mathbb{N}(a,z\successor)\mlc |d|\leq \mathfrak{s}|\vec{s}|\mlc \mathbb{U}(a,b,z\successor,d)\mlc |\vec{c}|\leq\mathfrak{s}|\vec{s}|\mlc \mathbb{D}(b,\vec{s},\vec{c}).
\end{equation}
This way we win, because (\ref{y44}) is true and hence so is (\ref{fjs}). Namely, the truth of  (\ref{y44}) follows from the truth of (\ref{gla}) in view of the fact that, on our assumption $d<z$, $\mathbb{U}(a,b,z,d)$ obviously implies $\mathbb{U}(a,b,z\successor,d)$ and 
$\mathbb{N}(a,z)$ implies $\mathbb{N}(a,z\successor)$.

Now suppose $d=z$. So, our resource (\ref{gla}) is the same as 
\begin{equation}\label{ggj}
 z<\mathfrak{t}|\vec{s}|\mlc  \mathbb{N}(a,z)\mlc  |z|\leq \mathfrak{s}|\vec{s}|\mlc \mathbb{U}(a,b,z,z)\mlc |\vec{c}|\leq\mathfrak{s}|\vec{s}|\mlc \mathbb{D}(b,\vec{s},\vec{c}).
\end{equation}
The $\mathbb{D} (b,\vec{s},\vec{c})$ component of  (\ref{ggj})  contains sufficient information on whether the  configuration  $b$ has any unadulterated successors other than itself.\footnote{Namely, $b$ has an unadulterated successor other than itself iff the state component of $b$ --- which can be found in $\mbox{\em Title}_{c_{2\mathfrak{y}+3}}$ --- is not a move state.}  
If not,  $\mathbb{N}(a,z)$ obviously implies $\mathbb{N}(a,z\successor)$ and $\mathbb{U}(a,b,z,z)$ implies $\mathbb{U}(a,b,z\successor,z)$; hence, (\ref{ggj}) implies 
\[ \mathbb{N}(a,z\successor)\mlc  |z|\leq \mathfrak{s}|\vec{s}|\mlc \mathbb{U}(a,b,z\successor,z)\mlc |\vec{c}|\leq\mathfrak{s}|\vec{s}|\mlc \mathbb{D}(b,\vec{s},\vec{c}),
\]
which, in turn, implies 
\begin{equation}\label{y8}
\mathbb{N}(a,z\successor)\mlc \cle t\bigl( |z|\leq \mathfrak{s}|\vec{s}|\mlc \mathbb{U}(a,t,z\successor,z)\mlc |\vec{c}|\leq\mathfrak{s}|\vec{s}|\mlc \mathbb{D}(t,\vec{s},\vec{c})\bigr).
\end{equation}
So, we win 
(\ref{glc}) by bringing it down to the true (\ref{y8}). 

 Now, for the rest of this proof, assume $b$ has unadulterated successors other than itself. From the $\mathbb{U}(a,b,z,z)$ conjunct of (\ref{ggj}) we also know that $b$ is a $z$th unadulterated successor of $a$.  Thus, a $(z+1)$st unadulterated successor of $a$ --- call it $e$ --- exists, implying the truth of 
\begin{equation}\label{y14}
\mathbb{U}(a,e,z\successor,z\successor).
\end{equation}

In order to solve (\ref{glc}), we want to find a tuple 
$\vec{d}=d_1,\ldots,d_{2\mathfrak{y}+3}$ of constants satisfying 
\begin{equation}\label{kantee}
\mathbb{D}(e,\vec{s},\vec{d})
\end{equation}
 --- that is, satisfying conditions 2(a) through $2(f)$ of Section \ref{bignot} with $e$, $\vec{s}$ and $\vec{d}$ in the roles of $x$, $\vec{s}$ and $\vec{v}$, respectively.  In doing so below, we shall rely on the truth of $\mathbb{D}(b,\vec{s},\vec{c})$ implied by (\ref{ggj}). 
We shall then also rely on our knowledge of the scene of $b$ obtained from $\mathbb{D}(b,\vec{s},\vec{c})$ based on Lemma \ref{sclemma}, and our knowledge of the  
state component of $b$ obtained from $c_{2\mathfrak{y}+3}$  (the $(2\mathfrak{y}+3)$rd constant of the tuple $\vec{c}$). 

First of all, notice that, no matter how we select  $\vec{d}$, condition 2(a)  of Section  \ref{bignot} is satisfied with $e$ in the role of $x$. This is so because, as implied by  $\mathbb{D}(b,\vec{s},\vec{c})$, that condition is satisfied with $b$ in the role of $x$, and $e$ is an unadulterated successor  of $b$, meaning that $b$ and $e$ have identical run-tape contents.  

From $ \mathbb{D}(b,\vec{s},\vec{c})$, we know that the location of $b$'s 1st work-tape head is $c_1$; based on our knowledge of the state and the scene of $b$, we can also figure out whether that tape's  scanning head moves to the right, to the left, or stays put  on the transition from $b$ to $e$. If it moves to the right, we  apply the Successor axiom and compute the value  $d_1$ to be $c_1\successor$.  If the head stays put or tries to move to the left while $c_1=0$ (whether $c_1=0$ we figure out using Fact \ref{tri}), we know that  $d_1=c_1$. Finally, if it moves to the left while $c_1\not=0$, then $d_1=c_1-1$, and we compute this value  using Facts \ref{numerals} and \ref{minus}.  
We find the constants $d_2,\ldots,d_{\mathfrak{y}}$ in a similar manner.

The values  $d_{\mathfrak{y}+1},\ldots,d_{2\mathfrak{y}}$  can be computed from $c_{\mathfrak{y}+1},\ldots,c_{2\mathfrak{y}}$ and our knowledge --- determined by $b$'s state and scene --- of the symbols written on $\mathcal X$'s work tapes on the transition from $b$ to $e$. If such a symbol was written in a previously non-blank cell (meaning that the size of the work tape content did not change), we shall rely on Fact \ref{br} in computing $d_{\mathfrak{y}+i}$ from $c_{\mathfrak{y}+i}$ ($1\leq i\leq \mathfrak{y}$), as the former is the result of changing one bit in the latter. Otherwise, if the new symbol was written in a previously blank (the leftmost blank) cell, then $d_{\mathfrak{y}+i}$ is either $c_{\mathfrak{y}+i}+c_{\mathfrak{y}+i}$ (if the written symbol is $0$) or $c_{\mathfrak{y}+i}+c_{\mathfrak{y}+i}+\hat{1}$ (if the written symbol is $1$); so, $d_{\mathfrak{y}+i}$ can be computed using Facts  \ref{numerals} and  \ref{plus}. 

We find the value $d_{2\mathfrak{y}+1}$ in a way similar to the way we found $d_1,\ldots,d_{\mathfrak{y}}$. 

From the state and the scene of $b$, we can also figure out whether the length of the numer of the string in the buffer has increased (by $1$) or not on the transition from $b$ to $e$. If not, we determine that 
$d_{2\mathfrak{y}+2}=c_{2\mathfrak{y}+2}$. If yes, then $d_{2\mathfrak{y}+2}=c_{2\mathfrak{y}+2}\successor$, which we compute using the Successor axiom. 

From the $\mathbb{N}(a,z)$ component of (\ref{ggj}) we know that configuration $a$ is uncorrupt and hence semiuncorrupt. From 
(\ref{y14}) we also know that $e$ is an unadulterated successor of $a$. As an unadulterated successor of a semiuncorrupt configuration,  $e$ obviously remains semiuncorrupt, meaning that its titular number $d_{2\mathfrak{y}+3}$ is an element of the set $\{0,\ldots,\mathfrak{m}-1\}$. Which of these $\mathfrak{m}$ values is precisely assumed by $d_{2\mathfrak{y}+3}$  is fully determined by the title and the scene of $b$, both of which we know.   
All $2\mathfrak{y}+3$ constants from the $\vec{d}$ group are now found.

As our next step, from (\ref{kantee}) --- from $\mathbb{D}(e,\vec{s},\vec{d})$, that is --- we figure out whether $e$ is corrupt in the same style as from $\mathbb{D}(x,\vec{s},\vec{c})$  we figured out whether $x$ was corrupt when building our strategy for (\ref{yy111}).  If $e$ is corrupt, we choose $\gneg \mathbb{N}(a,z\successor)$ in (\ref{glc}) and win. Now, for the rest of this proof, assume 
\begin{equation}\label{hgks}
\mbox{\it $e$ is uncorrupt.}\end{equation}
 Using the Successor axiom, we compute the value $g$ of $z\successor$ and then we bring (\ref{glc})  down to 
\begin{equation}\label{fjss}
\mathbb{N}(a,g)\mlc \cle t \bigl( |g|\leq \mathfrak{s}|\vec{s}|\mlc \mathbb{U}(a,t,g,g)\mlc |\vec{d}|\leq\mathfrak{s}|\vec{s}|\mlc \mathbb{D}(t,\vec{s},\vec{d})\bigr),\end{equation}
which is a  logical consequence of 
\begin{equation}\label{y44s}
\mathbb{N}(a,g)\mlc |g|\leq \mathfrak{s}|\vec{s}|\mlc \mathbb{U}(a,e,g,g)\mlc |\vec{d}|\leq\mathfrak{s}|\vec{s}|\mlc \mathbb{D}(e,\vec{s},\vec{d}).
\end{equation}

To declare victory, it remains to see that (\ref{y44s}) is true. The $3$rd and the $5$th conjuncts of 
(\ref{y44s}) are true because they are nothing but (\ref{y14}) and (\ref{kantee}), respectively. The $4$th conjunct can be seen to follow from (\ref{kantee}) and (\ref{hgks}). From (\ref{ggj}), we know that $z<\mathfrak{t}|\vec{s}|$, which implies $g\leq \mathfrak{t}|\vec{s}|$ and hence $|g|\leq |\mathfrak{t}|\vec{s}||$. Since $e$ is uncorrupt,  by clause 2 of Definition \ref{dcor}, we also have  $|\mathfrak{t}|\vec{s}||\leq \mathfrak{s}|\vec{s}|$. Thus, the second conjunct  of (\ref{y44s}) is also true. Finally, for the first conjunct of (\ref{y44s}), observe the following. 
According to (\ref{ggj}), $\mathbb{N}(a,z)$ is true, meaning that $a$ does not have a corrupt $k$th unadulterated successor for any $k$ with $k\leq z$.  By (\ref{hgks}), $e$ ---  
which is the $(z+1)$th unadulterated successor of $a$ --- is uncorrupt. Thus,  $a$ does not have a corrupt $k$th unadulterated successor for any $k$ with $k\leq z+1=g$. This means nothing but that  $\mathbb{N}(a,g)$ is true. 
 \end{proof}

\subsection{Junior lemmas}
\begin{lemma}\label{l3}
 $\thr\vdash  \ade z \bigl(z=\mathfrak{t}|\vec{s}|\mlc  \mathbb{Q}(\vec{s},z)\bigr)$. 
\end{lemma}

\begin{proof} Argue in $\thr$. Using Fact \ref{tri} several times, we find the greatest number $s$ among $\vec{s}$. Then, relying on the Log axiom and condition 2 of Definition \ref{dadm},
we compute  the value $b$ of $\mathfrak{t}|s|$.
Specifying $z$ as  $b$ in the resource provided by  Lemma \ref{l1}, we bring the latter down to 
\begin{equation}\label{dsw}
b\leq \mathfrak{t}|\vec{s}| \mli \mathbb{Q}(\vec{s},b).
\end{equation}
 Now, the target $\ade z \bigl(z=\mathfrak{t}|\vec{s}|\mlc  \mathbb{Q}(\vec{s},z)\bigr)$ is won by specifying $z$ as $b$, and then synchronizing the second conjunct of the resulting  $b=\mathfrak{t}|\vec{s}|\mlc \mathbb{Q}(\vec{s},b)$ with the consequent of (\ref{dsw}) --- that is, acting in the former exactly as the provider of (\ref{dsw}) acts in the latter, and ``vice versa'':  acting in the latter as Environment acts in former. 
\end{proof}
  
For the purposes of the following two lemmas, we agree that  $\mbox{\em Nothing}\hspace{1pt}(t,q)$\label{xnothing} is an elementary formula asserting that the numer $c$ of the move found in configuration $t$'s buffer  does not have a $q$th most significant bit (meaning that either $q=0$ or $|c|<q$). Next, 
$\mbox{\em Zero}\hspace{1pt}(t,q)$\label{xzero} means ``$\gneg \mbox{\em Nothing}\hspace{1pt}(t,q)$ and the $q$th most significant bit of the numer of the move found in $t$'s buffer is a $0$''. Similarly,  
$\mbox{\em One}\hspace{1pt}(t,q)$\label{xone} means ``$\gneg \mbox{\em Nothing}\hspace{1pt}(t,q)$ and the $q$th most significant bit of the numer of the move found in $t$'s buffer is a $1$''. 

\begin{lemma}\label{l4a}
 $\thr$ proves 
\begin{equation}\label{exta}
\begin{array}{l}
z\leq \mathfrak{t}|\vec{s}|\mli  
\cla x\cla t \bigl(\mathbb{N}(x,z)\mlc |\vec{v}|\leq \mathfrak{s}|\vec{s}|\mlc \mathbb{D}^\epsilon(x,\vec{s},\vec{v})\mlc 
  \mathbb{U}(x,t,z,z)\mli \\
 \mbox{Nothing}\hspace{1pt}(t,q)\add \mbox{Zero}\hspace{1pt}(t,q)\add \mbox{One}\hspace{1pt}(t,q) \bigr).
\end{array}
\end{equation}
\end{lemma}

\begin{proof} Argue in $\thr$.
Reasonable Induction on $z$. The basis is 
\[\cla x\cla t \bigl(\mathbb{N}(x,0)\mlc |\vec{v}|\leq \mathfrak{s}|\vec{s}|\mlc \mathbb{D}^\epsilon(x,\vec{s},\vec{v})\mlc 
  \mathbb{U}(x,t,0,0)\mli \mbox{\em Nothing}\hspace{1pt}(t,q)\add \mbox{\em Zero}\hspace{1pt}(t,q)\add \mbox{\em One}\hspace{1pt}(t,q) \bigr),\]
which is obviously won by choosing $\mbox{\em Nothing}\hspace{1pt}(t,q)$ in the consequent.

The inductive step is 
\begin{equation}\label{extaa}
\begin{array}{l}
 z<\mathfrak{t}|\vec{s}|\mlc
\cla x\cla t \bigl(\mathbb{N}(x,z)\mlc|\vec{v}|\leq \mathfrak{s}|\vec{s}|\mlc \mathbb{D}^\epsilon(x,\vec{s},\vec{v})\mlc 
  \mathbb{U}(x,t,z,z)\mli \\
\mbox{\em Nothing}\hspace{1pt}(t,q)\add \mbox{\em Zero}\hspace{1pt}(t,q)\add \mbox{\em One}\hspace{1pt}(t,q) \bigr)\mli
\cla x\cla t \bigl(\mathbb{N}(x,z\successor)\mlc|\vec{v}|\leq \mathfrak{s}|\vec{s}|\mlc\\ \mathbb{D}^\epsilon(x,\vec{s},\vec{v})\mlc 
  \mathbb{U}(x,t,z\successor,z\successor)\mli 
\mbox{\em Nothing}\hspace{1pt}(t,q)\add \mbox{\em Zero}\hspace{1pt}(t,q)\add \mbox{\em One}\hspace{1pt}(t,q) \bigr).
\end{array}
\end{equation}
To solve (\ref{extaa}), 
we wait till the adversary makes a choice in the  antecedent. If it chooses $\mbox{\em Zero}\hspace{1pt}(t,q)$ or $\mbox{\em One}\hspace{1pt}(t,q)$, we make the same choice in the consequent, and rest our case. Suppose now the adversary chooses $\mbox{\em Nothing}\hspace{1pt}(t,q)$, thus bringing (\ref{extaa}) down to 
\begin{equation}\label{d18}
\begin{array}{l}
 z<\mathfrak{t}|\vec{s}|\mlc \cla x\cla t \bigl(\mathbb{N}(x,z)\mlc|\vec{v}|\leq \mathfrak{s}|\vec{s}|\mlc \mathbb{D}^\epsilon(x,\vec{s},\vec{v})\mlc 
  \mathbb{U}(x,t,z,z)\mli \\ \mbox{\it Nothing}\hspace{1pt}(t,q) \bigr)
 \mli 
 \cla x\cla t \bigl(\mathbb{N}(x,z\successor)\mlc|\vec{v}|\leq \mathfrak{s}|\vec{s}|\mlc \\ \mathbb{D}^\epsilon(x,\vec{s},\vec{v})\mlc 
  \mathbb{U}(x,t,z\successor,z\successor)\mli \mbox{\it Nothing}\hspace{1pt}(t,q)\add \mbox{\em Zero}\hspace{1pt}(t,q)\add \mbox{\em One}\hspace{1pt}(t,q) \bigr).
\end{array}
\end{equation}
In order to win (\ref{d18}), we need a strategy that, for arbitrary ($\cla$) and unknown $a$ and $c$, wins 
\begin{equation}\label{d18b}
\begin{array}{l}
 z<\mathfrak{t}|\vec{s}|\mlc\cla x\cla t \bigl(\mathbb{N}(x,z)\mlc|\vec{v}|\leq \mathfrak{s}|\vec{s}|\mlc \mathbb{D}^\epsilon(x,\vec{s},\vec{v})\mlc 
  \mathbb{U}(x,t,z,z)\mli \\ \mbox{\it Nothing}\hspace{1pt}(t,q) \bigr)
 \mli 
\bigl(\mathbb{N}(a,z\successor)\mlc|\vec{v}|\leq \mathfrak{s}|\vec{s}|\mlc \\ \mathbb{D}^\epsilon(a,\vec{s},\vec{v})\mlc  
  \mathbb{U}(a,c,z\successor,z\successor)\mli  \mbox{\em Nothing}\hspace{1pt}(c,q)\add \mbox{\em Zero}\hspace{1pt}(c,q)\add \mbox{\em One}\hspace{1pt}(c,q)\bigr).
\end{array}
\end{equation}
To solve (\ref{d18b}), assume both the antecedent and the antecedent of the consequent of it are true (otherwise we win no matter what). So, all of the following statements are true:
\begin{eqnarray}
\label{q1} & z< \mathfrak{t}|\vec{s}|; & \\
\label{q1.5} & \cla x\cla t \bigl(\mathbb{N}(x,z)\mlc|\vec{v}|\leq \mathfrak{s}|\vec{s}|\mlc \mathbb{D}^\epsilon(x,\vec{s},\vec{v})\mlc 
 \mathbb{U}(x,t,z,z)\mli \mbox{\it Nothing}\hspace{1pt}(t,q) \bigr); & \\
\label{q2} & \mathbb{N}(a,z\successor)\mlc |\vec{v}|\leq \mathfrak{s}|\vec{s}| \mlc \mathbb{D}^\epsilon(a,\vec{s},\vec{v});  & \\
\label{q4} & \mathbb{U}(a,c,z\successor,z\successor). & 
\end{eqnarray}
Assumption (\ref{q4}) implies that $a$ has (not only a $(z\successor)$th but also) a $z$th unadulterated successor. Let $b$ be that successor. Thus, the following is true:
\begin{equation}\label{d18d}
\mathbb{U}(a,b,z,z ).
\end{equation}
The $\mathbb{N}(a,z\successor)$ conjunct of (\ref{q2}), of course, implies 
\begin{equation}\label{fre}
\mathbb{N}(a,z).
\end{equation}
From (\ref{q1.5}),  we also get 
\[
\mathbb{N}(a,z)\mlc |\vec{v}|\leq \mathfrak{s}|\vec{s}|\mlc \mathbb{D}^\epsilon(a,\vec{s},\vec{v})\mlc 
 \mathbb{U}(a,b,z,z)\mli \mbox{\it Nothing}\hspace{1pt}(b,q),\]
 which, together with (\ref{q2}), (\ref{d18d}) and (\ref{fre}), implies 
\begin{equation}\label{k15}
 \mbox{\it Nothing}(b,q). 
\end{equation}

From (\ref{q1}), we have $z'\leq \mathfrak{t}|\vec{s}|$. Hence, using Lemma \ref{l1} in combination with the Successor axiom, we can obtain the resource $ \mathbb{Q}(\vec{s},z\successor)$, which disabbreviates as 
\[\cla x\Bigl[\mathbb{D}_{\mbox{\tiny $\sqcup$}}(x,\vec{s})\mli  \gneg \mathbb{N}(x,z\successor)\add\Bigl(\mathbb{N}(x,z\successor)\mlc \cle t \bigl(\mathbb{U}_{\mbox{\tiny $\sqcup$}}^{\vec{s}}(x,t,z\successor)\mlc  \mathbb{D}_{\mbox{\tiny $\sqcup$}}(t,\vec{s})\bigr)\Bigr)\Bigr]. \]
We bring the above down to 
\begin{equation}\label{uy9}
\cla x\Bigl[|\vec{v}|\leq\mathfrak{s}|\vec{s}|\mlc \mathbb{D}(x,\vec{s},\vec{v})\mli  \gneg \mathbb{N}(x,z\successor)\add\Bigl(\mathbb{N}(x,z\successor)\mlc \cle t \bigl(\mathbb{U}_{\mbox{\tiny $\sqcup$}}^{\vec{s}}(x,t,z\successor)\mlc  \mathbb{D}_{\mbox{\tiny $\sqcup$}}(t,\vec{s})\bigr)\Bigr)\Bigr].
\end{equation}
Now (\ref{uy9}), in conjunction with (\ref{q2}) and the obvious fact $\cla\bigl(\mathbb{D}(x,\vec{s},\vec{v})\mli \mathbb{D}^\epsilon(a,\vec{s},\vec{v})\bigr)$,  implies 
$ \cle t\bigl( \mathbb{U}_{\mbox{\tiny $\sqcup$}}^{\vec{s}}(a,t,z\successor)\mlc$ $\mathbb{D}_{\mbox{\tiny $\sqcup$}} (t,\vec{s}) \bigr)$, i.e.,  
\begin{equation}\label{q5}
\cle t \bigl(\ade |r|\leq\mathfrak{s}|\vec{s}|\mathbb{U}(a,t,z\successor,r)\mlc \mathbb{D}_{\mbox{\tiny $\sqcup$}} (t,\vec{s}) \bigr).
\end{equation}

From (\ref{q4}), by \pa, we know that  $c$ is the unique number satisfying $\mathbb{U}(a,t,z\successor,r)$ in the role of $t$ for some $r$ (in fact, for $r=z\successor$ and only for $r=z\successor$). This implies that  the 
provider of (\ref{q5}), in fact, provides (can only provide) the resource 
\[\ade |r|\leq\mathfrak{s}|\vec{s}|\mathbb{U}(a,c,z\successor,r)\mlc \mathbb{D}_{\mbox{\tiny $\sqcup$}} (c,\vec{s}).\]
 Thus, $\mathbb{D}_{\mbox{\tiny $\sqcup$}} (c,\vec{s})$ is at our disposal, which disabbreviates as 
$\ade |\vec{v}|\leq\mathfrak{s}|\vec{s}|\mathbb{D}(c,\vec{s},\vec{v})$.
The provider of this resource will have to bring it down to  
\begin{equation}\label{q7}
|\vec{d}|\leq\mathfrak{s}|\vec{s}|\mlc \mathbb{D}(c,\vec{s},\vec{d})
\end{equation}
for some tuple $\vec{d}=d_1,\ldots,d_{2\mathfrak{y}+3}$ of constants. Here $d_{2\mathfrak{y}+2}$ is 
the length of the numer of the move found in $c$'s buffer.
 Using Fact \ref{tri}, we figure out whether $d_{2\mathfrak{y}+2}=q$. If $d_{2\mathfrak{y}+2}\not=q$, we choose $\mbox{\em Nothing}(c,q)$ in the consequent of (\ref{d18b}). Now suppose 
$d_{2\mathfrak{y}+2}=q$. In this case, from $d_{2\mathfrak{y}+3}$ (the title of $c$), we extract information about what bit has been placed into the buffer on the transition from $b$ to $c$.\footnote{A symbol other than $0$ or $1$ could not have been placed into the buffer, because then, by clause 7 of Definition \ref{dcor}, $c$ would be corrupt, contradicting the $\mathbb{N}(a,z\successor)$ conjunct of (\ref{q2}).} If that bit is $1$, we choose $ \mbox{\em One}\hspace{1pt}(c,q)$ in  (\ref{d18b}); otherwise choose $\mbox{\em Zero}\hspace{1pt}(c,q)$. With a little thought and with (\ref{k15}) in mind, it can be seen that our strategy succeeds. 
\end{proof}

\begin{lemma}\label{l4}
$\thr$  proves 
\begin{equation}\label{ext}
\begin{array}{l}
 \cle x \cle t\cle y \bigl(\mathbb{N}(x,\mathfrak{t}|\vec{s}|)\mlc \mathbb{D}^\epsilon(x,\vec{s},\vec{v})\mlc 
  \mathbb{U}_{\mbox{\tiny $\exists$}}^{\vec{s}}(x,t) \mlc \mathbb{F}(t,y)\mlc \bit(r,y)\bigr)\add\\
\gneg  \cle x \cle t\cle y 
\bigl( \mathbb{N}(x,\mathfrak{t}|\vec{s}|)\mlc \mathbb{D}^\epsilon(x,\vec{s},\vec{v})\mlc 
  \mathbb{U}_{\mbox{\tiny $\exists$}}^{\vec{s}}(x,t) \mlc \mathbb{F}(t,y)\mlc \bit(r,y)\bigr).
\end{array}
\end{equation}
\end{lemma}

\begin{proof}   Argue in $\thr$. From $\pa$ we know that values $x$, $t$, $y$ satisfying 
\begin{equation}\label{uyoe}
 \mathbb{D}^\epsilon(x,\vec{s},\vec{v})\mlc 
  \mathbb{U}_{\mbox{\tiny $\exists$}}^{\vec{s}}(x,t) \mlc \mathbb{F}(t,y)
\end{equation}
exist ($\cle$) and are unique. Fix them for the rest of this proof. This allows us to switch from (\ref{ext}) to  (\ref{extrs}) as the target for our strategy, because the two paraformulas are  identical as a games: 
\begin{equation}\label{extrs}
\begin{array}{l}
\bigl(\mathbb{N}(x,\mathfrak{t}|\vec{s}|)\mlc \bit(r,y)\bigr)\add\gneg
\bigl( \mathbb{N}(x,\mathfrak{t}|\vec{s}|)\mlc \bit(r,y)\bigr).
\end{array}
\end{equation}

Relying on the Log axiom, Fact \ref{tri}  and clause 2 of Definition \ref{dadm}, 
we find the value of $\mathfrak{s}|\vec{s}|$.  Then, 
using that value and relying on the Log axiom and Fact \ref{tri} again, we figure out the truth status of 
$|\vec{v}|\leq \mathfrak{s}|\vec{s}|$.  If it  is false, then, with a little analysis of Definition \ref{dcor},
$x$ can be seen to be corrupt; for this reason, $\mathbb{N}(x,\mathfrak{t}|\vec{s}|)$ is false, so we choose the right $\add$-disjunct in (\ref{extrs}) and rest our case. Now, for the remainder of this proof, assume 
\begin{equation}\label{short}
 |\vec{v}|\leq \mathfrak{s}|\vec{s}|.
\end{equation}
 
By Lemma \ref{l3},  the resource $\mathbb{Q}(\vec{s},\mathfrak{t}|\vec{s}|)$, i.e., 
\[\cla x\Bigl[ \mathbb{D}_{\mbox{\tiny $\sqcup$}}(x,\vec{s})\mli \gneg\mathbb{N}(x,\mathfrak{t}|\vec{s}|)\add \Bigl(\mathbb{N}(x,\mathfrak{t}|\vec{s}|)\mlc \cle t \bigl( \mathbb{U}_{\mbox{\tiny $\sqcup$}}^{\vec{s}}(x,t,\mathfrak{t}|\vec{s}|)\mlc  \mathbb{D}_{\mbox{\tiny $\sqcup$}}(t,\vec{s})\bigr)\Bigr)\Bigr],\]
 is at our disposal. We bring it down to  
\[\cla x\Bigl[|\vec{v}|\leq\mathfrak{s}|\vec{s}|\mlc \mathbb{D}(x,\vec{s},\vec{v})\mli \gneg\mathbb{N}(x,\mathfrak{t}|\vec{s}|)\add \Bigl(\mathbb{N}(x,\mathfrak{t}|\vec{s}|)\mlc \cle t \bigl( \mathbb{U}_{\mbox{\tiny $\sqcup$}}^{\vec{s}}(x,t,\mathfrak{t}|\vec{s}|)\mlc  \mathbb{D}_{\mbox{\tiny $\sqcup$}}(t,\vec{s})\bigr)\Bigr)\Bigr],\]
which, in view of (\ref{uyoe}), (\ref{short}) and the fact $\cla\bigl(\mathbb{D}^\epsilon(x,\vec{s},\vec{v})\mli\mathbb{D}(x,\vec{s},\vec{v})\bigr)$,  implies  
\begin{equation}\label{saxli1}
\gneg\mathbb{N}(x,\mathfrak{t}|\vec{s}|)\add \Bigl(\mathbb{N}(x,\mathfrak{t}|\vec{s}|)\mlc \cle t \bigl( \mathbb{U}_{\mbox{\tiny $\sqcup$}}^{\vec{s}}(x,t,\mathfrak{t}|\vec{s}|)\mlc  \mathbb{D}_{\mbox{\tiny $\sqcup$}}(t,\vec{s})\bigr)\Bigr).\end{equation}
We wait till one of the two $\add$-disjuncts of (\ref{saxli1}) is selected by the provider. If the left disjunct is selected, 
 we choose the right $\add$-disjunct in (\ref{extrs}) and retire. Now suppose the right disjunct of (\ref{saxli1}) is selected. Such a move, with $\mathbb{U}_{\mbox{\tiny $\sqcup$}}^{\vec{s}}(x,t,\mathfrak{t}|\vec{s}|)$ and $\mathbb{D}_{\mbox{\tiny $\sqcup$}}(t,\vec{s})$ disabbreviated, brings (\ref{saxli1}) down to   
\begin{equation}\label{w2}
\mathbb{N}(x,\mathfrak{t}|\vec{s}|)\mlc\cle t\Bigl(\ade u \bigl(|u|\leq \mathfrak{s}|\vec{s}|\mlc \mathbb{U}(x,t,\mathfrak{t}|\vec{s}|,u)\bigr)\mlc  \ade \vec{v}\bigl( |\vec{v}|\leq \mathfrak{s}|\vec{s}|\mlc \mathbb{D}(t,\vec{s},\vec{v})\bigr)\Bigr).
\end{equation}
We wait till   (\ref{w2}) is fully resolved by its provider, i.e., is brought down to 
\begin{equation}\label{koro}
\mathbb{N}(x,\mathfrak{t}|\vec{s}|)\mlc\cle t\bigl(|a|\leq \mathfrak{s}|\vec{s}|\mlc \mathbb{U}(x,t,\mathfrak{t}|\vec{s}|,a) \mlc |\vec{d}|\leq \mathfrak{s}|\vec{s}|\mlc \mathbb{D}(t,\vec{s},\vec{d})\bigr)
\end{equation}
for some constant $a$ and tuple  $\vec{d}=d_1,\ldots,d_{2\mathfrak{y}+3}$ of constants.  By $\pa$,  (\ref{uyoe}) and (\ref{koro})  imply 
\begin{equation}\label{wx}
\mathbb{N}(x,\mathfrak{t}|\vec{s}|)\mlc \mathbb{U}(x,t,\mathfrak{t}|\vec{s}|,a)\mlc \mathbb{D}(t,\vec{s},\vec{v}).
\end{equation}
The $\mathbb{U}(x,t,\mathfrak{t}|\vec{s}|,a)$ conjunct of (\ref{wx}) further implies 
\begin{equation}\label{nsl}
a\leq\mathfrak{t}|\vec{s}| \mlc \mathbb{U}(x,t,a,a).
\end{equation}
By $\pa$, the $\mathbb{N}(x,\mathfrak{t}|\vec{s}|)$ conjunct of (\ref{wx}) and the $a\leq\mathfrak{t}|\vec{s}|$ conjunct of (\ref{nsl}) imply
\begin{equation}\label{csa}
\mathbb{N}(x,a).
\end{equation}

The $\mathbb{D}(t,\vec{s},\vec{d})$ conjunct of (\ref{wx}) implies that  
   $d_{2\mathfrak{y}+2}$ is the length of the numer of the move residing in  $t$'s buffer. By the $\mathbb{F}(t,y)$ conjunct of (\ref{uyoe}) we know that $y$ is such a numer. Thus,  
   $d_{2\mathfrak{y}+2}=|y|$. Let $q=d_{2\mathfrak{y}+2}\ominus  r$. This number can be computed using Fact \ref{minus}. 
The $r$th least significant bit of $y$  is nothing but the $q$th most significant bit of $y$. 

By Lemma \ref{l4a}, we have  
\begin{equation}\label{ea}
\begin{array}{l}
a\leq \mathfrak{t}|\vec{s}|\mlc\mathbb{N}(x,a)\mlc |\vec{v}|\leq \mathfrak{s}|\vec{s}|\mlc \mathbb{D}^\epsilon(x,\vec{s},\vec{v})\mlc 
  \mathbb{U}(x,t,a,a)\mli \\ \mbox{Nothing}\hspace{1pt}(t,q)\add \mbox{Zero}\hspace{1pt}(t,q)\add \mbox{One}\hspace{1pt}(t,q).
\end{array}
\end{equation}
The $a\leq \mathfrak{t}|\vec{s}|$ and $\mathbb{U}(x,t,a,a)$ conjuncts of the antecedent of (\ref{ea}) are true by (\ref{nsl}); the $\mathbb{N}(x,a)$ conjunct is true by (\ref{csa}); the $|\vec{v}|\leq \mathfrak{s}|\vec{s}|$ conjunct is true by (\ref{short}); and the  $\mathbb{D}^\epsilon(x,\vec{s},\vec{v})$ conjunct is true by (\ref{uyoe}). Hence, 
the provider of (\ref{ea}) has to resolve the $\add$-disjunction in the consequent.  If it chooses $\mbox{\it One}(t,q)$, we choose the left $\add$-disjunct in (\ref{extrs}); otherwise we choose the right $\add$-disjunct. 
 In either case we win.
\end{proof}

\subsection{Senior lemmas}\label{smainl}

Let $E$ be a formula not containing the variable $y$. We say that a formula $H$ is a {\bf $(\oo,y)$-development} of  $E$ iff $H$ is the result of replacing in $E$: 
\begin{itemize}
\item either a surface occurrence of a subformula $F_0\adc F_1$ by $F_i$ ($i=0$ or $i=1$), 
\item  or  a surface occurrence of a subformula $\ada xF(x)$ by $F(y)$.   
\end{itemize}

{\bf $(\pp ,y)$-development} is defined in the same way, only with $\add,\ade$ instead of $\adc,\ada$.  

\begin{lemma}\label{l6}
Assume $E(\vec{s})$ is a formula all of whose free variables are among $\vec{s}$,  $y$ is a variable not occurring in $E(\vec{s})$,  and $H(\vec{s},y)$ is a
 $(\oo,y)$-development  of $E(\vec{s})$. Then $\thr$  proves $\tilde{E}^{\circ}   (\vec{s})  \mli  \tilde{H}^{\circ}(\vec{s},y)$.
\end{lemma}

\begin{proof} 
Assume the conditions of the lemma. The target formula whose $\thr$-provability we want to show  partially disabbreviates as   
\begin{equation}\label{g2}
\begin{array}{l}
\cle x\bigl({E}^\circ(x,\vec{s})\mlc \mathbb{D}^{\epsilon}_{\mbox{\tiny $\sqcup$}} (x,\vec{s})\bigr)\mli\cle x\bigl({H}^\circ(x,\vec{s},y)\mlc \mathbb{D}^{\epsilon}_{\mbox{\tiny $\sqcup$}} (x,\vec{s},y)\bigr).
\end{array}
\end{equation}

Let $\oo\beta$ be the labmove that brings $E(\vec{s})$ down to $H(\vec{s},y)$,\footnote{In the rare cases where there are more than one such 
$\beta$, take the lexicographically smallest one.} and let $\alpha$ be the header of $\beta$.  For instance, if $E(\vec{s})$ is $G\mli F_0\adc F_1$ 
and $H(\vec{s},y)$ is $G\mli F_0$, then both $\beta$ and $\alpha$ are ``$\oo 1.0$''; and if $E(\vec{s})$ is $G\mli \ada zF(x)\mld J$ and 
$H(\vec{s},y)$ is $G\mli F(y)\mld J$, then $\beta$ is $1.0.\#y$  and $\alpha$ is $1.0.\#$.

For each natural number $j$, let $j^+$ be the number such that the first three titular components of $\mbox{\em Title}_{j^+}$ are the same as those of  $\mbox{\em Title}_{j}$, and the $4$th titular component of 
$\mbox{\em Title}_{j^+}$ is obtained from that of $\mbox{\em Title}_{j}$ by appending $\oo\alpha$ to it. 
Intuitively, if $\mbox{\em Title}_{j}$ is the title of a given configuration $x$, then $\mbox{\em Title}_{j^+}$ is the title of the configuration that results from $x$ in the scenario where $\oo$ made the (additional) move $\beta$ on the transition to $x$ from the predecessor configuration.
Observe that, if $j$ is a member of $\{0,\ldots,\mathfrak{m}-1\}$, then so is $j^+$.   

Argue in $\thr$. To win (\ref{g2}), we wait till Environment brings it down to 
\begin{equation}\label{lkk}
\begin{array}{l}
\cle x\bigl({E}^\circ(x,\vec{s})\mlc |\vec{c}|\leq\mathfrak{s}|\vec{s}|\mlc \mathbb{D}^{\epsilon} (x,\vec{s},\vec{c})\bigr)\mli  \cle x\bigl({H}^\circ(x,\vec{s},y)\mlc \mathbb{D}^{\epsilon}_{\mbox{\tiny $\sqcup$}} (x,\vec{s},y)\bigr) 
\end{array}
\end{equation}
for some tuple $\vec{c}=c_1,\ldots,c_{2\mathfrak{y}+3}$ of constants. 
Based on clause 2 of Definition \ref{dadm}
and Facts \ref{numerals} and \ref{tri}, we check whether 
$c_{2\mathfrak{y}+3}<\hat{\mathfrak{m}}$. If not, the antecedent of (\ref{lkk}) can be seen to be false, so we win (\ref{lkk}) by doing nothing. Suppose now $c_{2\mathfrak{y}+3}<\hat{\mathfrak{m}}$. In this case we  bring (\ref{lkk}) down to 
\begin{equation}\label{g4}
\begin{array}{l}
\cle x\bigl({E}^\circ(x,\vec{s})\mlc |\vec{c}|\leq\mathfrak{s}|\vec{s}|\mlc \mathbb{D}^{\epsilon} (x,\vec{s},\vec{c})\bigr)\mli \\ \cle x\bigl({H}^\circ(x,\vec{s},y)\mlc |\vec{c}^{\hspace{2pt}+}|\leq\mathfrak{s}|\vec{s}|\mlc \mathbb{D}^{\epsilon}(x,\vec{s},y,
\vec{c}^{\hspace{2pt}+})\bigr),
\end{array}
\end{equation}
where $\vec{c}^{\hspace{2pt}+}$ is the same as $\vec{c}$, only with $c_{2\mathfrak{y}+3}^+$ instead of $c_{2\mathfrak{y}+3}$.
The  elementary formula (\ref{g4}) can be easily seen to be true, so we win. 
\end{proof}

\begin{lemma}\label{l5}
Assume $E(\vec{s})$ is a formula all of whose free variables are among $\vec{s}$,  $y$ is a variable not occurring in $E(\vec{s})$, and       
 $H_1(\vec{s},y),\ldots,H_n(\vec{s},y)$ are all of the $(\pp,y)$-developments of $E(\vec{s})$. Then $\thr$  proves
\begin{equation}\label{m2ex}    \tilde{E}^{\circ}  (\vec{s})  \mli \tilde{E}^{\bullet}   (\vec{s})\add \gneg\mathbb{W}\add 
\ade y \tilde{H}_{1}^{\circ}(\vec{s},y)\add\ldots\add \ade y \tilde{H}_{n}^{\circ}(\vec{s},y)  . \end{equation}
\end{lemma}

\begin{proof} Assume the conditions of the lemma and 
argue in $\thr$  to justify (\ref{m2ex}). The antecedent 
of (\ref{m2ex}) disabbreviates as  
$ \cle x\bigl({E}^\circ(x,\vec{s})\mlc \ade |\vec{v}|\leq\mathfrak{s}|\vec{s}| \mathbb{D}^{\epsilon} (x,\vec{s},\vec{v})\bigr).$
At the beginning, we wait till the $\ade |\vec{v}|\leq\mathfrak{s}|\vec{s}| \mathbb{D}^{\epsilon} (x,\vec{s},\vec{v})$ subcomponent of it 
 is resolved and thus (\ref{m2ex}) is brought down to  
\begin{equation}\label{jh}\begin{array}{l}
 \cle x\bigl({E}^\circ(x,\vec{s})\mlc |\vec{c}|\leq\mathfrak{s}|\vec{s}|\mlc \mathbb{D}^{\epsilon}(x,\vec{s},\vec{c})\bigr)
\mli \\ \tilde{E}^{\bullet}   (\vec{s})\add \gneg\mathbb{W}\add 
\ade y \tilde{H}_{1}^{\circ}(\vec{s},y)\add\ldots\add \ade y \tilde{H}_{n}^{\circ}(\vec{s},y)
\end{array}
\end{equation}
for some  tuple $\vec{c}=c_1,\ldots,c_{2\mathfrak{y}+3}$ of constants. 
From now on, we shall assume that the antecedent of (\ref{jh}) is true, or else we win no matter what. Let then $x_0$ be 
the obviously unique number that, in the role of $x$, makes the antecedent of  (\ref{jh})  true. That is,  we have 
\begin{equation}\label{ks}
{E}^\circ(x_0,\vec{s})\mlc |\vec{c}|\leq\mathfrak{s}|\vec{s}|\mlc \mathbb{D}^{\epsilon}(x_0,\vec{s},\vec{c}).
\end{equation}
In order to win  (\ref{jh}), it is sufficient to figure out how to win  its consequent, so, from now on, our target will be 
\begin{equation}\label{csg}
\tilde{E}^{\bullet}   (\vec{s})\add \gneg\mathbb{W}\add 
\ade y \tilde{H}_{1}^{\circ}(\vec{s},y)\add\ldots\add \ade y \tilde{H}_{n}^{\circ}(\vec{s},y).
\end{equation}

For some ($\ade$) constant $a$, Lemma \ref{l3} provides the resource $a=\mathfrak{t}|\vec{s}|\mlc  \mathbb{Q}(\vec{s},a)$, which 
   disabbreviates as

\[a=\mathfrak{t}|\vec{s}|\mlc  \cla x\Bigl[\mathbb{D}_{\mbox{\tiny $\sqcup$}}(x,\vec{s})\mli  \gneg \mathbb{N}(x,a)\add\Bigl(\mathbb{N}(x,a)\mlc \cle t \bigl(\mathbb{U}_{\mbox{\tiny $\sqcup$}}^{\vec{s}}(x,t,a)\mlc  \mathbb{D}_{\mbox{\tiny $\sqcup$}}(t,\vec{s})\bigr)\Bigr)\Bigr].\]
We use $\vec{c}$ to resolve the $\mathbb{D}_{\mbox{\tiny $\sqcup$}}(x,\vec{s})$ component of the above game, bringing the latter it down to
\begin{equation}\label{exc}
\begin{array}{l}
a=\mathfrak{t}|\vec{s}|\mlc  \cla x\Bigl[|\vec{c}|\leq\mathfrak{s}|\vec{s}|\mlc \mathbb{D}(x,\vec{s},\vec{c})\mli  \\ \gneg \mathbb{N}(x,a)\add\Bigl(\mathbb{N}(x,a)\mlc \cle t \bigl(\mathbb{U}_{\mbox{\tiny $\sqcup$}}^{\vec{s}}(x,t,a)\mlc  \mathbb{D}_{\mbox{\tiny $\sqcup$}}(t,\vec{s})\bigr)\Bigr)\Bigr].
\end{array}
\end{equation}
Plugging the earlier fixed $x_0$ for $x$ in (\ref{exc})  and observing that $|\vec{c}|\leq\mathfrak{s}|\vec{s}|\mlc \mathbb{D}(x_0,\vec{s},\vec{c})$
is true by (\ref{ks}), it is clear that having the resource (\ref{exc}), in fact, implies having  
\begin{equation}\label{re}
a=\mathfrak{t}|\vec{s}|\mlc  \Bigl(\gneg \mathbb{N}(x_0,a)\add\bigl(\mathbb{N}(x_0,a)\mlc \mathbb{U}_{\mbox{\tiny $\sqcup$}}^{\vec{s}}(x_0,t_0,a)\mlc  \mathbb{D}_{\mbox{\tiny $\sqcup$}}(t_0,\vec{s})\bigr)\Bigr)
\end{equation}
for some ($\cle$) $t_0$. We wait till the displayed $\add$-disjunction of (\ref{re}) is resolved by the provider. 

Suppose the left $\add$-disjunct  $\gneg\mathbb{N}(x_0,a)$ is chosen in (\ref{re}). Then $\mathbb{N}(x_0,a)$ has to be false. This means that $x_0$ has a corrupt unadulterated successor. At the same time, from  the ${E}^\circ(x_0,\vec{s})$ conjunct of (\ref{ks}), we know that $x_0$ is a reachable semiuncorrupt configuration. All this, together with (\ref{aprili20}), (\ref{pp1}) and (\ref{pp2}), as can be seen with some analysis, implies that   $\mathbb{W}$ is false.\footnote{Namely, $\mathbb{W}$ is false because $\mathcal X$ ``does something wrong'' after reaching the configuration $x_0$.} So, we win (\ref{csg}) by choosing its $\add$-disjunct 
$\gneg\mathbb{W}$. 

 Now suppose the right $\add$-disjunct is chosen in (\ref{re}), bringing the game down to 
\[a=\mathfrak{t}|\vec{s}|\mlc\mathbb{N}(x_0,a)\mlc \mathbb{U}_{\mbox{\tiny $\sqcup$}}^{\vec{s}}(x_0,t_0,a)\mlc  \mathbb{D}_{\mbox{\tiny $\sqcup$}}(t_0,\vec{s}).\] 
We wait till  the above  is further brought down to 
\begin{equation}\label{exgg}
a=\mathfrak{t}|\vec{s}|\mlc \mathbb{N}(x_0,a)\mlc
|b|\leq\mathfrak{s}|\vec{s}|\mlc \mathbb{U}(x_0,t_0,a,b) 
\mlc |\vec{d}|\leq\mathfrak{s}|\vec{s}|\mlc  \mathbb{D}(t_0,\vec{s},\vec{d}) 
\end{equation}
for some constant $b$ and some 
tuple $\vec{d}$ of constants. Take a note of the fact that, by the  $\mathbb{U}(x_0,t_0,a,b)$ conjunct of (\ref{exgg}), $t_0$ is a $b$th unadulterated successor of $x_0$. Using Fact \ref{tri}, we figure out whether $b=a$ or $b\not=a$.

First, assume $b=a$, so that, in fact, (\ref{exgg}) is  
\begin{equation}\label{exggw}
a=\mathfrak{t}|\vec{s}|\mlc\mathbb{N}(x_0,a)\mlc 
|a|\leq\mathfrak{s}|\vec{s}|\mlc \mathbb{U}(x_0,t_0,a,a) 
\mlc |\vec{d}|\leq\mathfrak{s}|\vec{s}|\mlc  \mathbb{D}(t_0,\vec{s},\vec{d}) .
\end{equation}
In this case we choose $\tilde{E}^{\bullet}   (\vec{s})$  
in (\ref{csg}) and then further bring the latter down to   
\begin{equation}\label{exl}
 \cle x\bigl(E^{\bullet}   (x,\vec{s})\mlc |\vec{c}|\leq \mathfrak{s}|\vec{s}|\mlc \mathbb{D}^\epsilon(x,\vec{s},\vec{c})\bigr).
\end{equation}
According to  (\ref{ks}),  ${E}^\circ(x_0,\vec{s})$  is true.   From the first and the fourth conjuncts of (\ref{exggw}), we also know that the run tape content of $e$ persists for ``sufficiently long'', namely, for at least $\mathfrak{t}|\vec{s}|$ steps. Therefore, ${E}^\circ(x_0,\vec{s})$ implies 
${E}^{\bullet} (x_0,\vec{s})$. For this reason,  (\ref{exl}) is true, as it follows from (\ref{ks}). We thus win.

Now, for the rest of this proof, assume $b\not=a$. Note that then, by the $\mathbb{U}(x_0,t_0,a,b)$ conjunct of (\ref{exgg}), 
$b<a$ and, in the scenario that we are dealing with, $\mathcal X$ made a move on the $(b+1)$st step after reaching configuration $x_0$, i.e., immediately ($1$ step) after reaching configuration $t_0$. Let us agree to refer to that move as $\sigma$, and use $t_1$ to refer to the configuration that describes the  $(b+1)$st step after reaching configuration $x_0$ --- that is, the step on which the move $\sigma$ was made.     In view of \cite{cl12}'s stipulation that an HPM never adds anything to its buffer when transitioning to a move state, we find that $\sigma$ is exactly the move found in configuration $t_0$'s buffer. 

Applying Comprehension to the formula (\ref{ext}) of Lemma \ref{l4} and taking $\vec{c}$ in the role of $\vec{v}$, we get
\[\begin{array}{l}
\ade |w|\leq \mathfrak{a}|\vec{s}|\cla r< \mathfrak{a}|\vec{s}|\Bigl(\bit(r,w)\leftrightarrow \\
\cle x \cle t\cle y \bigl(\mathbb{N}(x,\mathfrak{t}|\vec{s}|)\mlc \mathbb{D}^\epsilon(x,\vec{s},\vec{c})\mlc 
  \mathbb{U}_{\mbox{\tiny $\exists$}}^{\vec{s}}(x,t) \mlc \mathbb{F}(t,y)\mlc \bit(r,y)\bigr)\Bigr).
\end{array}\]
The provider of the above resource will have to choose a value $w_0$ for $w$ and bring the game down to
\begin{equation}\label{stage}
\begin{array}{l}
|w_0|\leq \mathfrak{a}|\vec{s}|\mlc \cla r< \mathfrak{a}|\vec{s}|\Bigl(\bit(r,w_0)\leftrightarrow \\
\cle x \cle t\cle y \bigl( \mathbb{N}(x,\mathfrak{t}|\vec{s}|)\mlc \mathbb{D}^\epsilon(x,\vec{s},\vec{c})\mlc 
  \mathbb{U}_{\mbox{\tiny $\exists$}}^{\vec{s}}(x,t) \mlc \mathbb{F}(t,y)\mlc \bit(r,y)\bigr)\Bigr).
\end{array}
\end{equation}

From (\ref{ks}) we know that $\mathbb{D}^\epsilon(x_0,\vec{s},\vec{c})$ is true, and then from $\pa$ we know that $x_0$ is a unique number satisfying  $\mathbb{D}^\epsilon(x_0,\vec{s},\vec{c})$. Also remember from (\ref{exgg}) that $\mathfrak{t}|\vec{s}|=a$. For these reasons, the (para)formula 
\begin{equation}\label{qxw}
\cle x \cle t\cle y \bigl( \mathbb{N}(x,\mathfrak{t}|\vec{s}|)\mlc \mathbb{D}^\epsilon(x,\vec{s},\vec{c})\mlc 
  \mathbb{U}_{\mbox{\tiny $\exists$}}^{\vec{s}}(x,t) \mlc \mathbb{F}(t,y)\mlc \bit(r,y)\bigr)
\end{equation}
can be equivalently re-written as 
\begin{equation}\label{qxe}
\cle t\cle y \bigl( \mathbb{N}(x_0,a)\mlc 
  \mathbb{U}_{\mbox{\tiny $\exists$}}^{\vec{s}}(x_0,t) \mlc \mathbb{F}(t,y)\mlc \bit(r,y)\bigr).
\end{equation}
From the $a=\mathfrak{t}|\vec{s}|$ and $ \mathbb{U}(x_0,t_0,a,b)$ conjuncts of (\ref{exgg}), by $\pa$, we know that $t_0$ is a unique number satisfying $\mathbb{U}_{\mbox{\tiny $\exists$}}^{\vec{s}}(x_0,t_0)$. From (\ref{exgg}) we also know that 
$ \mathbb{N}(x_0,a)$ is true. 
And, from $\pa$, we also know that there is ($\cle$) a unique number --- let us denote it by $y_0$ --- satisfying  $\mathbb{F}(t_0,y_0)$. Consequently, (\ref{qxe}) can be further re-written as 
$\bit(r,y_0)$. So, (\ref{qxw}) is equivalent to $\bit(r,y_0)$, which allows us to re-write (\ref{stage}) as
\begin{equation}\label{stage1}
\begin{array}{c}
|w_0|\leq \mathfrak{a}|\vec{s}|\mlc \cla r< \mathfrak{a}|\vec{s}|\bigl(\bit(r,w_0)\leftrightarrow \bit(r,y_0)\bigr).
\end{array}
\end{equation}

With the $ \mathbb{N}(x_0,a)$ conjunct of (\ref{exgg}) in mind, by $\pa$ we can see that $t_0$, being a $b$th unadulterated successor of $x_0$ with $b<a$, is uncorrupt. If so, remembering that $y_0$ is the numer of the move $\sigma$ found in $t_0$'s buffer, by condition 7 of Definition \ref{dcor}, we have $|y_0|\leq \mathfrak{a}|\vec{s}|$.  This fact, together with (\ref{stage1}), obviously implies that $y_0$ and $w_0$ are simply the same. Thus, $w_0$ is the numer of $\sigma$. 

In view of the truth of the $\mathbb{D}(t_0,\vec{s},\vec{d})$ conjunct of (\ref{exgg}), $d_{2\mathfrak{y}+3}$ contains information on the  header of $\sigma$. From this header, we can determine the number $i\in\{1,\ldots,n\}$ such that the move $\sigma$ by $\mathcal X$ in position $E(\vec{s})$ yields $H_i(\vec{s},w_0)$. Fix such an $i$. Observe that the following is true:
\begin{equation}\label{dges}H_{i}^{\circ}(t_1,\vec{s},w_0).\end{equation}

From $d_{2\mathfrak{y}+3}$  we determine the state of $t_0$. Lemma \ref{sclemma} further allows us to determine the scene of $t_0$ as well. These two pieces of information, in turn, determine the titular number  of $t_0$'s successor configuration $t_1$. Let $e$ be that titular number. Let $\vec{d}^e$ be the same as $\vec{d}$, only with $e$ instead of $d_{2\mathfrak{y}+3}$. 

From the $E^{\circ}(x_0,\vec{s})$ conjunct of  (\ref{ks}) we know that $x_0$ is uncorrupt and hence semiuncorrupt. This implies that $t_1$ is also 
 semiuncorrupt, because  $x_0$ has evolved to $t_1$ in the scenario where Environment made no moves. For this reason, the titular number $e$ of $t_0$ is smaller than $\mathfrak{m}$. From $E^{\circ}(x_0,\vec{s})$ and $x_0$'s being uncorrupt, in view of clause 3 of Definition \ref{dcor}, we also know that $\mathfrak{m}\leq\mathfrak{s}|\vec{s}|$. Consequently, 
$e\leq\mathfrak{s}|\vec{s}|$. This fact, together with the $|\vec{d}|\leq\mathfrak{s}|\vec{s}|$ conjunct of (\ref{exgg}), implies that   
\begin{equation}\label{dges1}
|\vec{d}^e|\leq\mathfrak{s}|\vec{s}|.
\end{equation}
 Next, from (\ref{exgg}) again, we know that $\mathbb{D}(t_0,\vec{s},\vec{d})$ is true. This fact, in view of our earlier assumption that $\mathcal X$ never moves its scanning heads and never makes any changes on its work tapes on a transition to a move state, obviously implies that the following is also true:
\begin{equation}\label{dges2}
\mathbb{D}(t_1,\vec{s},w_0,\vec{d}^e).\end{equation}

At this point, at last, we are  ready to describe our strategy for (\ref{csg}). First, relying on Fact \ref{tri} several times, we figure out whether  $|\vec{d}^e|\leq\mathfrak{s}|\vec{s},w_0|$. If not, then, in view of (\ref{dges1}), $\mathfrak{s}$ is not monotone and hence $\mathbb{W}$ is false. In this case we select the $\gneg\mathbb{W}$ disjunct of (\ref{csg}) and celebrate victory. Now suppose $|\vec{d}^e|\leq\mathfrak{s}|\vec{s},w_0|$.
In this case we select the $\ade y \tilde{H}_{i}^{\circ}(\vec{s},y)$ disjunct of (\ref{csg}), then bring the resulting game down to $\tilde{H}_{i}^{\circ}(\vec{s},w_0)$, i.e., to 
 $\cle x\bigl(H_{i}^{\circ}(x,\vec{s},w_0)\mlc \ade |\vec{v}|\leq\mathfrak{s}|\vec{s},w_0|\mathbb{D}(x,\vec{s},w_0,\vec{v})\bigr)$, which we then further bring down to  
\[\cle x\bigl(H_{i}^{\circ}(x,\vec{s},w_0)\mlc |\vec{d}^e|\leq\mathfrak{s}|\vec{s},w_0|\mlc \mathbb{D}(x,\vec{s},w_0,\vec{d}^e)\bigr).\] The latter is true in view  (\ref{dges}), (\ref{dges2}) and our assumption $|\vec{d}^e|\leq\mathfrak{s}|\vec{s},w_0|$, so we win. 
\end{proof}

\subsection{Main lemma}\label{ssss}
 
\begin{lemma}\label{m2c}
Assume $E(\vec{s})$ is a formula all of whose free variables are among $\vec{s}$. 
Then $\thr$  proves $\tilde{E}^{\circ}  (\vec{s})  \mli \overline{E(\vec{s})}$.  
\end{lemma}

\begin{proof} We prove this lemma by (meta)induction on the complexity of $E(\vec{s})$. By the induction hypothesis, for any $(\oo,y)$- or $(\pp,y)$-development $H_i(\vec{s},y)$ of $E(\vec{s})$ (if there are any), $\thr$  proves 
\begin{equation}\label{m2d}
\tilde{H}_{i}^{\circ}  (\vec{s},y)  \mli \overline{H_i(\vec{s},y)},
\end{equation} 
which is the same as 
\begin{equation}\label{h5}
\cle x\bigl({H}_{i}^{\circ}  (x,\vec{s},y)\mlc \mathbb{D}^{\epsilon}_{\mbox{\tiny $\sqcup$}} (x,\vec{s})\bigr)  \mli \overline{H_i(\vec{s},y)}.
\end{equation}

Argue in $\thr$  to justify $\tilde{E}^{\circ}  (\vec{s})  \mli \overline{E(\vec{s})}$, which disabbreviates as 
\begin{equation}\label{h1}
 \cle x\bigl({E}^{\circ}  (x,\vec{s})\mlc  \mathbb{D}^{\epsilon}_{\mbox{\tiny $\sqcup$}} (x,\vec{s})\bigr)  \mli \overline{E(\vec{s})}.
\end{equation}
To win (\ref{h1}), we wait till Environment brings it down to  
\begin{equation}\label{h2}
\cle x\bigl({E}^{\circ}  (x,\vec{a})\mlc  |\vec{c}|\leq\mathfrak{s}|\vec{s}|\mlc \mathbb{D}^\epsilon(x,\vec{a},\vec{c}) \bigr)  \mli \overline{E(\vec{a})}
\end{equation}
for some  tuples $\vec{a}=a_1,\ldots,a_n$ and $\vec{c}=c_1,\ldots,c_{2\mathfrak{y}+3}$ of constants.\footnote{Here, unlike the earlier followed practice, for safety, we are reluctant to use the names  $\vec{s},\vec{v}$ for those constants.} 
  Assume the antecedent of (\ref{h2})  is true (if not, we win). Our goal is to show how to win the consequent $\overline{E(\vec{a})}$.  Let $b$ be the (obviously unique) constant satisfying the  antecedent of (\ref{h2}) in the role of $x$.

Let ${H}_{1}^{\circ}(\vec{s},y),\ldots, {H}_{n}^{\circ}(\vec{s},y)$ be all of the $(\top,y)$-developments of $E(\vec{s})$.  By  Lemma 
\ref{l5},  the following resource is at our disposal: 
\begin{equation}\label{h3} 
\begin{array}{l}
\cle x\bigl({E}^{\circ}  (x,\vec{s})\mlc  \mathbb{D}^{\epsilon}_{\mbox{\tiny $\sqcup$}}(x,\vec{s})\bigr)  \mli \\ \tilde{E}^{\bullet}  (\vec{s})\add \gneg\mathbb{W}\add \ade y \tilde{H}_{1}^{\circ}(\vec{s},y)\add\ldots\add \ade y \tilde{H}_{n}^{\circ}(\vec{s},y). 
\end{array} \end{equation}
We bring (\ref{h3}) down to 
\begin{equation}\label{h4}
\begin{array}{l} 
\cle x\bigl({E}^{\circ}  (x,\vec{a})\mlc |\vec{c}|\leq\mathfrak{s}|\vec{a}|\mlc  \mathbb{D}^{\epsilon}(x,\vec{a},\vec{c})\bigr) \mli \\ \tilde{E}^{\bullet}  (\vec{a})\add\gneg\mathbb{W}\add   \ade y \tilde{H}_{1}^{\circ}(\vec{a},y)\add\ldots\add \ade y \tilde{H}_{n}^{\circ}(\vec{a},y). 
\end{array}
\end{equation}
 Since the  antecedent of (\ref{h4}) is identical to the  antecedent of (\ref{h2}) and hence is true,   the provider of (\ref{h4}) will have to choose one of the $\add$-disjuncts in the consequent  
\begin{equation} \label{m2f}   \tilde{E}^{\bullet}  (\vec{a})\add \gneg\mathbb{W}\add  \ade y \tilde{H}_{1}^{\circ}(\vec{a},y)\add\ldots\add \ade y \tilde{H}_{n}^{\circ}(\vec{a},y). \end{equation}

\begin{description}
\item[Case 1] $\gneg\mathbb{W}$ is chosen in (\ref{m2f}). $\mathbb{W}$ has to be false, or else the provider loses. By Lemma \ref{jan4d}, the resource $\mathbb{W}\mld \cla\overline{E(\vec{s})}$ is at our disposal, which, in view of $\mathbb{W}$'s being false, simply means having $\cla\overline{E(\vec{s})}$. But the strategy that wins the latter, of course, also (``even more so'') wins our target $\overline{E(\vec{a})}$.

\item[Case 2] One of $\ade y \tilde{H}_{i}^{\circ}(\vec{a},y)$ is chosen in (\ref{m2f}). This should be followed by a further choice of some constant  $d$ for  $y$, yielding $\tilde{H}_{i}^{\circ}(\vec{a},d)$. Plugging $\vec{a}$  and $d$ for $\vec{s}$  and $y$ in (\ref{m2d}), we get $\tilde{H}_{i}^{\circ}  (\vec{a},d)  \mli\overline{ H_i(\vec{a},d)}$. Thus,  the two resources $\tilde{H}_{i}^{\circ}(\vec{a},d)$ and $\tilde{H}_{i}^{\circ}  ( \vec{a},d)  \mli \overline{H_i( \vec{a},d)}$ are at our disposal. Hence so is $\overline{H_i( \vec{a},d)}$.  
 But, remembering that the formula $H_i(\vec{s},y)$ is a $(\pp,y)$-development of the formula $E(\vec{s})$, we can now win $\overline{E(\vec{a})}$ by making a move $\alpha$  that brings ($E(\vec{a})$ down to  $H_i( \vec{a},d)$ and hence) $\overline{E(\vec{a})}$ down to  $\overline{H_i( \vec{a},d)}$, which we already know how to win.   For example, imagine $E(\vec{s})$ is $Y(\vec{s})\mli Z(\vec{s})\add T(\vec{s})$ and $H_i(\vec{s},y)$ is $Y(\vec{s})\mli Z(\vec{s})$.  Then the above move $\alpha$ will be ``$1.0$''. It indeed brings 
($Y(\vec{a})\mli Z(\vec{a})\add T(\vec{a})$ down to $Y(\vec{a})\mli Z(\vec{a})$ and hence) $\overline{Y(\vec{a})\mli Z(\vec{a})\add T(\vec{a})}$ down to $\overline{Y(\vec{a})\mli Z(\vec{a})}$. As another example, imagine $E(\vec{s})$ is $Y(\vec{s})\mli \ade w 
Z(\vec{s},w)$ and $H_i( \vec{s},y)$ is $Y(\vec{s})\mli Z( \vec{s},y)$.
Then the above move $\alpha$ will be ``$1.\#d$''. It indeed brings
$\overline{Y(\vec{a})\mli \ade w Z( \vec{a},w)}$ down to
$\overline{Y(\vec{a})\mli Z( \vec{a},d)}$.
\item[Case 3]  $\tilde{E}^{\bullet}   (\vec{a})$, i.e., $\cle x\bigl({E}^{\bullet}(x,\vec{a})\mlc \mathbb{D}^{\epsilon}_{\mbox{\tiny $\sqcup$}} (x,\vec{a})\bigr)$, 
is chosen in (\ref{m2f}). It has to be true, or else the provider loses. For this reason, $\cle x E^{\bullet} (x,\vec{a})$ is also true.

\item[Subcase 3.1] The formula $\tilde{E}^{\bullet}   (\vec{s})$ is critical. Since $\cle x E^{\bullet} (x,\vec{a})$ is true, so is $\cle E^{\bullet} (z,\vec{s})$. By Lemma \ref{august20a}, we also have 
$\cle E^{\bullet} (z,\vec{s})  \mli \cla\overline{E(\vec{s})}$. So, we have a winning strategy for $\cla\overline{E(\vec{s})}$. Of course, the same strategy also wins $\overline{E(\vec{a })}$.

\item[Subcase 3.2] The formula $\tilde{E}^{\bullet}   (\vec{s})$ is not critical. From $\cle x E^{\bullet} (x,\vec{a})$ and Lemma \ref{august20b}, by LC, we find that the elementarization of $\overline{E(\vec{a})}$ is true.
 This obviously means that if Environment does not move in  $\overline{E(\vec{a})}$, we win the latter. So, assume Environment makes a move $\alpha$ in $\overline{E(\vec{a})}$. 
The move should be legal, or else we  win. Of course,  for one of the $(\oo,y)$-developments 
$H_i(\vec{s},y)$ of the formula $E(\vec{s})$ and some constant $d$,  $\alpha$ brings $ \overline{E(\vec{a})}$  down to $\overline{H_i(\vec{a},d)}$.  For example, if  $E(\vec{s})$ is $Y(\vec{s})\mli Z(\vec{s})\adc T(\vec{s})$, $\alpha$ could be the move ``$1.0$'', which brings $\overline{Y(\vec{a})\mli Z(\vec{a})\adc T(\vec{a})}$ down to $\overline{Y(\vec{a})\mli Z(\vec{a})}$; the formula $Y(\vec{s})\mli Z(\vec{s})$ is indeed a $(\oo,y)$-development of the formula $Y(\vec{s})\mli Z(\vec{s})\adc T(\vec{s})$. As another example, imagine $E(\vec{s})$ is $Y(\vec{s})\mli \ada w Z(\vec{s},w)$. Then the above move $\alpha$ could be ``$1.\#d$'',  which brings $\overline{Y(\vec{a})\mli \ada wZ( \vec{a},w)}$ down to $\overline{Y(\vec{a})\mli Z( \vec{a},d)}$; the formula $Y(\vec{s})\mli Z(  \vec{s},y)$ is indeed a $(\oo,y)$-development of the formula $Y(\vec{s})\mli \ada wZ( \vec{s},w)$. Fix the above formula $H_i(\vec{s},y)$ and constant $d$. Choosing   $\vec{a}$ and $d$ for   $\vec{s}$ and $y$ in the resource $\tilde{E}^{\circ}   ( \vec{s})  \mli   \tilde{H}_{i}^{\circ}(  \vec{s},y)$ provided 
by   Lemma 
\ref{l6}, we get the resource $\tilde{E}^{\circ} (\vec{a})  \mli  \tilde{H}_{i}^{\circ}(\vec{a},d)$. 
Since $\tilde{E}^{\bullet}   (\vec{a})$  
is chosen in (\ref{m2f}), we have a winning strategy for $\tilde{E}^{\bullet}   (\vec{a})$ and hence for the weaker 
$\tilde{E}^{\circ}   (\vec{a})$. This, together with   $\tilde{E}^{\circ} (\vec{a})  \mli \tilde{H}_{i}^{\circ}(\vec{a},d)$, by LC, yields  
$\tilde{H}_{i}^{\circ}(\vec{a},d)$. 
 By choosing   $\vec{a}$ and $d$ for    $\vec{s}$ and $y$  in (\ref{m2d}), we now get the resource 
$\overline{H_i(\vec{a},d)}$. That is, we have a strategy for the game $\overline{H_i(\vec{a},d)}$ to which $\overline{E(\vec{a})}$ has 
evolved after Environment's move $\alpha$. We switch to that strategy
and win.\qedhere
\end{description}
\end{proof}

\subsection{Conclusive steps}\label{sconc}
Now we are ready to claim the target result of this section.  Let $a$ be the (code of the) start configuration of $\mathcal X$
where the run tape is empty. Without loss of generality we may assume that the titular number of $a$ is $0$. Let $\vec{0}$ stand for a $(2\mathfrak{y}+3)$-tuple of $0$s.   Of course,  
$\pa$ proves $X^\circ(\hat{a})\mlc \mathbb{D}^\epsilon(\hat{a},\vec{0})$,\footnote{Whatever would normally appear as an additional $\vec{s}$ argument of $\mathbb{D}^\epsilon$ is empty in the present case.} and hence $\pa$ also proves 
$\cle x\bigl(X^\circ(x)\mlc \mathbb{D}^\epsilon(x,\vec{0})\bigr)$. Then, by LC, $\thr$ proves $\cle x\bigl(X^\circ(x)\mlc \ade|\vec{v}|\leq \mathfrak{s}(0)\mathbb{D}^\epsilon(x,\vec{v})\bigr)$, i.e., 
$\cle x\bigl(X^\circ(x)\mlc \mathbb{D}^{\epsilon}_{\mbox{\tiny $\sqcup$}}(x)\bigr)$, i.e.,  
$\tilde{X}^\circ$.
 By Lemma \ref{m2c},   $\thr$  also proves $\tilde{X}^{\circ}   \mli \overline{X}$. These two   imply the desired  $\overline{X}$ by LC, thus completing our proof of the extensional completeness of $\thr$.

\section{Intensional completeness}\label{s19e}

\subsection{The intensional completeness of \texorpdfstring{$\areleven^{\mathcal R}_{{\mathcal A}!}$}{CLA11AR}}

 Let us fix an arbitrary regular theory $\thr$ and an arbitrary sentence $X$ with an $\mathcal R$ tricomplexity solution. Proving the intensional completeness of $\areleven^{\mathcal R}_{{\mathcal A}!}$ --- i.e., the completeness part of clause 2 of \mbox{Theorem \ref{tt1}} ---  means showing  that $\areleven^{\mathcal R}_{{\mathcal A}!}$ proves (not only $\overline{X}$ but also) $X$.  This is what the present section is devoted to. Let $\mathcal X$, $(\mathfrak{a},\mathfrak{s},\mathfrak{t})$, $\mathbb{W}$ be as in Section \ref{s19}, and so be the meaning of the overline notation.

\begin{lemma}\label{feb15}
 $\thr\vdash \mathbb{W}\mli X$.
\end{lemma}

\begin{proof} First, by induction on the complexity of $E$, we  want to show that
\begin{equation}\label{m14a}
\mbox{\em For any formula $E$, $\thr\vdash \cla (\overline{E}\mlc \mathbb{W}\mli E)$.} 
\end{equation}
If  $E$ is a literal, then $\cla(\overline{E}\mlc \mathbb{W}\mli E)$ is nothing but 
$\cla\bigl((\mathbb{W}\mli E)\mlc  \mathbb{W}\mli E\bigr)$. Of course $\thr$ proves this elementary sentence, which happens to be classically valid.  
Next, suppose $E$ is $F_0\mlc F_1$. By the induction hypothesis, $\thr$ proves both $\cla (\overline{F_0}\mlc \mathbb{W}\mli F_0)$ and  $\cla (\overline{F_1}\mlc \mathbb{W}\mli F_1)$. These two, by LC, imply $\cla \bigl((\overline{F_0}\mlc\overline{F_1})\mlc  \mathbb{W}\mli F_0\mlc F_1\bigr)$. And the latter is nothing but the desired $\cla (\overline{E}\mlc \mathbb{W}\mli E)$. 
The remaining cases where $E$ is $F_0\mld F_1$, $F_0\adc F_1$, $F_0\add F_1$, $\ada xF(x)$, $\ade xF(x)$, $\cla xF(x)$ or $\cle xF(x)$ are handled in a similar way.
 (\ref{m14a}) is thus proven.

(\ref{m14a}) implies that $\thr$ proves $\overline{X}\mlc \mathbb{W}\mli X$.  As established in Section \ref{s19}, $\thr$ also proves $\overline{X}$. From these two, by LC, $\thr$ proves  $ \mathbb{W}\mli X$, as desired. 
\end{proof} 

As we remember from Section \ref{s19}, $\mathbb{W}$ is a true elementary sentence. As such, it is an element of 
${\mathcal A}!$ and is thus provable in $\areleven^{\mathcal R}_{{\mathcal A}!}$. By Lemma \ref{feb15},  
$\areleven^{\mathcal R}_{{\mathcal A}!}$  also proves both $\overline{X}$ and $\overline{X}\mlc \mathbb{W}\mli X$. Hence, by LC, $\areleven^{\mathcal R}_{{\mathcal A}!}\vdash X$.
This proves the completeness part of Theorem \ref{tt1}.   

\subsection{The intensional strength of \texorpdfstring{$\thr$}{CLA11AR}}
While $\areleven^{\mathcal R}_{{\mathcal A}!}$ is intensionally complete, $\thr$ generally is not. Namely, the G\"{o}del-Rosser incompleteness theorem  precludes $\thr$ from being intensionally complete as long as it is consistent and $\mathcal A$ is recursively enumerable. Furthermore, in view of Tarski's theorem on the undefinability of truth, it is not hard to see that 
$\thr$, if sound, cannot be intensionally complete even if the set $\mathcal A$ is just arithmetical, i.e., if the predicate ``$x$ is the code of some element of $\mathcal A$'' is expressible in the language of $\pa$.

Intensionally, even though incomplete,  $\thr$ is still very strong. The last sentence of Section \ref{sprc},  
in our present terms, reads: 
\begin{quote} {\em ... If a sentence $F$ is not
provable in $\thr$, it is unlikely that anyone would find an $\mathcal R$ tricomplexity algorithm solving the problem
expressed by $F$: either such an algorithm does not exist, or showing its correctness requires going beyond
ordinary combinatorial reasoning formalizable in $\pa$.}
\end{quote} 

\noindent To explain and justify this claim, assume $F$ has a $\bigl(\mathfrak{b}(x),\mathfrak{c}(x),\mathfrak{d}(x)\bigr)$ tricomplexity solution/algorithm $\mathcal F$, where $\bigl(\mathfrak{b}(x),\mathfrak{c}(x),\mathfrak{d}(x)\bigr)\in {\mathcal R}\amp\times {\mathcal R}\spa\times {\mathcal R}\tim$.
Let $\mathbb{V}$ be a sentence constructed from $F$, $\mathcal F$ and $(\mathfrak{b},\mathfrak{c},\mathfrak{d})$ in the same way as we earlier constructed $\mathbb{W}$ from $X$, $\mathcal X$ and $(\mathfrak{a},\mathfrak{s},\mathfrak{t})$. Note that $\mathbb{V}$ is a sentence asserting the ``correctness'' of $\mathcal F$. Now, assume a proof of $\mathcal F$'s correctness can be formalized in $\pa$, in the precise sense that $\pa\vdash \mathbb{V}$. According to Lemma \ref{feb15}, we also have  
$\thr\vdash \mathbb{V}\mli F$. Then, by LC, $\thr \vdash F$.

\makeatletter
\renewenvironment{theindex}
               {\section*{\indexname}%
                \@mkboth{\MakeUppercase\indexname}%
                        {\MakeUppercase\indexname}%
                \thispagestyle{plain}\parindent\z@
                \parskip\z@ \@plus .3\p@\relax
                \columnseprule \z@
                \columnsep 35\p@
                \let\item\@idxitem}
               {}
\makeatother
\twocolumn
\begin{theindex}
\item $\amp$ (as a subscript) \pageref{xtroica}
\item argument variable \pageref{xav}
\item arithmetical problem \pageref{x93} 
\item at least linear \pageref{x89}
\item at least logarithmic \pageref{x89}
\item at least polynomial \pageref{x90}
\item background (of a configuration) \pageref{xcbgsd}
\item basis of induction \pageref{x70}  
\item ``Big-O'' notation \pageref{xbon}
\item $\bit(y,x)$ \pageref{x50a}
\item Bit axiom \pageref{A3}
\item $\bitsum$ \pageref{xbitsum}
\item $\subcarry$ \pageref{xborr}
\item bound \pageref{x60}
\item boundclass \pageref{x62}
\item boundclass triple \pageref{x64}
\item bounded formula \pageref{x63}
\item bounded arithmetic \pageref{x13}
\item $\br_0(x,s)$, $\br_1(x,s)$ \pageref{xbrhg}
\item buffer-empty title \pageref{xbet}
\item $\mcarry$ \pageref{xsmacr}  
\item $\adcarry$ \pageref{xcar1}
\item clarithmetic \pageref{x3}
\item $\cltw$ \pageref{x17}
\item $\thr$ \pageref{x65}
\item {\bf CLA11} \pageref{x6} 
\item choice operators \pageref{x5}
\item cirquent calculus \pageref{0crq}
\item Comprehension ($\mathcal R$-Comprehension) \pageref{x69}
\item comprehension bound \pageref{xcbn}
\item comprehension formula \pageref{xcf} 
\item computability logic (CoL) \pageref{x1}
\item configuration \pageref{xconf} 
\item corrupt configuration \pageref{xncog}
\item critical formula \pageref{icritical}
\item $\mathfrak{d}$ \pageref{x5dd}
\item $\mathbb{D}$ \pageref{x41}  
\item $\mathbb{D}^\epsilon$ \pageref{x42}  
\item $\mathbb{D}_{\mbox{\tiny $\sqcup$}}$ \pageref{x43}  
\item $\mathbb{D}^{\epsilon}_{\mbox{\tiny $\sqcup$}}$ \pageref{x43}
\item elementary (formula, sentence) \pageref{x27}
\item elementary (game, problem) \pageref{x4}
\item elementary basis \pageref{x81} 
\item extended proof \pageref{x80}
\item extensional: strength \pageref{x9} completeness \pageref{x11}
\item $\mathbb{F}$ \pageref{x48}  
\item formula \pageref{x19}
\item header (of a move) \pageref{xhdr}  
\item HPM \pageref{xHPM}  
\item Induction (${\mathcal R}$-Induction) \pageref{x68}
\item induction bound \pageref{x75}
\item induction formula \pageref{x74}
\item inductive step \pageref{x71} 
\item instance of {\bf CLA11} \pageref{xinst}
\item intensional: strength \pageref{x8} completeness \pageref{x12}
\item $\mathfrak{j}$ \pageref{xscn}
\item $\mathfrak{k}$ \pageref{x5m} 
\item $\mathbb{L}$ \pageref{x18}
\item $\mathbb{L}$-sequent \pageref{xlsq}
\item LC \pageref{x67}
\item least significant bit \pageref{x46a}
\item left premise (of induction) \pageref{x72} 
\item linear closure \pageref{x82}
\item linearly closed \pageref{x84}
\item literal \pageref{xliteral} 
\item Log axiom \pageref{A2}
\item logical consequence (as a relation) \pageref{23}  
\item Logical Consequence (as a rule) \pageref{x66}
\item logically imply \pageref{xlcc}
\item logically valid \pageref{x24} 
\item $\mathfrak{m}$ \pageref{x5m}
\item  $\min $ \pageref{xmin} 
\item monotonicity \pageref{x61}
\item most significant bit \pageref{x47a}
\item $\mathbb{N}$ \pageref{xnnn} 
\item $\mbox{\em Nothing}$ \pageref{xnothing}
\item numer \pageref{xnmrte}
\item $\mbox{\em One}$ \pageref{xone}
\item paraformula \pageref{21}
\item paralegal move \pageref{xprlmv}
\item parasentence \pageref{22}
\item Peano arithmetic ($\pa$) \pageref{x2},\pageref{x26}
\item Peano axioms \pageref{x28},\pageref{A1}
\item politeral \pageref{ipoliteral}
\item polynomial closure \pageref{x83}
\item polynomially closed \pageref{x84}
\item provider (of a resource/game) \pageref{xprovider}
\item pterm (pseudoterm) \pageref{x35}
\item $\mathbb{Q}(\vec{s},z)$ \pageref{x47}  
\item reachable configuration \pageref{xreaco}
\item Reasonable $\mathcal R$-Comprehension \pageref{xrrco}
\item Reasonable $\mathcal R$-Induction \pageref{xrin}
\item regular boundclass triple \pageref{x86} 
\item regular theory \pageref{x91}
\item relevant branch \pageref{xrelv} 
\item relevant parasentence \pageref{xrpr}
\item representable \pageref{x10},\pageref{x95}
\item representation \pageref{x94}  
\item right premise (of induction) \pageref{x73}  
\item  scene (of a configuration) \pageref{xtitsc}
\item $\mbox{\it Scene}_i$ \pageref{xscn}
\item semiuncorrupt configuration \pageref{xscrff}
\item sentence \pageref{x20}
\item $\spa$ (as a subscript) \pageref{xtroica}
\item standard interpretation (model) \pageref{x29}
\item standard model of arithmetic \pageref{xsma}
\item Successor axiom \pageref{A15}
\item successor function  \pageref{x32},\pageref{x33s}
\item supplementary axioms \pageref{A4}
\item syntactic variation \pageref{x36}
\item $Th(N)$ \pageref{x103},\pageref{x101} 
\item $\tim$ (as a subscript) \pageref{xtroica} 
\item title (of a configuration) \pageref{xtit}  
\item $\mbox{\em Title}_i$ \pageref{x5m}  
\item titular component \pageref{xtitc}
\item titular number \pageref{xtnb} 
\item tricomplexity \pageref{x7},\pageref{xnw}
\item true \pageref{x30}
\item truth arithmetic \pageref{x103},\pageref{x101} 
\item $\mathbb{U}$ \pageref{x44}  
\item $\mathbb{U}_{\mbox{\tiny $\sqcup$}}^{\vec{s}}$ \pageref{x45} 
\item $\mathbb{U}_{\mbox{\tiny $\exists$}}^{\vec{s}}$ \pageref{x46}
\item unadulterated successor \pageref{xus}
\item unary numeral \pageref{x34}
\item uncorrupt configuration \pageref{xuncog}
\item value variable \pageref{xvv}
\item $\mathbb{W}$ \pageref{xwwwww}
\item $\mathbb{W}_1$ \pageref{xwlong}
\item $X$, $\mathcal X$ \pageref{xxxxx}  
\item $\mathfrak{y}$ \pageref{xyyy}
\item yield (of a configuration) \pageref{xyield}
\item $\mbox{\em Zero}$ \pageref{xzero}  
\item  \
\item  \ 
\item \ 
\item  \ 
\item \ 
\item  \ 
\item \ 
\item $\ada x\leq \mathfrak{p} $ (and similarly for the other quantifiers) \pageref{xabq}
\item $\ada |x|\leq \mathfrak{p} $ (and similarly for the other quantifiers) \pageref{xabq} 
\item ${\mathcal A}!$ \pageref{x100}
\item $\vdash$ \pageref{x98}
\item $\rep$ \pageref{x99}
\item $\preceq$ (as a relation between bounds/bound-\\classes)~\pageref{x85}
\item $\preceq$ (as a relation between tricomplexities) \pageref{xapr19}
\item $|x|$ \pageref{x40}
\item $\tau|\vec{x}|$ \pageref{xtx}
\item $(x)_y$ \pageref{x41a}
\item $x\successor$ \pageref{x32},\pageref{x33}
\item $\hat{n}$ \pageref{x31}
\item $F^\dagger$ \pageref{x29} 
\item  $\cla  F$, $\cle F$, $\ada F$, $\ade F$ \pageref{x25} 
\item $\adc,\add,\ada,\ade$ \pageref{x6a}
\item $\leftrightarrow$ \pageref{xleftrih}  
\item $E^\circ$ \pageref{xecr}
\item $E^\bullet$ \pageref{xecrb}
\item $\tilde{E}^\circ(\vec{s})$ \pageref{x49} 
\item $\tilde{E}^{\bullet}(\vec{s})$ \pageref{x50}
\item $\lfloor u/2\rfloor$ \pageref{xslja}
\item $\seq{\Phi}!F$ \pageref{xiprf} 
\item $\overline{E}$ (where $E$ is a formula) \pageref{ipver}
\item $S^\heartsuit$  \pageref{xmasti}
\item $S^\spadesuit$ \pageref{xmasti}

\end{theindex}
\vspace{-30 pt}
\end{document}